\newif\ifcomments
\newif\ifstoc
\newif\ifanonymous
\newif\ifacmart

\ifacmart
  \documentclass[sigconf,screen]{acmart}
\else
  \documentclass[11pt,pdfa,letterpaper]{article}
\fi

\ifacmart
  \newcommand{\cal}[1]{\mathcal{#1}}
\else
\usepackage[in]{fullpage}

\usepackage{iftex}
\ifPDFTeX
  \usepackage[utf8]{inputenc}
  \usepackage[noTeX]{mmap}
  \usepackage[T1]{fontenc}
\fi
\ifLuaTeX
  \usepackage{luatex85}
  \usepackage[noTeX]{mmap}
\fi

\usepackage{amssymb}
\usepackage{xcolor}
\definecolor{linkblue}{HTML}{001487}
\usepackage[colorlinks=true,allcolors=linkblue]{hyperref}

\usepackage{authblk}

\IfFileExists{./nick-flag}
{\usepackage[style=alphabetic,minalphanames=3,maxalphanames=4,maxnames=99,backref=true,backend=bibtex8]{biblatex}}
{\usepackage[style=alphabetic,minalphanames=3,maxalphanames=4,maxnames=99,backref=true]{biblatex}}
\renewcommand*{\labelalphaothers}{\textsuperscript{+}}
\renewcommand*{\multicitedelim}{\addcomma\space}
\DeclareFieldFormat{eprint:iacr}{Cryptology ePrint Archive: \href{https://ia.cr/#1}{\texttt{#1}}}
\DeclareFieldFormat{eprint:iacrarchive}{Cryptology ePrint Archive: \href{https://eprint.iacr.org/archive/#1}{\texttt{#1}}}
\AtEveryBibitem{%
 \clearlist{address}
 \clearfield{date}
 \clearfield{isbn}
 \clearfield{issn}
 \clearlist{location}
 \clearfield{month}
 \clearfield{series}

 \ifentrytype{book}{}{%
  \clearlist{publisher}
  \clearname{editor}
 }
}
\fi

\usepackage{amsfonts}
\usepackage{amsmath}
\usepackage{mathtools,amsthm} %
\usepackage{graphicx}
\usepackage{comment}
\usepackage{qtree}
\usepackage{tree-dvips}
\usepackage{float}
\usepackage[nameinlink,noabbrev]{cleveref}
\usepackage{braket}
\usepackage{mathrsfs}
\usepackage{tikz}
\usepackage{qcircuit}
\usepackage{xspace}
\usepackage{dsfont}
\usepackage[inline]{enumitem} %
\ifstoc
  \setlist[description]{noitemsep}
  \setlist[enumerate]{noitemsep}
  \setlist[itemize]{noitemsep}
\fi
\usepackage{longfbox}
\usepackage{csquotes}
\usepackage{mdframed}
\usepackage{microtype}

\usepackage{framed}

\ifcomments
  \newcommand{\hnote}[1]{{\color{magenta}\bf Henry: #1}} 
  \newcommand{\john}[1]{{\color{green}\bf John: #1}}
  \newcommand{\luowen}[1]{{\color{brown}\bf Luowen: #1}}
  \newcommand{\nick}[1]{{\color{orange}\bf Nick: #1}}
\else
  \newcommand{\hnote}[1]{\ignorespaces}
  \newcommand{\john}[1]{\ignorespaces}
  \newcommand{\luowen}[1]{\ignorespaces}
  \newcommand{\nick}[1]{\ignorespaces}
\fi

\theoremstyle{plain}
\ifacmart
\newtheorem{theorem}{Theorem}
\else
\newtheorem{theorem}{Theorem}[section]
\fi
\newtheorem*{theorem*}{Theorem}
\newtheorem{definition}[theorem]{Definition}
\newtheorem{lemma}[theorem]{Lemma}
\newtheorem{claim}[theorem]{Claim}
\Crefname{claim}{Claim}{Claims}
\newtheorem{algorithm}{Algorithm}
\newtheorem*{lemma*}{Lemma}
\newtheorem{fact}[theorem]{Fact}
\newtheorem{corollary}[theorem]{Corollary}
\newtheorem*{corollary*}{Corollary}
\newtheorem{proposition}[theorem]{Proposition}

\theoremstyle{definition}

\theoremstyle{remark}
\newtheorem{remark}[theorem]{Remark}

\crefname{step}{step}{steps}
\Crefname{step}{Step}{Steps}

\usepackage{stmaryrd} %
\makeatletter
\renewcommand{\paragraph}{%
  \@startsection{paragraph}{4}%
  {\z@}{2.25ex \@plus 1ex \@minus .2ex}{-1em}%
  {\normalfont\normalsize\bfseries}%
}
\makeatother
\newcommand{\class}[1]{\mathsf{#1}}
\newcommand{\poly}{\mathrm{poly}}

\newcommand{\negl}{\mathrm{negl}}

\newcommand*{\interact}{\mathord{\leftrightarrows}}

\newcommand{\Sim}{\mathsf{Sim}}

\newcommand{\eps}{\epsilon}

\DeclareMathOperator{\Span}{span}

\newcommand{\ketbra}[2]{\ket{#1}\!\!\bra{#2}}
\renewcommand{\cal}[1]{\mathcal{#1}}
\newcommand{\R}{\mathbb{R}}
\newcommand{\C}{\mathbb{C}}
\newcommand{\N}{\mathbb{N}}
\newcommand{\Z}{\mathbb{Z}}
\newcommand{\E}{\mathop{\mathbb{E}}}
\newcommand{\Tr}{\mathrm{Tr}}

\newcommand{\reg}[1]{\mathsf{#1}}

\newcommand{\Id}{\id}

\newcommand{\td}{\mathrm{td}}

\newcommand{\fidelity}{\mathrm{F}}

\newcommand{\wt}[1]{\widetilde{#1}}

\newcommand{\setft}[1]{\textnormal{#1}}
\newcommand{\id}{\setft{id}}

\newcommand{\bits}{\ensuremath{\{0, 1\}}}

\usepackage{mleftright}
\newcommand{\paren}[1]{\left(#1\right)}
\newcommand{\mparen}[1]{\mleft(#1\mright)}
\newcommand{\mbracket}[1]{\mleft[#1\mright]}
\newcommand{\abs}[1]{\left|#1\right|}

\newcommand{\ot}{\ensuremath{\otimes}}

\newcommand{\norm}[1]{\left\lVert#1\right\rVert}

\newcommand{\opnorm}[1]{\norm{#1}_{\infty}}

\DeclareMathOperator{\img}{img}
\DeclareMathOperator{\sgn}{sgn}

\let\1\relax
\newcommand{\1}{\mathds{1}}

\newcommand{\states}{\setft{S}}

\newcommand{\proj}[1]{\ket{#1}\!\!\bra{#1}}

\newcommand{\HProj}[2]{\Pi^{#1}_{#2}}

\newcommand{\aux}{\mathsf{aux}}
\newcommand{\init}{\mathsf{init}}

\newcommand{\APMeas}{\mathcal{M}}

\newcommand{\RegH}{\mathsf{H}}
\newcommand{\RegW}{\mathsf{W}}
\newcommand{\MeasUnitary}[1]{U_{\APMeas_{#1}}}
\newcommand{\NumReps}{\mathsf{NumReps}}
\newcommand{\Bin}{\mathsf{Bin}}
\newcommand{\MeasJor}{\mathcal{P}_\mathsf{Jor}}
\newcommand{\JorProj}[1]{\Pi^{\mathsf{Jor}}_{#1}}

\newcommand{\APMeasProjExt}[2]{\tilde{\Pi}_{#1}^{#2}}
\newcommand{\EffJor}{\mathsf{EffJor}}
\newcommand{\proto}{\pi}
\newcommand{\buresdist}{\mathrm{d_{Bures}}}

\newcommand{\OUT}{\mathsf{OUT}}
\newcommand{\ACCEPT}{\mathsf{ACCEPT}}
\newcommand{\REVEAL}{\mathsf{REVEAL}}
\newcommand{\tg}{\mathsf{tg}}
\newcommand{\mim}{\mathsf{mim}}
\newcommand{\sch}{\mathsf{schedule}}
\newcommand{\KeyGen}{\mathsf{KeyGen}}
\newcommand{\Mint}{\mathsf{Mint}}
\newcommand{\Ver}{\mathsf{Ver}}

\usepackage{soul}
\usepackage{xcolor}
\usepackage{xparse}
\makeatletter
\ExplSyntaxOn
\cs_new:Npn \white_text:n #1
{
  \fp_set:Nn \l_tmpa_fp {#1 * .01}
  \llap{\textcolor{white}{\the\SOUL@syllable}\hspace{\fp_to_decimal:N \l_tmpa_fp em}}
  \llap{\textcolor{white}{\the\SOUL@syllable}\hspace{-\fp_to_decimal:N \l_tmpa_fp em}}
}
\NewDocumentCommand{\whiten}{ m }
{
  \int_step_function:nnnN {1}{1}{#1} \white_text:n
}
\ExplSyntaxOff

\NewDocumentCommand{ \varul }{ D<>{5} O{0.2ex} O{0.1ex} +m } {%
  \begingroup
  \setul{#2}{#3}%
  \def\SOUL@uleverysyllable{%
    \setbox0=\hbox{\the\SOUL@syllable}%
    \ifdim\dp0>\z@
    \SOUL@ulunderline{\phantom{\the\SOUL@syllable}}%
    \whiten{#1}%
    \llap{%
      \the\SOUL@syllable
      \SOUL@setkern\SOUL@charkern
    }%
    \else
    \SOUL@ulunderline{%
      \the\SOUL@syllable
      \SOUL@setkern\SOUL@charkern
    }%
    \fi}%
  \ul{#4}%
  \endgroup
}
\makeatother

\ifacmart
\setcopyright{rightsretained}
\acmDOI{10.1145/3618260.3649603}
\acmYear{2024}
\copyrightyear{2024}
\acmSubmissionID{stoc24main-p10-p}
\acmISBN{979-8-4007-0383-6/24/06}
\acmConference[STOC '24]{Proceedings of the 56th Annual ACM Symposium on Theory of Computing}{June 24--28, 2024}{Vancouver, BC, Canada}
\acmBooktitle{Proceedings of the 56th Annual ACM Symposium on Theory of Computing (STOC '24), June 24--28, 2024, Vancouver, BC, Canada}
\received{12-NOV-2023}
\received[accepted]{2024-02-11}
\else
\title{An efficient quantum parallel repetition theorem and applications}

\date{}

\author[1]{John Bostanci}
\author[2]{Luowen Qian}
\author[3]{Nicholas Spooner}
\author[1]{Henry Yuen}

\affil[1]{Columbia University}
\affil[2]{Boston University}
\affil[3]{University of Warwick \& NYU}
\fi

\begin{document}
\ifacmart
\else
\maketitle
\fi

\begin{abstract}
    We prove a tight parallel repetition theorem for $3$-message computationally-secure quantum interactive protocols between an efficient challenger and an efficient adversary. We also prove under plausible assumptions that the security of $4$-message computationally secure protocols does not generally decrease under parallel repetition. These mirror the classical results of Bellare, Impagliazzo, and Naor~\cite{bellare1997does}. Finally, we prove that all quantum argument systems can be generically compiled to an equivalent $3$-message argument system, mirroring the transformation for quantum proof systems \cite{kitaev2000parallelization,kempe2007using}.

    As immediate applications, we show how to derive hardness amplification theorems for quantum bit commitment schemes (answering a question of Yan~\cite{yan2023general}), EFI pairs (answering a question of Brakerski, Canetti, and Qian~\cite{brakerski2022computational}), public-key quantum money schemes (answering a question of Aaronson and Christiano~\cite{aaronson2012quantum}), and quantum zero-knowledge argument systems.
    We also derive an XOR lemma \cite{yao82theory} for quantum predicates as a corollary.
\end{abstract}

\ifacmart
\else
    \ifstoc\newpage\fi
    \tableofcontents
    \ifstoc
      \newpage
      \pagenumbering{arabic}
    \fi
\fi

\section{Introduction}
\label{sec:intro}

In this work we study one of the most fundamental questions in theoretical cryptography: can we transform a ``weakly'' secure construction of a primitive into one that is ``truly'' secure? A common strategy for such a transformation is \emph{parallel repetition}: if the adversary's success probability against the original construction is bounded away from $1$, then the adversary's success probability against the repeated construction should tend to zero with the number of repetitions.
In classical cryptography this question is well-studied, beginning with the seminal work of Yao~\cite{yao82theory,levin1987one,goldreich2011yao} and leading to a long sequence of works \cite{bellare1997does,canetti2005hardness,haitner2009parallel,haastad2010efficient,chung2010parallel,pass2012parallel,chung2015tight,berman2021tight}.
Hardness amplification is also an essential tool for bootstrapping circuit lower bounds (see \cite{SV08} and the references therein).

Our focus in this work is on hardness amplification for \emph{quantum} cryptographic primitives; in particular we focus on the following class of quantum interactive protocols between an \emph{efficient challenger} (specified as part of the protocol) and an \emph{efficient adversary} indexed by a security parameter $\lambda$:
\begin{itemize}
    \item \emph{$3$-message}: The adversary sends the first message, the challenger the second, and the adversary the third. After the protocol ends, the challenger decides to accept or reject. All messages may be quantum.
    \item \emph{Weakly computationally secure}: No efficient ($\poly(\lambda)$-size) adversary can cause the challenger to accept with probability greater than (say) $1 - \frac1{\poly(\lambda)}$.
\end{itemize}

The security of many quantum cryptographic primitives --- including quantum (non-interactive) commitments, quantum money and 3-message quantum arguments --- can be naturally formulated in terms of the a quantum 3-message protocol associated with the primitive like above.
This quantum protocol is often called a \emph{security game}. %

Similarly to the classical setting one would like a generic method for amplifying the security of quantum cryptographic primitives. A natural approach is to repeat the construction in parallel; the security of the repeated construction usually corresponds to the parallel-repeated security game. Ideally, one would want the adversary's maximum success probability in a repeated security game to decrease exponentially with the number of repetitions. 

Our first result is a \emph{tight} parallel repetition theorem for all 3-message quantum protocols.
\begin{theorem}[3-message efficient parallel repetition, informal]
\label{thm:main-informal}
	Let $\proto$ be a $3$-message $\gamma$-computationally secure quantum protocol. Then the $k$-fold parallel repetition $\proto^{\otimes k}$ is $(\gamma^k + \negl(\lambda))$-computationally secure.
\end{theorem}

\ifstoc\else
\noindent We prove \Cref{thm:main-informal} by identifying the key high level approach used in proving both the classical Yao's XOR lemma \cite{yao82theory,levin1987one,goldreich2011yao} and classical tight $3$-message parallel repetition theorem of Canetti, Halevi, and Steiner~\cite{canetti2005hardness}, and then instantiating this high level approach by designing quantum components that work with an arbitrary quantum adversary.
As one would expect, handling quantum protocols and adversaries is much more challenging than classical: (1) the classical reduction involves cloning of the adversary's internal state during protocol execution, which may be computationally infeasible or even information-theoretically impossible (due to entanglement with the challenger); and %
(2) the classical analysis relies on conditional distributions which breaks down in the quantum setting due to non-commutativity.
To resolve these challenges, we combine techniques from recent works on quantum rewinding \cite{chiesa2022post} and quantum algorithmic techniques such as the quantum singular value transform \cite{gilyen2019quantum}, as well as additional new ideas to make them compatible with our setting.
We explain these in more detail in \Cref{sec:overview}.
\fi

We stress that our reduction is \emph{uniform} in the strongest possible sense: if an adversary uses quantum advice $\ket{\aux}$ then the reduction uses quantum advice $\ket{\aux}^{\otimes t}$ for some polynomial $t$.
Furthermore, $t = 1$ is possible for any $\cal A$ as long as $\ket\aux$ is an appropriate eigenstate.
See \Cref{remark:advice-preservation} for details.

\paragraph{On tightness of the reduction.}
We remark that $\gamma^k + \negl(\lambda)$ is likely the best general bound that one could hope for.
The $\gamma^k$ term is inherent since if the best attack on the original protocol has success probability $\gamma$, then simply running this attack on each repetition independently yields an attack achieving success probability $\gamma^k$.
The negligible term also cannot be eliminated under reasonable assumptions.
In particular, the classical 2-message counterexample by Dodis, Jain, Moran, and Wichs~\cite{DJMW12} generalizes to the post-quantum setting, thus the negligible term is inherent assuming existence of exponentially hard post-quantum extended second-preimage resistant hash functions.

\subsection{Applications of 3-message hardness amplification}

\Cref{thm:main-informal} immediately implies hardness amplification for several quantum cryptographic primitives.

\paragraph{Quantum commitments.} Bit commitments are a fundamental cryptographic primitive where a sender can commit to a bit $b$ without revealing it at first (this is the \emph{hiding} property), and later can reveal the bit but without the ability to change the bit (this is the \emph{binding} property).
Recently our understanding of commitment schemes in the quantum setting has considerably advanced.
In particular, there is a robust existential equivalence between commitments and many quantum cryptographic primitives including EFI pairs, which are pairs of efficient mixed states that can only be inefficiently distinguished~\cite{brakerski2022computational}.
Therefore, it is likely that commitments and EFI pairs play a similar ``minimal assumption'' role (analogous to one-way functions classically) to quantum cryptography.

An important question that has remained open is whether the computational security of quantum commitments (and friends) can be amplified.
In other words, given an arbitrary quantum commitment scheme where either the hiding or binding property holds with \emph{weak} (computational) security, can we generically transform it into another quantum commitment scheme where hiding and binding hold with \emph{strong} security? This question was explicitly raised by Yan~\cite{yan2023general}.

Our parallel repetition theorem for computationally secure protocols directly implies hardness amplification for quantum bit commitments, and thus showing robustness of the existence of commitments from a new angle.
\begin{corollary}[Hardness amplification for commitments]
\label{thm:commitments-informal}
	There is a quantum commitment scheme but only with computational weak hiding (or binding) security, if and only if there is a strong quantum commitment scheme. 
\end{corollary}
We argue this as follows: without loss of generality it suffices to consider noninteractive commitment schemes using Yan's compiler \cite{yan2023general}.
The binding security of the noninteractive scheme can be formulated in terms of the success probability of any efficient adversary in a 2-message security game; correspondingly the security of the repeated scheme can be formulated in terms of any efficient adversary's success probability in the parallel repeated security game, which by \Cref{thm:main-informal} decays to negligible at an exponential rate.
Amplification of hiding can be achieved via flavor-switching \cite{yan2023general,gunn2023commitments,hhan2023hardness}.
\ifacmart\else We describe this in detail in \Cref{sec:commitments}.\fi
We also show how this can be used to drastically simplify constructing commitments from hardness of decoding black hole radiation, originally proven by Brakerski~\cite{brakerski2023black}, in \ifacmart\else\Cref{sec:blackholes}\fi.

\paragraph{Quantum Yao's XOR lemma.}
By the equivalence of quantum commitments and EFI pairs, we also obtain hardness amplification for EFI pairs, answering an open question of Brakerski, Canetti, and Qian~\cite{brakerski2022computational}.
In fact, we can even use it to show polarization for EFI pairs\ifacmart\else{} (\Cref{cor:efi-polarization})\fi.

\begin{corollary}[XOR lemma for EFI pairs]
\label{thm:efi-informal}
If there exists (an ensemble of) weak EFI pairs $(\rho_0,\rho_1)$ that are statistically far but cannot be distinguished with advantage better than $\epsilon$, then the $k$-fold XOR of $(\rho_0,\rho_1)$ cannot be distinguished with advantage better than $\epsilon^{k/2} + \negl(\lambda)$.
In particular, this gives a (strong) EFI pair if $\epsilon^k$ is negligible.
\end{corollary}

We point out that from this and leveraging an equivalence between quantum state distinguishing and quantum predicates, we can immediately derive a quantum analogue of Yao's XOR lemma \cite{yao82theory}, which states that weak computational unpredictability of Boolean predicates (over some distribution of inputs) is amplified when the results of several independent instances are XOR-ed together.
A quantum predicate can be defined as two orthogonal average-case inputs $\rho_+$ (YES), $\rho_-$ (NO) with $\rho_+ \rho_- = 0$, and the goal of the predictor is to correctly predict the sign with advantage $\epsilon$.
This question was previously asked by Brakerski~\cite{brakerski2023black} (private communication) and Colisson~\cite{Col19}.

\begin{corollary}[Quantum Yao's XOR lemma]
    The $k$-fold XOR of an $\epsilon$-unpredictable quantum predicate for $\rho_+, \rho_-$ is $(\epsilon^{k/2} + \negl)$-unpredictable.
\end{corollary}

To see a circuit lower bound application of this, we can naturally define ``projection complexity classes'', a quantum-input analogue of decision complexity classes.
Then we have that for any such class $\mathsf C$ that is closed under composition with a polynomial fan-in XOR (like the analogue for $\mathsf{PSPACE}$), $\mathsf C$ is strongly hard-on-average against $\mathsf{BQP}$ machines if and only if $\mathsf C$ is weakly hard-on-average against them.

\paragraph{Quantum money.} A public-key quantum money scheme consists of quantum states (called \emph{quantum banknotes}) that can be publicly verified by anyone with the public-key, yet remain computationally infeasible to clone. A major goal of quantum cryptography research has been to construct \emph{public-key} quantum money schemes with security based on well-understood assumptions. Aaronson and Christiano~\cite{aaronson2012quantum} proved a per-key amplification for a special class of schemes called projective money schemes, and asked whether strong hardness amplification is possible for quantum money schemes.
We prove a general amplification that applies to \emph{any} public-key quantum money scheme:
\begin{corollary}[Hardness amplification for quantum money]
\label{thm:money-informal}
	Public-key quantum money schemes satisfying weak uncloneability exist, if and only if there exists a public-key quantum money scheme (satisfying strong unclonability).
\end{corollary}
Similar to amplifying commitments, this also follows directly from the observation that the security of a public-key quantum money scheme can be formulated in terms of a $2$-message security game, thus it immediately generalizes to e.g.\ quantum lightning and private-key quantum money.
\ifacmart\else We describe this in detail in \Cref{subsec:money}.\fi

\ifstoc\else
\paragraph{Amplification of post-quantum security.}
We remark that, if the original protocol is \emph{classical}, then the repeated protocol is also classical. Hence \Cref{thm:main-informal} also implies a parallel repetition theorem for general 3-message \emph{post-quantum} protocols; this was not previously known.
\fi

\subsection{Barrier for parallel repetition beyond \texorpdfstring{$3$}{3}-message protocols}
We also show that our 3-message parallel repetition theorem (\Cref{thm:main-informal}) cannot extend to $4$-message protocols under reasonable cryptographic assumptions, even if we are restricted to the post-quantum setting.
This is a (post-)quantum analogue of the classical result by Bellare, Impagliazzo and Naor~\cite[Section 3.3]{bellare1997does}%
\footnote{
    This is essentially the same but one subtle difference is that our counterexample does not require setup.
    Using setup in a counterexample is arguably problematic as pointed out in \cite[Section 2.1]{PW07}.
    Note that (post-quantum) non-interactive non-malleable commitments with setup can be instantiated from much weaker assumptions, e.g.\ using a non-malleable encryption or in the (quantum) random oracle model as was done in \cite{bellare1997does}.
}.

\begin{theorem}[Impossibility of parallel repetition, informal]
    If there is a post-quantum $c$-message concurrent-secure many-to-many non-malleable commitment scheme, then for every polynomial $k$ there is a $2c$-message post-quantum interactive protocol such that the security of a $k$-fold repetition of the protocol does not decrease compared to the original protocol.
\end{theorem}

For the special case of non-interactive commitments ($c = 2$), we would get a $4$-message impossibility.
We note that while there are no known post-quantum secure non-interactive non-malleable commitments, ``pre-quantum'' non-interactive non-malleable commitments can be constructed from various subexponential hardness assumptions \cite{KS17,BL18,GKLW21}, and so we view this assumption as plausible.
Note that a weaker post-quantum \emph{one-to-one} secure constant-round non-malleable commitment scheme is known to exist assuming post-quantum one-way functions \cite{liang2022new}, and this suffices for a special case of $k = 2$.

We note that classically, stronger impossibilities are known: there is a 4-message protocol whose $k$-fold computational security cannot be shown to decrease with black-box reductions for any polynomial $k$ \cite[Section 3.4]{bellare1997does}, and there is an 8-message protocol whose $k$-fold computational security is at least constant, regardless of proof techniques \cite{PW07}.
These might also generalize to the post-quantum setting, assuming strong but reasonable assumptions like post-quantum CCA-secure non-interactive commitments and post-quantum constant-round universal arguments.
We consider this sufficient evidence to conjecture that parallel repetition does not amplify 4-message (post-)quantum protocols, but we leave improving the impossibility for future work.

\subsection{Round compression for quantum argument systems}
An \emph{interactive argument} is a form of interactive proof where the completeness and soundness conditions hold with respect to computationally efficient provers.
An important complexity measure of interactive arguments (and interactive proofs in general) is the round complexity. One surprising result in the theory of quantum interactive proofs, due to 
Kitaev and Watrous~\cite{kitaev2000parallelization}, is that all (single-prover) quantum interactive \emph{proof} systems (where soundness holds against computationally unbounded adversaries) can be compressed to just three messages.  We show the analogous statement for quantum interactive arguments via the round compression technique of Kempe, Kobayashi, Matsumoto, and Vidick~\cite{kempe2007using}.  Our technical contribution is to make the reduction efficient.

\begin{theorem}[Round compression, informal]
\label{thm:compression-informal}
    Let $L$ be a language with an $m$-message quantum interactive argument with completeness $1 - c$ and soundness error $s = 1 - \delta$ for $m \ge 3$.
    Then there exists a $3$-message quantum interactive argument for $L$ with completeness $1 - 2c/(m - 1)$ and soundness error $1 - \delta/(m - 1)^4$.
    The verifier and communication complexity incur only
 a $\poly(m)$ multiplicative overhead.
\end{theorem}

To counteract the worse soundness error,
we can again apply \Cref{thm:main-informal} to the compressed protocol to obtain a $3$-message interactive argument for $L$ with negligible soundness error.
Combining these two results, we obtain a general round-preserving soundness amplification theorem for quantum arguments:

\begin{corollary}[Round-preserving amplification for arguments]
\label{thm:soundness-amp-informal}
    Let $L$ be a language with an $m$-message quantum interactive argument with completeness $1 - \negl$ (resp., $1$) and soundness error $1 - 1/\poly$.
    Then there exists a $\min\{3, m\}$-message quantum interactive argument for $L$ with completeness $1 - \negl$ (resp., $1$), negligible soundness error, and similar complexity.
\end{corollary}

\ifacmart\else We prove these formally in \Cref{sec:compression}.\fi
We remark that the crucial aspect of \Cref{thm:compression-informal} and \Cref{thm:soundness-amp-informal} is that they preserve the communication complexity and the verifier complexity of the original protocol.
(Indeed, a trivial round compression for argument systems that is not complexity-preserving can be obtained by having the prover forward its input to the verifier.)
We are not aware of any classical analogue of this round compression result.

\ifacmart\else In \Cref{sec:zk}, we \fi note that similar techniques allow us to further compile any quantum argument into to a (quantum communication) $\Sigma$-protocol \cite{Marriott2005}, and thus starting from an honest-verifier zero knowledge protocol, we can get a $3$-message malicious-verifier zero knowledge protocol, albeit the soundness becomes worse.
We further discuss how to get back negligible soundness \ifacmart\else at the end of \Cref{sec:zk}\fi.

\begin{corollary}[Round compression of zero-knowledge protocols]
    For any language $L$ that admits an honest-verifier quantum statistical (resp.\ computational) zero-knowledge protocol and computational (resp.\ statistical) soundness, $L$ also admits a malicious-verifier public-coin statistical (resp.\ computational) zero knowledge protocol with 3 messages, and $1 - 1/\poly$ computational (resp.\ statistical) soundness, and similar complexity.
\end{corollary}

\ifstoc\else
\subsection{Related works}

Prior works have derived quantum direct product theorems or quantum XOR lemmas in the query-efficient (or communication-efficient) setting \cite{ASW06,She11,Lee2013}.
Morally these are 2-message (post-quantum) parallel repetitions.
However, the query-efficient setting is usually weaker than our time-efficient setting and uses drastically different (non-algorithmic) techniques.
For the rest of the discussion we focus on time-efficient hardness amplification.

In \cite{bostanci2023unitary} a parallel repetition theorem for quantum canonical form commitments was proved, but it only handled classical side information and furthermore only achieved a polynomial rate of decay of the success probability in the repeated protocol.
As a consequence, we also improve their Theorem 6.8 such that any inverse polynomial fidelity (or any error that is inverse polynomially bounded away from 1) suffices.

In the classical setting, parallel repetition for three-message arguments (or ``weakly-verifiable puzzles'') was studied by \cite{bellare1997does,canetti2005hardness}, with the latter showing an optimal exponential soundness amplification.
Our three-message parallel repetition result also follows the high-level proof strategy of \cite{canetti2005hardness} while borrowing insights from proofs of XOR lemma \cite{yao82theory,levin1987one,goldreich2011yao}.
Recent works have observed that in some cases, the \cite{canetti2005hardness} amplification result can be applied essentially without modification in the quantum setting.
Radian and Sattath~\cite{radian2019semi} point out that \cite{canetti2005hardness} generalizes to handle 2-message \emph{post-quantum} (classical communication) protocols. %

Morimae and Yamakawa~\cite{morimae2022one} extend this argument further, adapting \cite{canetti2005hardness} to give a parallel repetition theorem for 2-message quantum protocols of the following special form:
	\begin{enumerate}
            \item[\ifacmart(0)\else0.\fi] Both parties a priori agree on a parameter $t$.
		\item The challenger generates a classical verification key $k$, then uses $k$ to generate $t$ copies of a quantum ``puzzle'' state $\ket{\mathsf{puz}}$, which it sends to the adversary.
		\item The adversary returns a classical answer $k'$.
		\item The challenger accepts or rejects based on $k, k'$.
	\end{enumerate}
They use this result to argue that \emph{weak} one-way state generators (OWSGs) imply OWSGs, analogous to Yao's amplification of one-way functions. Due to the restriction on the behavior of the challenger~--- essentially, that its secret state is classical~--- this result does not suffice for parallel repetition of general 2-message quantum protocols, and does not extend to 3-message protocols even with classical communication.
Furthermore, it always requires many copies of the adversary's auxiliary input whereas our reduction can be advice preserving for eigenstates.

In addition, neither commitments nor quantum money fall within the scope of their result. In the commitment case, this is because both messages in the security game are quantum, and furthermore a general quantum commitment does not have a classical verification key; indeed, the information required to verify the commitment is typically entangled with the state sent to the adversary. For quantum money, the issue is instead that the \cite{morimae2022one} reduction shows only that given an adversary for the parallel repetition of a $t$-copy protocol, we obtain an adversary for a single repetition of the corresponding $t'$-copy protocol for some $t' = t \cdot \poly(\lambda)$. This corresponds to giving the adversary multiple copies of the money state, which of course makes the cloning task trivial.

Our reductions share many techniques with prior works in quantum cryptographic reductions, especially in the area of quantum rewinding \cite{watrous2006zero,chia2021black,chiesa2022post,lombardi2022post}. Like the cited works, we make extensive use of Jordan's lemma and alternating sequences of projective measurements.

In recent work by Lombardi, Ma, and Spooner~\cite{lombardi2022post}, they achieved expected polynomial time quantum rewinding, in part by accelerating certain components of \cite{chiesa2022post} using the quantum singular value transform (QSVT). In this work, we also make use of the QSVT, but for a quite different purpose: coherent post-selection. Unlike in \cite{lombardi2022post}, we crucially rely on the ability of the QSVT to manipulate singular vectors while maintaining coherence between subspaces; see \Cref{sec:tech-twofold} for more details.

\fi

\newcommand{\printacks}{
We thank Scott Aaronson for bringing to our attention the open question of hardness amplification for quantum money raised by Aaronson and Christiano~\cite{aaronson2012quantum}.
We thank Ran Canetti for the discussions as well as suggesting the round compression idea.
We also thank Xiao Liang and Miranda Christ for the references on post-quantum non-malleability.
LQ is supported by DARPA under Agreement No.\ HR00112020023. JB and HY are supported by AFOSR award FA9550-21-1-0040, NSF CAREER award CCF-2144219, and the Sloan Foundation.
}

\ifanonymous\else\ifacmart\else
\paragraph{Acknowledgments.}
\printacks
\fi\fi

\section{Technical overview}
\label{sec:overview}

\subsection{2-message non-uniform parallel repetition}
\label{sec:tech-twofold}

In this section, we give an informal proof sketch for the special case of taking a $2$-fold parallel repetition of a $2$-message quantum protocol.
This special case is easier to understand and cannot be immediately handled by easy changes to \cite{canetti2005hardness}.
It turns out that the proof for this special case also contains most of the main ideas in the proof for the general non-uniform reduction.

We begin with some notation.
A challenger in a $2$-message protocol is identified with a pair $(V, P)$, for $V$ a unitary and $P$ a projector, and an adversary in a $2$-message protocol is identified with a pair $(U, \ket{\aux})$, for $U$ a unitary and $\ket{\aux}$ a quantum input.
There are three registers: $\reg A, \reg M, \reg C$, being the adversary's register, the message register, and the challenger's register respectively.
We can write the protocol as follows:
\begin{itemize}
    \item (Challenge) The challenger initializes both $\reg M, \reg C$ to $\ket{0}$, and applies the unitary $V$ to registers $\reg{MC}$.
    \item (Response) The adversary applies some unitary $U$ to registers $\reg{AM}$, where $\reg A$ initially contains some ``advice'' state $\ket{\aux}$.
    \item (Decision) The challenger applies a projective measurement $\{P, \Id - P\}$ to registers $\reg{MC}$, and accepts if and only if he gets outcome $P$.
\end{itemize}
Without loss of generality, we assume all operations are unitaries or projective measurements since we can expand the private registers $\reg{C}$ and $\reg{A}$ appropriately.  A $2$-fold parallel repetition of $(V, P)$ is simply $(V^{\otimes 2}, P^{\otimes 2})$, acting on registers $\reg M_1, \reg M_2, \reg C_1, \reg C_2$.
For $i \in \{1,2\}$, we write $V_i$ to denote the unitary that applies $V$ on registers $\reg M_i, \reg C_i$; $P_i$ to denote the projective measurement on registers $\reg M_i, \reg C_i$.%

Suppose $(V, P)$ has computational soundness $\epsilon + \negl$, and we would like to prove that $(V^{\otimes 2}, P^{\otimes 2})$ has computational soundness $\epsilon^2 + \negl$.
Assume for the sake of contradiction that there is a 2-fold adversary $(U,\ket{\aux})$ that achieves an inverse polynomial (for simplicity) advantage over $\epsilon^2$.
That is, the adversary is accepted with probability $\delta^2$, where $\delta - \epsilon$ is inverse polynomial.
Our goal is to construct an $1$-fold adversary that is accepted by the original challenger with probability close to $\delta$.

We first give a unified high level approach of the \emph{classical} proof for both tight parallel repetition \cite{canetti2005hardness} and the XOR lemma \cite{yao82theory} (or Levin's isolation lemma \cite{levin1987one,goldreich2011yao}).
Later we will extend this high level approach to the quantum setting.
The main idea behind all these proofs is similar, we construct a $1$-fold adversary by simulating a second challenger with a suitable challenge.  Consider the following two cases.
\begin{enumerate}[label=(\roman*)]
    \item\label{case:good_second_challenge}  There exists a fixed challenge $c_2$ such that running the $2$-fold adversary on $(c, c_2)$ outputs a response that is accepted by the first repetition with probability $\ge \delta$. 
    \item\label{case:all_second_challenges_bad} For every challenge in the second repetition, the adversary is accepted by the first repetition with probability $\le \delta$.
\end{enumerate}
If we are in case \ref{case:good_second_challenge}, then we can construct a non-uniform adversary by giving the $1$-fold adversary $c_2$ as advice.  On the other hand, if we are in case \ref{case:all_second_challenges_bad}, then the $2$-fold adversary is accepted by the second repetition with probability $\ge \delta$ whenever it breaks the first repetition.
To see why this is the case, let $G_1(c), G_2(c)$ be the events that the adversary is accepted by the first/second repetition on a random challenge $c$ respectively.
Then by Bayes' rule,%
\begin{align}
    \ifacmart
\else
    \label{eq:coreidea}
    \delta^2 = \Pr[G_1 \land G_2] = \E_{c_1, c_2}[\Pr[G_1] \cdot \Pr[G_2 | G_1]] \le \delta \cdot \E_{c_1, c_2}[\Pr[G_2 | G_1]],
    \fi
\end{align}
implying that $\E_{c_2}[\Pr[G_2 | G_1]] \geq \delta$.  Thus the algorithm for the $1$-fold adversary is to simulate the $2$-fold protocol, with a real challenger sampling $c_1$ for the first repetition and the challenge $c$ in the second repetition until the first repetition accepts, and then return the response to the second challenger.  

We now attempt to generalize this to the quantum setting. As a first attempt, a natural quantum analogue of case \ref{case:good_second_challenge} could be the condition
\begin{equation}
\label{eq:q-case1}
\exists \ket{m},~ \norm{P_1UV_1(\ket\aux_{\reg A} \ket{m}_{\reg M_2} \ket0_{\reg M_1\reg C_1\reg C_2})}^2 \ge \delta,
\end{equation}
which says that there is some message $\ket m$ we can insert into the second repetition so the adversary wins the first repetition with probability at least $\delta$. The reduction for this case is straightforward: put the real challenge in $\reg M_1$, run the adversary $U$, then output $\reg M_1$; this succeeds with probability $\delta$ by \cref{eq:q-case1}. We will see soon that case \ref{case:all_second_challenges_bad} requires a slightly different condition, but for now we will proceed with \cref{eq:q-case1} as stated.
\Cref{eq:coreidea} suggests the following natural reduction for case \ref{case:all_second_challenges_bad}:
\begin{enumerate}
    \item Initialize $\ket{0}_{\reg M_1 \reg C_1}$ and simulate the challenger in the first repetition by running $V_1$.
    \item Put the real challenge in $\reg M_2$.
    \item Run the 2-fold adversary $U$ on $\reg A \reg M_1 \reg M_2$.
    \item ``Post-select'' on the event that the challenger accepts in the first repetition (i.e., on $P_1$).
    \item Output $\reg M_2$ as response.
\end{enumerate}
Before the post-selection step, the state of the system is $UV_1V_2\ket\aux\ket0$. We know that $\norm{P_1P_2UV_1V_2\ket\aux\ket0}^2 \allowbreak\geq \delta^2$ by assumption (the adversary is accepted with probability $\geq \delta^2$), and that $\norm{P_1UV_1V_2\ket\aux\ket0}^2\allowbreak < \delta$ by the negation of \cref{eq:q-case1}.
Suppose that we are now able to post-select on $P_1$; i.e., to prepare the state $\frac{P_1UV_1V_2\ket\aux\ket0}{\norm{P_1UV_1V_2\ket\aux\ket0}}$. Then as in the classical case, we would be done, since that state achieves success probability
\begin{align}
    \ifacmart
\else
    \label{eq:postselectionprobability}
    \norm{P_2 \cdot \frac{P_1UV_1V_2\ket\aux\ket0}{\norm{P_1UV_1V_2\ket\aux\ket0}}}^2 = \frac{\norm{P_1P_2UV_1V_2\ket\aux\ket0}^2}{\norm{P_1UV_1V_2\ket\aux\ket0}^2} \ge \frac{\delta^2}\delta = \delta.
    \fi
\end{align}
How do we perform post-selection? Classically, this is achieved by rejection sampling.

As a seasoned reader might expect at this point, na\"ive rejection sampling does not immediately generalize to the quantum setting. This is because measuring $P_1$ disturbs $\reg M_2$, and it is not in general possible to clone the state on $\reg M_2$; worse, it may be that this state is entangled with the challenger's private register $\reg C_2$.
Indeed, for canonical form commitment schemes, $\reg M_2$ and $\reg C_2$ are highly entangled, and the challenger will later check for the presence of entanglement.

\paragraph{Attempt: Alternating projectors.}
Classical rejection sampling can be thought of as a form of \emph{rewinding}.
Hence a natural first attempt is to try to apply recent quantum rewinding techniques \cite{watrous2006zero,chia2021black,chiesa2022post,lombardi2022post}.
Following these works, we can implement a form of post-selection without cloning by alternating $P_1$ (the first repetition accepting) with the projective measurement $Q_1 := (UV_1)\ketbra00_{\reg M_1\reg C_1}\allowbreak(UV_1)^\dagger$ (the first repetition being initialized correctly) until $P_1$ accepts.

There are a few issues with this attempt.
Alternating projector algorithms can be analyzed via the Jordan (singular value) decomposition of $P_1Q_1 = \sum_i\varsigma_i\ketbra{w_i}{v_i}$.
Before the post-selection step, the state is clearly in $Q_1$, and so it can be written as $UV_1V_2\ket\aux\ket0 = \sum_i \alpha_i\ket{v_i}$.
For simplicity assume for now that we are able to rotate all the singular vectors and the singular values are all non-zero, then the output state of the alternating projectors will be
\[\sum_{i : \varsigma > 0} \alpha_i \ket{w_i} \otimes  \ket{\tilde\varsigma_i},\]
where $\ket{\tilde\varsigma_i}$ is the alternating projection history register that only depends on the singular value $\varsigma_i$ (which may be subnormalized).
The presence of the history register is problematic since tracing it out amounts to measuring the singular value $\varsigma_i$.
Since this measurement is unlikely to commute with $P_2$, we cannot argue that the success probability is at least $\delta$ as above. To avoid this problem, we would need to uncompute the history register, which we do not know how to do.

Even if we ignore this issue, and assume we can somehow uncompute the history to obtain the state
\begin{equation*}
	\ket{\psi} = \sum_{i: \varsigma_i > 0} \alpha_i \ket{w_i},
\end{equation*}
we still would not be able to say that the adversary is accepted with high probability. Recall that our ``target'' state is
\begin{equation}
    \label{eq:postselecttarget}
    \frac{P_1UV_1V_2\ket\aux\ket0}{\norm{P_1UV_1V_2\ket\aux\ket0}} \approx \sum_i \frac{\varsigma_i}{\sqrt\delta} \alpha_i \ket{w_i}.
\end{equation}
The best bound we can get (via the triangle inequality and \cref{eq:postselectionprobability}) is
\[
  \norm{P_2 \ket{\psi}} > \sqrt{\delta} - \norm{P_2 \sum_{i: \varsigma_i > 0} \mparen{1 - \frac{\varsigma_i}{\sqrt\delta}} \alpha_i \ket{w_i}},
\]
which may be trivial (e.g. if $\alpha_i \approx 1$ for $\varsigma_i \ll \sqrt{\delta}$).
Note that this last term can be shown to be non-negative in the classical case, but this could fail quantumly due to the possibility of destructive interference with respect to $P_2$.
Therefore, we cannot hope to simply improve the bound on the probability without changing the state $\ket{\psi}$ itself.

\paragraph{Solution: QSVT.} To summarize, the alternating projectors approach suffers from two issues: \begin{enumerate*}[label=(\alph*)]
 	\item loss of coherence due to explicit computation of $\varsigma_i$, and
 	\item incorrect weighting of different singular vectors.
 \end{enumerate*}
To solve both of these issues, we make use of a more sophisticated quantum algorithmic tool, the \emph{quantum singular value transformation} (QSVT) \cite{gilyen2019quantum}. Roughly, the QSVT enables efficient, coherent transformations of the form
\begin{equation*}
	\sum_{i} \alpha_i \ket{v_i} \to \sum_{i} \alpha_i f(\varsigma_i) \ket{w_i}
\end{equation*}
for low-degree real polynomials $f$ with $|f(x)| \leq 1$ when $|x| \leq 1$.
We observe that our post-selection task corresponds to $f(\varsigma) = \varsigma/\sqrt{\delta}$.
Then Gilyen et al.\ \cite[Theorem 17]{gilyen2019quantum} show how to construct a low-degree function $g$ which does satisfy the boundedness conditions, and which approximates $\varsigma/\sqrt{\delta}$ on the range $[0,\sqrt{\delta}]$.
Applying the QSVT with respect to this $g$ achieves the necessary post-selection, \emph{provided the spectral norm (maximum singular value) of $P_1 Q_1$ is bounded by $\sqrt{\delta}$}. %
Furthermore, the reduction goes through as long as the approximation error is $\ll \delta - \epsilon$.

Now we want a promise that all of the singular values of $P_1Q_1$ are at most $\sqrt{\delta}$ in order to satisfy the necessary boundedness conditions.
To achieve this, we simply change the the condition for case \ref{case:all_second_challenges_bad} to be that the singular values of $P_1 Q_1$ are bounded by $\sqrt{\delta}$, and thus in this case we can safely apply QSVT to approximately post-select. However, we note that the negation of this condition is no longer \cref{eq:q-case1}, as $P_1Q_1$ might have a large singular value corresponding to a state that does not come from a state of the form $\ket{\aux}_{\reg{A}} \otimes \ket{m}_{\reg{M}_2}$.

Nevertheless, we can ``fix'' case 1 by taking advantage of non-uniformity.
Suppose that $P_1 Q_1$ has some singular value $\varsigma_i$ larger than $\sqrt{\delta}$, and let $\ket{v_i}_{\reg A \reg M_1 \reg C_1 \reg M_2}$ be a corresponding right singular vector. Since $\ket{v_i}$ is in $Q_1$, $U^{\dagger} \ket{v_i} = (V_1 \ket{0}_{\reg M_1 \reg C_1}) \otimes \ket{\psi}_{\reg A \reg M_2}$ for some advice state $\ket{\psi}$. Then
\begin{equation*}
\norm{P_1U(V_1\ket0_{\reg M_1\reg C_1}) \ket{\psi}_{\reg A \reg M_2}}^2 = \norm{P_1 \ket{v_i}}^2 = \varsigma_i^2 \ge \delta.
\end{equation*}
That is, in case \ref{case:good_second_challenge} the adversary $(U,\ket{\psi})$ achieves success probability $\delta$, which completes the proof in the non-uniform case. Note that, unlike in the classical case, $\ket\psi$ may be entangled across $\reg A$ and $\reg M_2$.

\paragraph{Extension to $k$-fold repetition.}
In the classical setting for general $k$, we have $k$ cases as follows. Let $G_i$ be the event that the adversary wins the $i$-th repetition of the protocol, and suppose that $\Pr[G_k] \geq \delta^k$. It is straightforward to generalize the above to see that there exists some $j \in [k]$ and $c_{j + 1}, ..., c_k$ such that
\begin{equation*}
	\Pr[G_j | \wedge_{i=1}^{j - 1} G_i, c_{j + 1}, ..., c_k] \geq \delta,
\end{equation*}
and we can follow the same rejection sampling strategy as above.

In the quantum setting, we similarly generalize the projector $Q_1$ from the 2-fold case as \[Q_{\leq j} \coloneqq U V_1 \ldots V_j \proj{0}_{\reg M_{\leq j} \reg C_{\leq j}} (U V_1 \ldots V_j)^{\dagger},\] and define $P_{\leq j} \coloneqq P_1 \ldots P_j$. By assumption, $\norm{P_{\leq k} Q_{\leq k} \ket{\aux}\ket{0}}^2 \geq \delta^k$, and so in particular the spectral norm of $P_{\leq k} Q_{\leq k}$ is at least $\sqrt{\delta^k}$. It follows that there is some $j \in [k]$ such that
	\begin{equation}
		\label{eq:spectral-norm-condition}
		\norm{P_{\leq j} Q_{\leq j}} \geq \sqrt{\delta^j} \text{ and } \norm{P_{<j} Q_{<j}} \leq \sqrt{\delta^{j-1}}~.
	\end{equation}
	Therefore, given as non-uniform advice a state $\ket{\psi}$ with $\norm{P_{\leq j} Q_{\leq j} \ket{\psi}}\allowbreak \geq \sqrt{\delta^j}$, by applying the QSVT with respect to $P_{<j} Q_{<j}$ as in case 2 above we obtain an adversary with success probability $\delta$. 

\subsection{Uniform reduction}

In the previous section, we made crucial use of non-uniformity to provide the adversary with an index $j$ satisfying \cref{eq:spectral-norm-condition} and a vector $\ket{\psi}$ with $\norm{P_{\leq j} Q_{\leq j} \ket{\psi}} \geq \sqrt{\delta^j}$. In this section, we will describe how to efficiently prepare $j,\ket{\psi}$ from (polynomially many copies of) the adversary's initial state $\ket{\aux}$.

We will need to start by relaxing \cref{eq:spectral-norm-condition}, as we cannot in general efficiently check the spectral norm of an operator. We address this by observing that our spectral norm condition for post-selection via the QSVT can be substantially weakened: it suffices for the \emph{input state} to have small ($\ll \delta^k$) amplitude on (right) singular vectors $\ket{v_i}$ of $P_{<j}Q_{<j}$ with singular value $\varsigma_i > \sqrt{\delta^{j-1}}$.

Our task then becomes, formally: find an index $j$ and state $\ket{\psi}$ such that
\begin{enumerate*}[label=(\roman*)]
	\item $\norm{P_{\leq j} Q_{\leq j} \ket{\psi}} \geq \sqrt{\delta^j}$, and
	\item writing $\ket{\psi} = \sum_{i} \alpha_i \ket{v_i}$ where $\sum_i \varsigma_i \ketbra{w_i}{v_i}$ is the singular value decomposition of $P_{< j} Q_{< j}$, we have $\sum_{i,\varsigma_i > \sqrt{\delta^{j-1}}} |\alpha_i|^2 \ll \delta^k$.
\end{enumerate*}
This is in fact a quantum analogue of a main algorithmic task in the preprocessing phase of \cite{canetti2005hardness}. In more detail, the analogous classical task is to find $j$ and a sequence of challenges $c_{j+1},\ldots,c_{k}$ such that, \begin{enumerate*}[label=(\roman*)]
 		\item after fixing challenges $c_{j+1},\ldots,c_{k}$ in repetitions $j+1,\ldots,k$, the residual probability of winning the first $j$ repetitions is at least $\delta^j$, and
 		\item \label{cond:first-j-bound} with probability $\gg 1 - \delta^k$ over $c_j$, after fixing $c_j,\ldots,c_{k}$ in repetitions $j,\ldots,k$, the probability of winning the first $j-1$ repetitions is at most $\delta^{j-1}$.
 	\end{enumerate*}

\paragraph{First attempt.}
Let $\ket{\psi_k} \coloneqq U V_1 \cdots V_k \ket\aux\ket0$. Recall that, by assumption, we have that $\norm{P_{\leq k} Q_{\leq k} \ket{\psi_k}} \geq \sqrt{\delta^k}$. For each $j$, let $\sum_i \varsigma_i^{(j)} \ketbra{w_i^{(j)}}{v_i^{(j)}}$ be the singular value decomposition of $P_{\leq j} Q_{\leq j}$.

Let us suppose for now that we have access to the binary projective ``singular value threshold'' measurement $\Pi^{(j)} = \sum_{i,\varsigma_i^{(j)} > \sqrt{\delta^{j}}}\allowbreak \proj{v_i^{(j)}}$, for each $j$.\footnote{Unlike in the classical setting, it is important here that we do not actually \emph{measure} the singular value, since this would cause too much disturbance.} We do not know how to realize this measurement efficiently, but it can be approximated in some sense \cite{gilyen2019quantum,chiesa2022post}. This will introduce a number of technical complications that we address later; for now, we assume access to the exact measurement. Observe that we can write condition \ref{cond:first-j-bound} equivalently as $\norm{\Pi^{(j-1)}\ket\psi}^2 \ll \delta^k$.

Our first attempt at a uniform reduction is as follows. We apply $(\Pi^{(k-1)}, I - \Pi^{(k-1)})$ to $t \gg 1/\delta^k$ copies of $\ket{\psi_k}$. If we ever see the outcome $\Pi^{(k-1)}$, the post-measurement state $\ket{\psi_{k-1}}$ is in $\Pi^{(k-1)}$, and so $\norm{P_{\leq k-1} Q_{\leq k-1} \ket\phi} > \sqrt{\delta^{k-1}}$, and we can then recurse on $\ket{\psi_{k-1}}$. Otherwise, since we never see the outcome $\Pi^{(k-1)}$, with high probability $\norm{\Pi^{(k-1)}\ket{\psi_k}}^2 \ll \delta^{k}$, and so we can output $j=k$ and $\ket{\psi} = \ket{\psi_k}$.
Finally, if we get to $\ket{\psi_1}$, we can simply output $j = 1$ and $\ket\psi = \ket{\psi_1}$.

Unfortunately, this approach only works for constant $k$. To see why, notice that to prepare a single copy of $\ket{\psi_{j-1}}$ we may need $1/\delta^k$ copies of $\ket{\psi_j}$.
Unlike in the classical setting, we cannot in general clone $\ket{\psi_j}$.
Hence the number of copies of $\ket{\psi_k}$ required (and the running time of the algorithm) scales as $\Omega(1/\delta^{k^2})$, which may be superpolynomial for $k = \omega(1)$.

\paragraph{Second attempt.}
To resolve this issue, we note that in order for the non-uniform reduction to work, it suffices to simply produce $j$ along with \emph{any} state in $\Pi^{(j)}$ with a small enough overlap with $\Pi^{(j - 1)}$, therefore in the case we measure $I - \Pi^{(j-1)}$, it suffices to recover a state from $\Pi^{(j)}$ instead of recovering exactly $\ket{\psi_j}$.
This is reminiscent of the ``state repair'' problem encountered in quantum rewinding \cite{chiesa2022post}; our algorithm will follow that template. In more detail, the reduction works as follows.
\begin{enumerate}
    \item \label[step]{step:overview-uniformmeasure} Measure the input state $\ket{\psi_k}_{\reg A}$ with $\Pi^{(k)}$. If it rejects, start over with a fresh copy of $\ket{\psi_k}$.
    \item \label[step]{step:overview-uniformloop} Repeat for $j = k, ..., 2$: \begin{enumerate}
        \item \label[step]{step:tech-est} Measure $\Pi^{(j - 1)}$ and $\Pi^{(j)}$ in an alternating fashion for up to $t \gg 1/\delta^k$ iterations.\footnote{Technically $\Pi^{(j - 1)}$ acts on an additional register $\reg C_j$. This is a minor point and does not really affect the algorithm nor the analysis. We can simply initialize all $\reg C_j$'s to 0 at the beginning and add them to $\reg A$ before aborting.} If $\Pi^{(j - 1)}$ ever accepts, go to the next iteration of the loop ($j-1$).
        \item \label[step]{step:tech-bail} Otherwise, keep performing alternating projections until $\Pi^{(j)}$ accepts, then output $j$ and $\reg A$ and abort.
    \end{enumerate}
    \item Output $j = 1$ and register $\reg A$.
\end{enumerate}
Using Jordan's lemma, and via similar reasoning to \cite{chiesa2022post}, it is possible to show that
\begin{enumerate*}[label=(\roman*)]
	\item because at the beginning of the $j$-th loop iteration, the state is in $\Pi^{(j)}$, the number of measurements performed in \cref{step:tech-bail} is $O(t)$ in expectation; and
	\item if we never see $\Pi^{(j-1)}$ in \cref{step:tech-est} then with high probability the state $\ket{\psi}$ output by the algorithm on termination satisfies $\norm{\Pi^{(j-1)} \ket{\psi}}^2 \ll \delta^k$.
\end{enumerate*}

\paragraph{Adapting to approximate POVMs.}
The algorithm described above is correct assuming access to the projectors $\Pi^{(j)}$.
In reality, we can only approximate them using (e.g.)\ Marriott--Watrous \cite{Marriott2005}. Furthermore, this approximate implementation is not a projection but a POVM; equivalently, it is a projection $\tilde{\Pi}^{(j)}$ acting on the register $\reg A$ and an auxiliary register $\reg W_j$ that is initially set to $\ket0$.

Following \cite{chiesa2022post}, the natural approach to extend the algorithm above to this case is to simply replace $\Pi^{(j-1)}$ and $\Pi^{(j)}$ measurements with their approximate counterparts, $\tilde{\Pi}^{(j-1)} \otimes \proj0_{\reg W_j}$ and $\tilde{\Pi}^{(j)}\otimes \proj0_{\reg W_{j-1}}$. The projection on the ancilla register for the other measurement aims to ensure its correct initialization.

This approach \emph{almost} works but for a subtle technical issue.
Even though $\reg W_{j-1}$ and $\reg W_{j}$ will be initialized to $\ket0$, after applying the first two projections in \cref{step:tech-est}, we no longer have any guarantees about the ancilla registers.
Therefore, even if we measure that $\tilde\Pi$ accepts, it does not imply that we have a state close to $\Pi$ since it could be that the ancilla registers were malformed.

As a starting point, let us first look at how well the previous algorithm works if we simply plug in $\tilde\Pi^{(j)}$'s (we omit the zero projector on the ancillas to keep the notations simple).
Since the ancillary issue only arises after we perform two projections $\tilde\Pi^{(j-1)}$ and $\tilde\Pi^{(j)}$, we observe:
\begin{enumerate}
    \item If $\tilde\Pi^{(j-1)}$ accepts in the first iteration, we must still (approximately) have a vector in $\Pi^{(j-1)}$ as the ancilla is initialized to zero at the beginning.
    \item Furthermore, the alternating projections can still estimate the singular value. If we, instead of going to $j - 1$ whenever $\Pi^{(j-1)}$ accepts, estimate the singular value and only declare we are in case $j$ when we are below some minuscule threshold, it turns out to still work. This is because as long as the threshold is small enough, when we are below the threshold, by gentle measurement, it must be the case that the auxiliaries are not too far from zero. Thus a small singular vector between $\tilde\Pi$'s is also a relatively small singular vector between $\Pi$'s.
    \item Now it remains to handle the last case where the first measurement rejects but the estimated singular value is still higher than the threshold. The final observation is that in fact the probability that we reach the last case is in fact bounded away from 1 for \emph{any} starting state: intuitively if the starting state has a large overlap with $\Pi^{(j-1)}$ then the first clause catches it with noticeable probability, otherwise the second clause catches it with noticeable probability. Therefore, when we reach the last case, we can simply recover \emph{any} state in $\Pi^{(j)}$ again so that we can restart from the beginning.
    Since the algorithm succeeds for any starting state with some probability, even if in each iteration the starting state is different, we will still eventually reach one of the two good cases with a sufficiently large number of trials.
\end{enumerate}

Leveraging these three observations, we solve this final issue by modifying the loop (\cref{step:overview-uniformloop}) with a more careful algorithm as follows:
\begin{enumerate}[noitemsep,label=(\alph*)]
    \item Repeat $t \gg 1/\tau$ times for some inverse polynomial threshold $\tau = \delta^{O(k)}$:
    \begin{enumerate}[label=(\roman*),ref=\theenumi(\roman*)]
	    \item \label[step]{step:uniform-bail} Initialize $\reg W_{j,j-1}$ to zero. Measure $(I - \tilde{\Pi}^{(j-1)}) \otimes \proj{0}_{\reg W_j}$. If the measurement rejects, proceed to the next iteration of the outer loop ($j-1$).
	    \item Otherwise, measure $\tilde{\Pi}^{(j)} \otimes \proj{0}_{\reg W_{j-1}},(I -\tilde{\Pi}^{(j-1)}) \otimes \proj{0}_{\reg W_j}$ in an alternating fashion for (say) $t^2$ iterations. Then, keep alternating until $\tilde{\Pi}^{(j)} \otimes \proj{0}_{\reg W_{j-1}}$ accepts. Let $\ket{\phi}_{\reg A,\reg W_{j}}\ket{0}_{\reg W_{j-1}}$ be the post-measurement state.
	    \item \label[step]{step:estimate-gamma} Use the outcomes of the alternating measurements to compute an estimate $\gamma$ of $\norm{((I -\tilde{\Pi}^{(j-1)}) \otimes \proj{0}_{\reg W_j})\ket\phi \ket{0}_{\reg W_{j-1}}}^2$. If $\gamma$ is above $1-\tau$, terminate the outer loop. Otherwise, proceed to the next iteration of the inner loop.
    \end{enumerate}
    \item Abort without any outputs.
\end{enumerate}
An additional key change is that we are now alternating $\tilde{\Pi}^{(j)} \otimes \proj{0}_{\reg W_{j-1}}$ and $(I - \tilde{\Pi}^{(j-1)}) \otimes \proj{0}_{\reg W_j}$.
We also use state repair again to recover a new state for the next iteration.
We remark that in order for the algorithm to work we also need to slightly shift the singular value threshold in each iteration, but we refer the readers to the full proof for these technical details.

We now formalize the observations above to analyze this new algorithm.
Note that if the \emph{first} application of $(I - \tilde{\Pi}^{(j-1)}) \otimes \proj{0}_{\reg W_j}$ \emph{rejects}, it must be that the post-measurement state is in $\tilde{\Pi}^{(j-1)}$ because $\reg W_{j}$ is initialized to $\ket{0}$; this is \emph{not true} for subsequent applications because the measurement may have rejected due to a malformed ancilla.

To argue correctness, we consider two cases. The first case is when, in some iteration of the inner loop, the estimate $\gamma$ is above the threshold $1-\tau$. In this case we must show that the post-measurement state $\rho = \Tr_{\reg W_j}(\proj{\phi})$ on $\reg A$ is (almost completely) in $\Pi^{(j)}$ and has very small overlap with $\Pi^{(j-1)}$. To see this, observe that by gentle measurement the state $\ket\phi$ is $\sqrt{\tau}$-close to a state of the form $\ket{\psi}_{\reg A} \ket{0}_{\reg W_{j}}$. The state $\ket{\psi}$ then has the property that $\ket{\psi}\ket{0}_{\reg W_j}\ket{0}_{\reg W_{j-1}}$ is $O(\sqrt{\tau})$-close to both $\tilde{\Pi}^{(j)}$ and $I - \tilde{\Pi}^{(j-1)}$. Since this latter state has ancillas initialized to zero, it follows that $\tilde{\Pi}^{(j)},\tilde{\Pi}^{(j-1)}$ approximate $\Pi^{(j)},\Pi^{(j-1)}$ on this state, and so $\ket{\psi}$ (which is close to $\rho$) is close to both $\Pi^{(j)}$ and $I-\Pi^{(j-1)}$.

Otherwise, if $\gamma$ is always below $1-\tau$, then in each iteration of the inner loop, we will terminate in \cref{step:uniform-bail} with probability at least $\tau$. It follows that, since $t \gg 1/\tau$, with overwhelming probability the loop will terminate in one of these two cases.

\begin{remark}[Advice preservation]
	\label{remark:advice-preservation} We note that, while our reduction preserves uniformity, it is not strictly advice-preserving (or constructive \cite{BBK22}), as it requires many copies of the adversary's advice state. %

This is inherent for any quantum reduction whose success probability ought to be higher than that of the adversary. Indeed, this is true even classically for randomized advice (and hence also for quantum advice via purification): given an adversary which succeeds with probability $\delta$ over the advice distribution, a black-box reduction given only one sample from the advice distribution cannot succeed with probability greater than $\delta$ in general.

We remark that the only reason for requiring many copies of the advice is in order to obtain a state in $\Pi^{(k)}$ in \cref{step:overview-uniformmeasure}. Thus, if the advice state is already in $\Pi^{(k)}$, one copy suffices.
\end{remark}

\subsection{Round compression}

We analyze the soundness of the round compression transformation of Kempe et al.~\cite{kempe2007using} when applied to argument systems.
At a high level, their transformation works by recursively converting an $(r+1)$-message protocol into an $(r/2 + 1)$-message compressed protocol.
In an honest execution, the prover begins by simulating the original $(r+1)$-message protocol until the $r/2$-th message, and sends the original uncompressed verifier's private registers to the challenger in the compressed protocol.  From there, the verifier flips a coin, deciding whether to continue by running the original protocol forwards or backwards in time.

If the verifier decides to execute the protocol backwards in time, the honest prover and verifier apply the inverse of uncompressed protocol, and at the end the verifier measures whether their private register returns to the state $\ket0$.
On the other hand, if the verifier decides to execute the protocol forwards, the honest prover and verifier execute the remainder of the uncompressed protocol and the verifier checks the same predicate that the uncompressed verifier does at the end of the uncompressed protocol.

Completeness is straightforward: the honest prover simply simulates the protocol using the original prover and verifier up to the midpoint, and then cooperate with the verifier to compute the protocol either in the forward or the backward direction.
To show (computational) soundness, we demonstrate an efficient reduction from an adversary for the compressed protocol to an adversary for the uncompressed protocol.  In particular, the adversary for the uncompressed protocol simulates an interaction between the compressed adversary and the compressed verifier, conditioned on the verifier executing the protocol backwards.  The adversary can then measure the simulated verifier's register, and conditioned that measurement accepting, the adversary now has a good initial state for the uncompressed protocol, and the state of the simulated verifier's register is $\ket{0}$, so it can be discarded.  

From there, the adversary sends their first message and continues by applying the inverse of the compressed adversary until round $r/2$.  After round $r/2$, they apply the same unitaries as the compressed adversary, conditioned on the compressed verifier executing the protocol forward in time.  Assuming that the compressed adversary was accepted with probability $(1-\epsilon)$, we show that the state after simulating either the forwards or backwards protocol is $(1 - 4\epsilon)$-close in squared Bures distance to a state that is accepted by the challenger in both cases.   Using the weak triangle inequality for the squared Bures distance, we find that the state of the verifier at the end of the protocol is $(1 - 16\epsilon)$-close in squared Bures distance to a state that is accepted by the challenger, implying that the challenger accepts with probability $1 - 16\epsilon$.  The use of squared Bures distance, instead of the more commonly-used trace distance, avoids a blowup from $\epsilon$ to $\sqrt{\epsilon}$ in this step.

This process halves the number of rounds at a cost of mapping $1 - \epsilon$ soundness to $1 - \epsilon/16$.  Iterating this protocol $\log m$ times, where $m$ is the number of messages in the original protocol, we arrive at a $3$-message protocol with soundness $1 - \epsilon/m^4$.

\ifstoc

\fi

\ifacmart
\else
\section{Preliminaries}

\subsection{Quantum information}
A quantum \emph{register} $\reg{R}$ is a named finite-dimensional complex Hilbert space.  We write $L(\reg{R})$ to denote the set of linear transformations on $\reg{R}$ and $\states(\reg{R})$ to denote the set of density matrices on $\reg{R}$ (i.e.\ positive semi-definite and unit-trace operators).  For a vector $\ket{\psi} \in \reg{R}$, we write $\psi$ to denote the density matrix $\proj{\psi}$, and for vectors $\ket{\psi}, \ket{\phi} \in \reg{R}$, we write $\braket{\psi|\phi}$ to denote the inner product.  For a vector $\ket{\psi} \in \reg{R}$, we write $\norm{\ket{\psi}}$ to denote the standard norm over $\reg{R}$, i.e. $\braket{\psi | \psi}$.  We write $\Tr(\cdot)$ to denote the trace and $\Tr_{\reg{R}}$ to denote the partial trace over a register $\reg{R}$.

For a linear operator $X \in L(\reg{R})$, let $\opnorm{X}$ be its operator norm and $\norm{X}_{1} = \Tr(\sqrt{X^{\dagger}X})$ be its trace norm.  For two density matrices $\rho, \sigma \in \states(\reg{R})$, let $\td(\rho, \sigma) = \frac{1}{2}\norm{\rho - \sigma}_{1}$ be the trace distance between the two.  We sometimes write $X_\reg{R}$ to indicate that $X$ acts on $\reg{R}$.  All un-labeled operators act on all registers that do not have an operator acting on them, and if an operator is associated with specific registers, we drop the register subscripts for brevity.

A \emph{binary projective measurement} is a pair $(\Pi,I-\Pi)$, where $\Pi$ is an orthogonal projector. By convention we refer to the outcome corresponding to $\Pi$ as $1$ and $I-\Pi$ as $0$. Since a binary projective measurement is completely specified by $\Pi$, we often refer to such a measurement simply as $\Pi$.

\begin{definition}[Eigenspace projectors]
	\label{def:eigenspace-projectors}
	Let $H = \sum_j \lambda_j \proj{j}$ be a Hermitian matrix. For $\kappa \in \R$, we denote by $\HProj{H}{<\kappa} \coloneqq \sum_{j, \lambda_j < \kappa} \proj{j}$ the projector on to eigenspaces of $H$ with eigenvalue less than $\kappa$. $\HProj{H}{\geq\kappa}$ is defined similarly.
\end{definition}

\subsection{Fidelity and Bures distance}
An important tool used in the paper will be the quantum fidelity and the related squared Bures distance.  Given two quantum states $\rho, \sigma \in \states(\reg{R})$, the \emph{fidelity} between $\rho$ and $\sigma$ is given by
\begin{equation*}
    \fidelity(\rho, \sigma) = \Tr\left(\sqrt{\sqrt{\rho}\sigma \sqrt{\rho}}\right)^2.
\end{equation*}
This definition of fidelity might be sometimes be referred to as the ``squared'' fidelity.  The fidelity can be related to the trace distance by a pair of inequalities called the Fuchs-van de Graaf inequalities.
\begin{proposition}[Fuchs-van de Graaf inequalities] 
For all density matrices $\rho$ and $\sigma$ over the same Hilbert space, we have that
\begin{equation*}
    1 - \sqrt{\fidelity(\rho, \sigma)} \leq \td(\rho, \sigma) \leq \sqrt{1 - \fidelity(\rho, \sigma)}\,.
\end{equation*}
\end{proposition}

It is well known that the fidelity is a useful quantity when examining the effects of performing a measurement on a quantum state.  Specifically, the gentle measurement lemma gives a bound on the trace distance a state can move after a measurement based on the probability of the measurement accepting.  

\begin{proposition}[Gentle measurement lemma~\cite{winter1999coding}]
	\label{lemma:gentle-measurement}
    Given a pure state $\rho$ and a projector $\Lambda$, let 
    \begin{equation*}
        \rho' = \frac{\Lambda \rho \Lambda}{\Tr(\Lambda \rho)}
    \end{equation*}
    be the post-measurement state.  Then $\fidelity(\rho', \rho) = \Tr(\Lambda \rho)$. It follows from the Fuchs-van de Graaf inequalities that $\td(\rho',\rho) \leq \sqrt{1 - \Tr(\Lambda \rho)}$.
\end{proposition}
Another way to view the trace of a projector applied to a state $\psi$ is as the fidelity with the closest state in the $+1$-eigenspace of the projector.  Formally we have the following lemma.

\begin{proposition}[Projector to max fidelity~{\cite[Theorem 9.2.2]{wilde2016quantum}}]
\label{lem:projector_to_fidelity}
	For any projector $\Pi$ and state $\proj{\psi}$,
	\begin{equation*}
		\Tr(\Pi \proj{\psi}) = \max_{\Tr(\Pi \sigma) = 1} \fidelity(\sigma, \proj{\psi})
    \end{equation*}
\end{proposition}

The fidelity corresponds to a squared inner product between states.  The \emph{squared Bures distance} is a related distance measure between two states $\rho$ and $\sigma$, defined as
    \begin{equation*}
        \buresdist(\rho, \sigma) = 2(1 - \sqrt{\fidelity(\rho, \sigma)}).
    \end{equation*}
Being a distance measure, the Bures distance obeys a weak triangle inequality.
\begin{proposition}[{Weak triangle inequality for Bures distance \cite[Proposition 2.1]{chailloux2015parallel}}]
	\label{lem:weak_triangle_inequality}
    Let $\rho_1, \rho_2, \rho_3$ be three quantum states, then
    \begin{equation*}
        \buresdist(\rho_1, \rho_3) \leq 2 (\buresdist(\rho_1, \rho_2) + \buresdist(\rho_2, \rho_3)).
    \end{equation*}
\end{proposition}  

\subsection{Quantum interactive protocols}

A $(2r+1)$-message quantum interactive protocol $\proto$ is specified by a quantum interactive algorithm $C$ (the ``challenger''), which interacts with an arbitrary quantum interactive algorithm $A$ (the ``adversary'').
An execution of $\proto$ consists of $r$ \emph{interactions}, each one consisting of a (quantum) message from the adversary followed by a (quantum) message from the challenger; and then a \emph{decision} after the adversary sends their final message to the challenger, wherein the challenger either accepts or rejects.  

In the following we give a detailed description of an interactive protocol and introduce the notation for registers we use throughout the paper.  A visual representation of a quantum interactive protocol can be found in \Cref{fig:quantum_interactive_argument}.

The adversary in an interactive protocol starts with an initial private register $\reg{A}_0$, and an initial ($0$-dimensional) message register $\reg{R}_{-1}$, and the challenger starts with an initial private register $\reg{W}_0$.\footnote{In the case of a non-uniform adversary, $\reg{A}_0$ may be initialized to an advice state $\ket{\aux}$, in which case $A_0$ may be taken to be the identity. For uniform adversaries we can assume that $\reg{A}_0$ is in the all-zero state.} Each interaction in an interactive protocol proceeds (without loss of generality) as follows:
\begin{enumerate}
	\item The adversary applies a unitary $A_i$ on $\reg A_i \reg R_i $ to obtain a state on registers $\reg A_{i+1} \reg M_i$ and sends the message register $\reg M_i$ to the challenger.  
	\item The challenger then performs a unitary $C_i$ to registers $\reg M_i \reg W_i$, to obtain a state over registers $\reg R_i \reg W_{i+1}$, here $\reg R_i$ represents the response register.  The challenger sends $\reg R_i$ to the adversary.
\end{enumerate}
We use $\reg{M}_i$ and $\reg{R}_i$ to distinguish the messages sent by the adversary and challenger respectively.  

In the special case where $r = 1$ (i.e., a $3$-message interactive protocol), once the adversary prepares their initial state, both the adversary and challenger have a single unitary to apply.  As a result, we drop the round index and refer to the adversary and challenger unitaries as $A$ and $C$, and the adversary's initial state as $\ket{\aux}$.  

Without loss of generality, we assume that after the adversary has sent their final message, the challenger performs some $2$-outcome measurement, described by the POVM $\{D, \id - D\}$, on $\reg M_r \reg W_r$.  
If the measurement outputs $D$, then the challenger accepts and otherwise the challenger rejects.  In this work, we focus on the \emph{soundness} of an interactive protocol.
We say that an interactive protocol has soundness $s$ if for every polynomial time (in the security parameter $\lambda$) adversary, the adversary is accepted by the challenger with probability at most $s(\lambda) + \negl(\lambda)$.  

\begin{figure}
    \centering
    \tikzset{every picture/.style={line width=0.75pt}} %

\begin{tikzpicture}[x=0.75pt,y=0.75pt,yscale=-1,xscale=1]

\draw    (63,83) -- (63,115) ;
\draw [shift={(63,117)}, rotate = 270] [color={rgb, 255:red, 0; green, 0; blue, 0 }  ][line width=0.75]    (10.93,-3.29) .. controls (6.95,-1.4) and (3.31,-0.3) .. (0,0) .. controls (3.31,0.3) and (6.95,1.4) .. (10.93,3.29)   ;
\draw    (90,130) -- (209,130) ;
\draw [shift={(211,130)}, rotate = 180] [color={rgb, 255:red, 0; green, 0; blue, 0 }  ][line width=0.75]    (10.93,-3.29) .. controls (6.95,-1.4) and (3.31,-0.3) .. (0,0) .. controls (3.31,0.3) and (6.95,1.4) .. (10.93,3.29)   ;
\draw    (244,83) -- (244,115) ;
\draw [shift={(244,117)}, rotate = 270] [color={rgb, 255:red, 0; green, 0; blue, 0 }  ][line width=0.75]    (10.93,-3.29) .. controls (6.95,-1.4) and (3.31,-0.3) .. (0,0) .. controls (3.31,0.3) and (6.95,1.4) .. (10.93,3.29)   ;
\draw    (244,146) -- (244,178) ;
\draw [shift={(244,180)}, rotate = 270] [color={rgb, 255:red, 0; green, 0; blue, 0 }  ][line width=0.75]    (10.93,-3.29) .. controls (6.95,-1.4) and (3.31,-0.3) .. (0,0) .. controls (3.31,0.3) and (6.95,1.4) .. (10.93,3.29)   ;
\draw    (208.5,195) -- (92,195) ;
\draw [shift={(90,195)}, rotate = 360] [color={rgb, 255:red, 0; green, 0; blue, 0 }  ][line width=0.75]    (10.93,-3.29) .. controls (6.95,-1.4) and (3.31,-0.3) .. (0,0) .. controls (3.31,0.3) and (6.95,1.4) .. (10.93,3.29)   ;
\draw    (63,207) -- (63,239) ;
\draw [shift={(63,241)}, rotate = 270] [color={rgb, 255:red, 0; green, 0; blue, 0 }  ][line width=0.75]    (10.93,-3.29) .. controls (6.95,-1.4) and (3.31,-0.3) .. (0,0) .. controls (3.31,0.3) and (6.95,1.4) .. (10.93,3.29)   ;
\draw    (91,256) -- (210,256) ;
\draw [shift={(212,256)}, rotate = 180] [color={rgb, 255:red, 0; green, 0; blue, 0 }  ][line width=0.75]    (10.93,-3.29) .. controls (6.95,-1.4) and (3.31,-0.3) .. (0,0) .. controls (3.31,0.3) and (6.95,1.4) .. (10.93,3.29)   ;
\draw  [draw opacity=0][fill={rgb, 255:red, 0; green, 0; blue, 0 }  ,fill opacity=1 ] (109,287.5) .. controls (109,286.12) and (110.12,285) .. (111.5,285) .. controls (112.88,285) and (114,286.12) .. (114,287.5) .. controls (114,288.88) and (112.88,290) .. (111.5,290) .. controls (110.12,290) and (109,288.88) .. (109,287.5) -- cycle ;
\draw  [draw opacity=0][fill={rgb, 255:red, 0; green, 0; blue, 0 }  ,fill opacity=1 ] (188,287.5) .. controls (188,286.12) and (189.12,285) .. (190.5,285) .. controls (191.88,285) and (193,286.12) .. (193,287.5) .. controls (193,288.88) and (191.88,290) .. (190.5,290) .. controls (189.12,290) and (188,288.88) .. (188,287.5) -- cycle ;
\draw  [draw opacity=0][fill={rgb, 255:red, 0; green, 0; blue, 0 }  ,fill opacity=1 ] (148,287.5) .. controls (148,286.12) and (149.12,285) .. (150.5,285) .. controls (151.88,285) and (153,286.12) .. (153,287.5) .. controls (153,288.88) and (151.88,290) .. (150.5,290) .. controls (149.12,290) and (148,288.88) .. (148,287.5) -- cycle ;
\draw    (91,351) -- (210,351) ;
\draw [shift={(212,351)}, rotate = 180] [color={rgb, 255:red, 0; green, 0; blue, 0 }  ][line width=0.75]    (10.93,-3.29) .. controls (6.95,-1.4) and (3.31,-0.3) .. (0,0) .. controls (3.31,0.3) and (6.95,1.4) .. (10.93,3.29)   ;
\draw    (63,303) -- (63,335) ;
\draw [shift={(63,337)}, rotate = 270] [color={rgb, 255:red, 0; green, 0; blue, 0 }  ][line width=0.75]    (10.93,-3.29) .. controls (6.95,-1.4) and (3.31,-0.3) .. (0,0) .. controls (3.31,0.3) and (6.95,1.4) .. (10.93,3.29)   ;
\draw   (221,335) -- (274,335) -- (274,367.7) -- (221,367.7) -- cycle ;
\draw  [draw opacity=0] (228.63,356.88) .. controls (231.8,349.31) and (239.28,344) .. (248,344) .. controls (257.21,344) and (265.03,349.93) .. (267.86,358.17) -- (248,365) -- cycle ; \draw   (228.63,356.88) .. controls (231.8,349.31) and (239.28,344) .. (248,344) .. controls (257.21,344) and (265.03,349.93) .. (267.86,358.17) ;  
\draw    (247,358) -- (265,344) ;

\draw (24,29) node [anchor=north west][inner sep=0.75pt]   [align=left] {Adversary};
\draw (209,29) node [anchor=north west][inner sep=0.75pt]   [align=left] {Challenger};
\draw (38,58) node [anchor=north west][inner sep=0.75pt]   [align=left] {$\displaystyle \ket{\mathsf{aux}}_{\mathsf{A_{0}}}$};
\draw (39,121) node [anchor=north west][inner sep=0.75pt]   [align=left] {$\displaystyle \mathsf{A_{1} M_{0}}$};
\draw (8,91) node [anchor=north west][inner sep=0.75pt]   [align=left] {$\displaystyle A_{0}$};
\draw (136,109) node [anchor=north west][inner sep=0.75pt]   [align=left] {$\displaystyle \mathsf{M_{0}}$};
\draw (224,58) node [anchor=north west][inner sep=0.75pt]   [align=left] {$\displaystyle \ket{0}_{\mathsf{W_{0}}}$};
\draw (222,121) node [anchor=north west][inner sep=0.75pt]   [align=left] {$\displaystyle \mathsf{W_{0} M_{0}}$};
\draw (274,154) node [anchor=north west][inner sep=0.75pt]   [align=left] {$\displaystyle C_{0}$};
\draw (222,186) node [anchor=north west][inner sep=0.75pt]   [align=left] {$\displaystyle \mathsf{W_{1} R_{0}}$};
\draw (136,174) node [anchor=north west][inner sep=0.75pt]   [align=left] {$\displaystyle \mathsf{R_{0}}$};
\draw (39,185) node [anchor=north west][inner sep=0.75pt]   [align=left] {$\displaystyle \mathsf{A_{1} R_{0}}$};
\draw (8,214) node [anchor=north west][inner sep=0.75pt]   [align=left] {$\displaystyle A_{1}$};
\draw (39,246) node [anchor=north west][inner sep=0.75pt]   [align=left] {$\displaystyle \mathsf{A_{2} M_{1}}$};
\draw (137,235) node [anchor=north west][inner sep=0.75pt]   [align=left] {$\displaystyle \mathsf{M_{1}}$};
\draw (137,330) node [anchor=north west][inner sep=0.75pt]   [align=left] {$\displaystyle \mathsf{M_{r}}$};
\draw (8,310) node [anchor=north west][inner sep=0.75pt]   [align=left] {$\displaystyle A_{r}$};
\draw (37,342) node [anchor=north west][inner sep=0.75pt]   [align=left] {$\displaystyle \mathsf{A_{r+1} M_{r}}$};
\draw (201,371) node [anchor=north west][inner sep=0.75pt]   [align=left] {$\displaystyle \{D,\ \mathsf{id} \ -\ D\}$};

\end{tikzpicture}
    \caption{A $2r+1$-message quantum interactive protocol.}
    \label{fig:quantum_interactive_argument}
\end{figure}
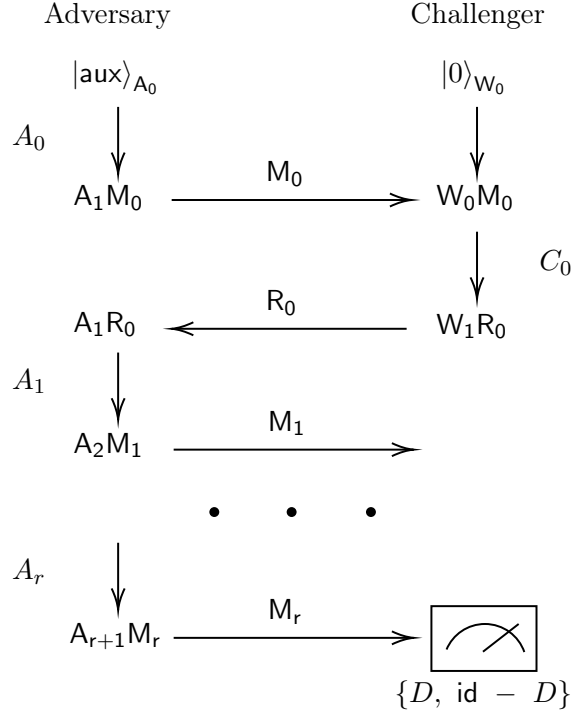

Given an interactive protocol, $\proto$, we define the \emph{$k$-fold parallel repetition} of the protocol, $\proto^{\otimes k}$ to be the protocol where the challenger and adversary execute $\proto$ $k$-times in parallel.  In $\proto^{\otimes k}$, the challenger performs the unitary $C_{i}^{\otimes k}$ on round $i$ and measures the two outcome POVM $\{D^{\otimes k}, \id - D^{\otimes k}\}$.  This means that the challenger in the $k$-fold parallel repetition only accepts if every decision POVM accepts.  We note that the adversary in the $k$-fold parallel repetition of the protocol may play a correlated, or even entangled strategy across the $k$ copies of the protocol.  

We use superscripts to denote the registers and unitaries applied by the individual repetitions of the protocol, so the first message register $\reg M_0$ of the $k$-fold parallel repetition of $\proto$ consists of $k$ many registers $\reg {M_0}^i$ where each $\reg {M_0}^i$ is sent to the $i$'th repetition of the protocol, and similarly for the registers $\reg{W_j}$, and $\reg{R_j}$.  We assume that there is only one adversary private register.

\section{Non-uniform parallel repetition of 3-message protocols}
\label{sec:non-uniform-rep}

In this section, we prove the following non-uniform version of our main theorem.

\begin{theorem}\label{cor:3_message_parallel_repetition}
    For any 3-message quantum interactive protocol $\proto$ with soundness $\epsilon$ against adversaries of size $s$, $\proto^{\otimes k}$ has soundness $\delta^k$ against adversaries of size $O\mparen{\frac{(\epsilon - \delta)\sqrt{\delta^{k - 1}}}{\log\delta} \cdot \frac sk}$ for any $k$ and $\delta > \epsilon$.

    In particular, if a family of $\proto$ has soundness $\epsilon$ against non-uniform polynomial-time adversaries, then $\proto^{\otimes k}$ has soundness $\epsilon^k$ against non-uniform polynomial-time adversaries for any polynomial $k$.
\end{theorem}

\newcommand{\Amplify}{\mathsf{Amp}}

Let $(A, \ket{\aux})$ be an adversary for $\proto^{\otimes k}$ that achieves success probability $\delta^k$.
For a collection of registers $\{\reg{C}^i\}_{i}$, let $\reg{C}$ be the concatenation of all $\reg C^i$, and $\reg{C}^{(\leq j)}$ be the concatenation of the first $j$ many $\reg C^i$.  Let the following operators be the challenger unitaries $C$ and decisions $D$ that are only concerned about the first $i$ repetitions of the protocol:
\begin{align*}
    \left(C^{(\leq i)}\right)_{\reg{M_0{^{(\leq i)}}\reg W_0^{(\leq i)}}} &= \left(\bigotimes_{j = 1}^{i} C_{\reg{M_0^j}\reg{W_0^j}}\right) \otimes \id\,,\\
    \left(D^{(\leq i)}\right)_{\reg{M_1^{(\leq i)}\reg W_1^{(\leq i)}}} &= \left(\bigotimes_{j = 1}^{i} D_{\reg{M_1^j}\reg{W_1^j}}\right) \otimes \id\,.
\end{align*}
Note that $C^{(\leq 0)} = D^{(\leq 0)} = \id$.
$D^{(<i)}$, $D^{(\neq i)}$, $D^{(> i)}$, etc. are defined similarly to be the restriction of $D^{\otimes k}$ to registers that satisfy the condition in the superscript.

Imagine a prefix of the $k$-fold game, where the $k$-fold adversary plays their strategy against all $k$ challengers but only the first $i$ repetitions perform $C$ and measure $D$.
Similar to \cite{canetti2005hardness}, our strategy for constructing an adversary for the $1$-fold game will be to find an index $i$ such that the $i$'th repetition has a high probability of accepting, conditioned on the first $i-1$ repetitions accepting, and then have the adversary ``post-select'' on the first $i-1$ repetitions accepting.
To that end, define the following projectors
\begin{equation}\label{eqn:game_projector}
    G_{i} = \left(C^{(\leq k)}\right)^{\dagger} A^{\dagger}_{\reg A_1\reg R^{(\leq k)}}\left(D^{(\leq i)}\right)A_{\reg A_1\reg R^{(\leq k)}} \left(C^{(\leq k)}\right).
\end{equation}
This further gives rise to the following POVM where we further enforce that the private registers of the verifiers is correctly initialized to $\ket{0}$.
\begin{equation*}
    \wt{G}_i = \left(\id \otimes \bra{0}_{\reg W_0^{(\le i)}}\right) G_i \left(\id \otimes \ket{0}_{\reg W_0^{(\le i)}}\right).
\end{equation*}
Crucially note that this operator only enforces the initialization of the first $i$ folds but permits arbitrary initialization of the other $(k - i)$ repetitions.
In other words, its input registers are $\reg{A_1}, \reg{R}^{(\le k)}, \reg W_0^{(> i)}$. 

Observe that $\norm{\wt G_i\ket\aux}^2$ exactly captures the probability of the success probability of winning the first $i$ repetitions when initialized with $\ket\aux$, and thus $\norm{\wt G_i}_\infty^2$ captures the maximum success probability for the first $i$ repetitions over all possible initializations.
The following corollary shows that there does always exist a ``good'' index to put the real challenger, for a particular definition of ``good'' that would suffice later for the reduction.

\begin{fact}[Discrete intermediate value theorem]
	\label{claim:p-sequence}
	Let $(p_0,...,p_k)$ be a sequence of reals such that $p_0 \le 0 \le p_k$. Then there exists an integer $1 \le i \le k$ such that $p_{i-1} \leq 0 \le p_i$.
\end{fact}
\begin{proof}
    Suppose this is not the case then we have $p_{i-1}\le 0 \implies p_i < 0$, and thus by induction, $p_k < 0$, a contradiction.
\end{proof}

\begin{corollary}\label{claim:n-fold_cases}
There exists some $1 \le i \le k$ such that $\opnorm{\wt{G}_{i - 1}}^2 \leq \delta^{i - 1}$ and $\opnorm{\wt{G}_{i}}^2 \geq \delta^{i}$.  
\end{corollary}
\begin{proof}
    Apply \Cref{claim:p-sequence} to the sequence $\left\{\opnorm{\wt{G}_{i}}^2 - \delta^{i}\right\}_{i}$.
\end{proof}

Since our reduction is non-uniform for now, we can assume that the adversary knows a good index and starts with an advice state that certifies the largeness of $\opnorm{\wt G_{i}}$.
In order to leverage the smallness of $\opnorm{\wt G_{i - 1}}$, we need to invoke the Quantum Singular Value Transform.

\begin{theorem}[{Uniform singular value amplification \cite[Theorem 17 (rephrased)]{gilyen2019quantum}}]\label{thm:uniform_singular_value_transform}
   Let $\Pi$, $\wt{\Pi}$ be projectors and $\gamma > 1$ and $\mu, \nu \in (0, \frac{1}{2})$.
   Let $\wt{\Pi}\Pi = \sum_{i} \varsigma_i \ket{w_i}\!\!\bra{v_i}$ be a singular value decomposition.  Then there is an $m = O(\frac{\gamma}{\mu}\log\left(\frac{\gamma}{\nu}\right))$ and efficiently computable $\Phi \in \mathbb{R}^{m}$ such that
   \begin{equation*}
       \left(\bra{+} \otimes \wt{\Pi}_{\leq \frac{1 - \mu}{\gamma}}\right) U_{\Phi} \left(\ket{+} \otimes \Pi_{\leq \frac{1 - \mu}{\gamma}}\right) = \sum_{i : \varsigma_i \leq \frac{1 - \mu}{\gamma}} \Tilde{\varsigma_i} \ket{w_i}\!\!\bra{v_i}, \text{ where } \left|\frac{\Tilde{\varsigma_i}}{\gamma \varsigma_i} - 1\right| \leq \nu
   \end{equation*}
   and where, for $\kappa \in [0,1]$ (using notation from \Cref{def:eigenspace-projectors}),
	\begin{equation*}
	    \Pi_{\leq \kappa} \coloneqq \sum_{i:\varsigma_i \leq \kappa} \proj{v_i} = \Pi^{\Pi\tilde{\Pi}\Pi}_{\leq \kappa^2}\, \qquad \text{and} \qquad \tilde{\Pi}_{\leq \kappa} \coloneqq \sum_{i:\varsigma_i \leq \kappa} \proj{w_i} = \Pi^{\tilde{\Pi}\Pi\tilde{\Pi}}_{\leq \kappa^2}\,.
	\end{equation*}
   Moreover $U_{\Phi}$ can be implemented using a single ancilla qubit with $m$ uses of $C_{\Pi}NOT$, $m$ uses of $C_{\wt{\Pi}}NOT$ and $m$ single qubit gates.
\end{theorem}

Note that in the theorem statement, $\varsigma_i$ are the \emph{singular values} of $\tilde{\Pi} \Pi$; the eigenvalues of $\Pi \tilde{\Pi} \Pi$, used later in the Jordan decomposition, are obtained as $p_i = \varsigma_i^2$.
Applying $U_{\Phi}$ to any pure state ``simulates'' the boosted singular value transform up to a small error $\nu$.
The following corollary makes the error more convenient later in our proof of the main theorem.

\begin{corollary}\label{cor:amplification_closeness}
    Let $\varsigma_i\ket{w_i}\!\!\bra{v_i}$, $\wt{\varsigma}_i$, $\mu$ and $\nu$ be as defined in Theorem~\ref{thm:uniform_singular_value_transform}.  For any pure state $\ket{\psi}$, let
    \begin{align*}
        \ket{\wt{\phi}} = \left(\sum_{i: \varsigma_i \leq \frac{1 - \mu}{\gamma}}\Tilde{\varsigma_i} \ket{w_i}\!\!\bra{v_i}\right) \ket{\psi},\\
        \ket{\phi} = \left(\sum_{i: \varsigma_i \leq \frac{1 - \mu}{\gamma}} \gamma \varsigma_i \ket{w_i}\!\!\bra{v_i}\right) \ket{\psi}.
    \end{align*}
    Then $\ket{\tilde\phi} = \ket\phi + \ket\nu$ where $\norm{\ket\nu}_2 \le \nu$.
\end{corollary}
\begin{proof}
    We simply verify by calculating
    \begin{align*}
        \norm{\ket\nu}_2 &= \norm{\ket{\tilde\phi} - \ket\phi}_2 \\
          &= \norm{\sum_{i: \varsigma_i \leq \frac{1 - \mu}{\gamma}}(\Tilde{\varsigma_i} - \gamma\varsigma_i) \ket{w_i}\braket{v_i|\psi}}_2 \\
          &= \sqrt{\sum_{i: \varsigma_i \leq \frac{1 - \mu}{\gamma}} \abs{(\Tilde{\varsigma_i} - \gamma\varsigma_i) \braket{v_i|\psi}}^2} \\
          &\le \max_{i: \varsigma_i \leq \frac{1 - \mu}{\gamma}} \abs{\Tilde{\varsigma_i} - \gamma\varsigma_i} \\
          &\le \max_{i: \varsigma_i \leq \frac{1 - \mu}{\gamma}} \gamma\varsigma_i\nu \\
          &\le \nu\,,
    \end{align*}
    where the second inequality is by the guarantee of the algorithm.
\end{proof}

We are now ready to state and prove the main result of this section.
In order to use this lemma in \Cref{sec:uniform-parallel-rep}, we introduce a parameter $\tau$ to account for a loss in the uniform reduction.
For the non-uniform case, it suffices to set $\tau = 1$ since we can prepare the advice state without any loss.

\begin{lemma}
    \label{theorem:non-uniform-rep}
    Let $i$ and $\ket{\aux}$ be an index and state such that $\norm{\Pi^{\wt{G}_{i-1}}_{> \delta^{i-1}} \ket{\aux}}^2 = 0$ and $\Tr\mparen{\wt{G}_{i}\proj{\aux}} \geq \tau\delta^{i}$ for some $0 < \tau \le 1$. Then for any $0 < \mu < \frac12$, there exists an adversary that takes as input $i, \mu, \delta, \tau$ and $\ket{\aux}$ and is accepted by the $1$-fold verifier with probability at least $(1 - 2\mu)^2\tau\delta$.  The adversary's running time is dominated by running the original $k$-fold prover $O\mparen{\frac i{\mu\sqrt{\delta^{i - 1}}}\log\frac1{\mu\delta\tau}}$ times.%
\end{lemma}

\begin{proof}
At a high level, the goal of the single-fold adversary will be to use \Cref{thm:uniform_singular_value_transform} to coherently do post selection such that the first $(i - 1)$ repetitions are accepted, in which case the $i$'th repetition should also be accepted with decent probability by the theorem assumption. 

Let $A$ be the unitary that the adversary performs on registers $\reg{AR}^{(\leq k)}$ in the $k$-fold game.  Apply \Cref{thm:uniform_singular_value_transform} with the following projectors and parameters.
\begin{align}
    \Pi &= \left(C^{(< i)} \otimes \id\right) \left(\id \otimes \proj{0}_{\reg{W_0}^{(<i)}}\right) \left(\left(C^{(<i)}\right)^{\dagger} \otimes \id\right)\,, \nonumber\\
    \label{def:tildePi}
    \wt{\Pi} &= (A^{\dagger}_{\reg{A_1R}^{(\leq k)}} \otimes \id)\left(D^{(<i)} \otimes \id\right) \left(A_{\reg{A_1R}^{(\leq k)}}  \otimes \id\right)\,,\\
    \label{def:gamma}
    \gamma &= \frac{1 - \mu}{\sqrt{\delta^{i - 1}}}\,,\\
    \nu &= \mu\sqrt{\tau\delta}\nonumber\,.
\end{align}
Let $W$ be the unitary satisfying the conclusions \Cref{thm:uniform_singular_value_transform} for that choice of parameters.
Note that $W$ acts on $\reg{A}$, the $k$-fold adversaries private register, $\reg{W}_1^{(< i)}$, the $(i - 1)$ many simulated challengers' private workspaces after sending the challenges, $\reg{R}^{(\leq k)}$, the response registers for all $k$ challengers (which $A$ expects to act on), and an ancilla register $\reg{P}$, which will be projected onto the $\proj{+}$ state.  Now consider the following prover for the $1$-fold game (the challenger's actions are included in monospace font to aid understanding).

\begin{longfbox}[breakable=false, padding=1em, margin-top=1em, margin-bottom=1em]
\begin{algorithm}
	\label{prot:non-uniform-adv}
    \textbf{Non-uniform adversary $\Amplify$ for the $1$-fold protocol}
\end{algorithm}
\noindent \textbf{Input: } Quantum registers $\reg{AM}_0^{(\leq k)}\reg{W}_0^{>i}$, index $i$, and slackness parameter $\mu$, and black-box oracle access to $A,C,D$.
\begin{enumerate}
		\item \label[step]{step:initw} Initialize registers $\reg{W}_0^{(< i)}$ to $\ket{0}$, the private workspace registers for the first $(i - 1)$ simulated challengers.
	    \item Send $\reg M_0^i$ to the challenger as the first message.
    \\ \texttt{(Challenger performs $C_{\reg{W}_0^i \reg{M}_0^i}$ and sends $\reg{R^i}$ back.)}
	    \item \label[step]{step:fakechallenge} Perform $C^{(\neq i)}$ on registers $\reg W_0\reg M_0^{(\neq i)}$ to get a state on $\reg W_1 \reg R^{(\neq i)}$.
		\item \label[step]{step:initauxp} Create ancilla register $\reg{P}$ initialized in $\ket{+}_{\reg{P}}$.
    \item \label[step]{step:gslw} Perform $W_{\reg{A}\reg W_1^{(<i)}\reg{R}^{(\leq k)}\reg{P}}$ as defined above.
    \item \label[step]{step:invokeadv} Perform $A_{\reg{AR}^{(\leq k)}}$.
    \item \label[step]{step:fakeverify} Measure $D^{(< i)} \otimes \proj{+}_{\reg{P}}$. If the measurement rejects, abort.  
    \item Send $\reg M_1^i$ to the challenger.
    \\ \texttt{(Challenger measures $D_{\reg{W_1^i}\reg{M_1^i}}$ and accepts or rejects.)}
\end{enumerate}
\end{longfbox}
We analyze the algorithm by describing the state of the combined prover-verifier system after every step.  

\noindent \textbf{After \cref{step:initw}:}
We assume that the registers $\reg{AM_0}^{(\leq k)}$ are initialized in the state $\ket{\aux}_{\reg{AM_0}^{(\leq k)}\reg W_0^{>i}}$ satisfying the theorem statement.  Thus, state of the adversary and challenger after \cref{step:initw} is
\begin{equation*}
    \ket{\init} := \ket{\aux}_{\reg{AM_0}^{(\leq k)}\reg W_0^{(>i)}} \otimes \ket{0}_{\reg W_0^{(\leq i)}}.
\end{equation*}

\noindent \textbf{After \cref{step:initauxp}:}
In \cref{step:fakechallenge}, together with the verifiers action, the verifier and prover perform $C^{(\leq k)}$ on registers $\reg W_0^{(\leq k)}\reg M_0^{(\leq k)}$, so the state of the system after \cref{step:initauxp} is given by
\begin{equation*}
    C^{(\leq k)}\ket{\init} \otimes \ket{+}_{\reg{P}}.
\end{equation*}
This state is in the $+1$ eigenstate of $(\Pi \otimes \proj{+}_{\reg{P}})$ as the first $(i - 1)$ repetitions are initialized correctly, so we can write the state as
\begin{equation*}
    (\Pi \otimes \ket{+}_{\reg{P}})C^{(\leq k)}\ket{\init}.
\end{equation*}

\noindent \textbf{After \cref{step:fakeverify}:}
After measuring $D^{(< i)} \otimes \proj{+}_{\reg P}$, we get the following state. 
\begin{equation*}
    (D^{(<i)} \otimes \bra{+}_{\reg{P}}) AW(\Pi \otimes \ket{+}_{\reg{P}})C^{(\leq k)}\ket{\init}.
\end{equation*}
The state might be subnormalized since we might abort in the event that the measurement rejects; aborting also ensures that the rejection part of the amplitude would not interfere with the rest of the algorithm and the analysis.
Now recall that $\wt{\Pi} = A^{\dagger} D^{(<i)} A$, thus $D^{(<i)} = A \wt{\Pi}A^\dagger$.  Performing this substitution gives us the following expression for the state above. 
\begin{equation*}
    A 
    (\wt{\Pi} \otimes \bra{+}_{\reg{P}})W(\Pi \otimes \ket{+}_{\reg{P}})C^{(\leq k)}\ket{\init}.
\end{equation*}
Note that we choose the parameter at \eqref{def:gamma} so that $\frac{1 - \mu}\gamma = \sqrt{\delta^{i - 1}}$ and thus we can now apply the guarantee of \Cref{cor:amplification_closeness} and get that the state can be written as
\begin{equation*}
    \gamma A\wt{\Pi} \Pi
    C^{(\leq k)}\ket{\init} + \ket\nu
\end{equation*}
for some $\norm{\ket\nu}_2 \le \nu$.
We now use again the fact that $C^{(\leq k)}\ket{\init}$ is invariant under $\Pi$ to remove $\Pi$.
We further plug in the definition of $\wt\Pi$, yielding
\begin{equation*}
    \gamma
    D^{(<i)}AC^{(\leq k)}\ket{\init} + \ket\nu,
\end{equation*}
which is exactly the ``post-selection'' state we would like to prepare up to a small error.

\noindent \textbf{After the final verifier decision:}
After the verifier measures $D_{\reg{W_1^i}\reg{M_1^i}}$, we obtain the state
\begin{equation*}
    \gamma
    D^{(\le i)}AC^{(\leq k)}\ket{\init} + \ket{\nu'},
\end{equation*}
where $\ket{\nu'} := D_{\reg{W_1^i}\reg{M_1^i}}\ket\nu$ which still has 2-norm at most $\nu$ as $D \preccurlyeq \id$.

By the theorem's assumption, the state $\ket{\aux}$ satisfies $\Tr(\wt{G_{i}} \proj{\aux}) \geq \tau\delta^{i}$, thus $\norm{D^{(\le i)}AC^{(\leq k)}\ket{\init}} = \norm{G_i\ket{\init}} = \norm{\wt{G_i}\ket{\aux}} \ge \sqrt{\tau\delta^i}$.
Therefore, the above state has $2$-norm at least
\begin{equation*}
    \gamma\sqrt{\tau\delta^i} - \nu = (1 - \mu)\sqrt{\tau\delta} - \mu\sqrt{\tau\delta} = (1 - 2\mu)\sqrt{\tau\delta}.
\end{equation*}  
So the prover is accepted with probability at least $(1 - 2\mu)^2\tau\delta$.
Finally the running time can be verified by plugging in the appropriate parameters and noting that $A$ is only used in \cref{step:invokeadv} and $\wt\Pi$ as defined in \cref{def:tildePi} in \cref{step:gslw}.
\end{proof}

Note that since the theorem holds for any pure state $\ket\aux$ and since the algorithm is linear, it immediately extends to mixed state inputs satisfying the same condition as well.

Finally, combining \Cref{theorem:non-uniform-rep} (instantiating $\mu = \frac{\delta - \epsilon}8$ and $\tau = 1$) with \Cref{claim:n-fold_cases}, we obtain the following non-uniform reduction.

\begin{proof}[{Proof of \Cref{cor:3_message_parallel_repetition}}]
    The first part of the theorem immediately follows by picking the parameters above.

    For the second part, assume the adversary's success probability is non-negligibly larger than $\epsilon^k$.
    Then there exists some function $\delta = \delta(\lambda)$ such that $\delta^{k} - \epsilon^{k}$ is some inverse polynomial and the adversaries success probability is at least $\delta^k$ infinitely often.
    Whenever the adversary in the $k$-fold parallel repetition of the original protocol achieves $\delta^k$, the success probability of \Cref{prot:non-uniform-adv} is at least $(1 - 2\mu)^2\delta \ge (1 - 4\mu)\delta \ge \frac{\delta + \epsilon}2$, which is larger than $\epsilon$ by a non-negligible function in $\lambda$, since $\delta - \epsilon \ge \frac{\delta^k - \epsilon^k}k$ for any real $\epsilon \le \delta \le 1 \le k$.
    Finally, the running time of this adversary has a multiplicative overhead of $\tilde O\mparen{\frac k{\mu\sqrt{\delta^{k - 1}}}} = \tilde O\mparen{\frac k{(\delta - \epsilon)\sqrt{\delta^{k - 1}}}} = \tilde O\mparen{\frac{k^2}{(\delta^k - \epsilon^k)\sqrt{\delta^{k - 1}}}} = \tilde O\mparen{k^2\delta^{-3k/2}}$; therefore, this new adversary is efficient as $\delta^k - \epsilon^k$ is inverse polynomial and $k$ is polynomial.
\end{proof}

\section{Uniform parallel repetition of 3-message protocols}
\label{sec:uniform-parallel-rep}

In this section, we prove the uniform version of our main theorem.
\begin{theorem}\label{thm:main_theorem_uniform}
    Let $\proto$ be a $3$-message quantum interactive protocol with soundness $s$ against polynomial-time (resp.\ polynomial-size) quantum adversaries.
    Then for any polynomial $k = k(\lambda)$, $\pi^{\otimes k}$ has soundness $s^{k}$ against polynomial-time (resp.\ polynomial-size) quantum adversaries.  
\end{theorem}

We do this by giving an efficient algorithm which prepares such a state from (polynomially many copies of) the adversary's initial state $\ket{\aux}$. Formally, we show the following, from which the theorem is immediate.

\newcommand{\UnifAmplify}{\textsf{Amp-U}}
\begin{lemma}
	\label{lemma:uniform-rep}
	There is a polynomial-time quantum oracle algorithm $\UnifAmplify$ with the following guarantee.
	Let $\proto = \{ C, D \}$ be a 3-message quantum interactive protocol. For $k \in \N$, $\delta \in [0,1]$, let $(A,\ket{\aux})$ be an adversary against $\proto^{\otimes k}$ which causes the challenger to accept with probability $\delta^k$. Then for every $\epsilon$ there is a $t = (\epsilon\delta^k)^{-O(1)}$ such that $(\UnifAmplify^{A,C,D}(1^{1/\delta^k},1^{1/\epsilon}),\ket{\aux}^{\otimes t})$ is an adversary against $\proto$ (i.e., a single repetition) which causes the challenger to accept with probability $\delta - \epsilon$.
\end{lemma}

\newcommand{\Subspace}{\mathcal{S}}
\newcommand{\ProjA}{\Pi_\mathsf{A}}
\newcommand{\ProjB}{\Pi_\mathsf{B}}
\newcommand{\JorKetA}[2]{\ket{v_{#1}^{#2}}}
\newcommand{\JorBraA}[2]{\bra{v_{#1}^{#2}}}
\newcommand{\JorKetB}[2]{\ket{w_{#1}^{#2}}}

\subsection{Jordan's lemma and alternating projectors}

\begin{lemma}[Jordan's lemma~\cite{jordan1875essai}]
\label{lemma:jordan}
For any two Hermitian projectors $\ProjA$ and $\ProjB$ on a Hilbert space $\reg{H}$, there exists an orthogonal decomposition of $\reg{H} = \bigoplus_j \Subspace_{j}$ (the \emph{Jordan decomposition} with respect to $\ProjA,\ProjB$) into one-dimensional and two-dimensional subspaces $\{\Subspace_{j}\}_{j}$ (the \emph{Jordan subspaces}), where each $\Subspace_{j}$ is invariant under both $\ProjA$ and $\ProjB$. Moreover:
\begin{itemize}[noitemsep]
\item in each one-dimensional space, $\ProjA$ and $\ProjB$ act as identity or rank-zero projectors; and
\item in each two-dimensional subspace $\Subspace_{j}$, $\ProjA$ and $\ProjB$ are rank-one projectors. In particular, there exist distinct orthogonal bases $\{\JorKetA{j}{1},\JorKetA{j}{0}\}$ and $\{\JorKetB{j}{1},\JorKetB{j}{0}\}$ for $\Subspace_{j}$ such that $\ProjA$ projects onto $\JorKetA{j}{1}$ and $\ProjB$ projects onto $\JorKetB{j}{1}$.
\end{itemize}
\end{lemma}

In order to unify the treatment of one- and two-dimensional subspaces, for a one-dimensional subspace $\Subspace_j = \Span(\ket{v})$ we denote $\ket{v}$ both by $\JorKetA{j}{\lambda_0}$ for $\ProjA \ket{v} = \lambda_0 \ket{v}$, and by $\JorKetB{j}{\lambda_1}$ for $\ProjB \ket{v} = \lambda_1 \ket{v}$. We define $\JorKetA{j}{1-\lambda_0}$ and $\JorKetB{j}{1-\lambda_1}$ to be the zero vector.

\begin{definition}
	\label{def:jordan-meas}
	For two Hermitian projectors $\ProjA,\ProjB$ we define the \emph{Jordan measurement} to be the projective measurement $\MeasJor[\ProjA,\ProjB] \coloneqq ( \JorProj{j} )_{j}$ with outcomes $j$, where $\JorProj{j}$ projects on to $\Subspace_j$.
	
	We define the \emph{value} of the subspace $\Subspace_j$ to be $p_j \coloneqq |\braket{v_j^1 | w_j^1}|^2$.
\end{definition}

The following straightforward but useful claim relates the Jordan decomposition with respect to $\ProjA,\ProjB$ to the spectral decompositions of $\ProjA \ProjB \ProjA$ and $\ProjB \ProjA \ProjB$.
\begin{claim}
	$\ProjA \ProjB \ProjA = \sum_j p_j \proj{v_j^1}$, and $\ProjB \ProjA \ProjB = \sum_j p_j \proj{w_j^1}$.
\end{claim}

\newcommand{\MWDist}[2]{\mathsf{MWDist}(#1,#2)}
\paragraph{Alternating projectors.}
Jordan's lemma allows us to characterize the behavior of an alternating sequence of binary projective measurements. Define the following (classical) probability distribution $\MWDist{p}{T}$ (for ``Marriott--Watrous distribution''), parameterized by $p \in [0,1]$ and $T \in \Z$:
\begin{longfbox}[breakable=false, padding=1em, margin-top=1em, margin-bottom=1em]
\noindent $\MWDist{p}{T}$:
\begin{enumerate}
    \item For each $i \in [T]$, set $a_i \coloneqq 0$ with probability $p$ and $a_i \coloneqq 1$ otherwise. 
    \item Let $b_0 \coloneqq 1$. For each $i \in [T]$, define $b_i \coloneqq b_{i-1} \oplus a_i$. 
    \item Output $b_1,b_2,\ldots,b_T$.
\end{enumerate}
\end{longfbox}

The following is a straightforward consequence of Jordan's lemma; see e.g.\ \cite{chiesa2022post} for a proof.
\begin{lemma}
\label{lemma:alt-meas-general-state}
The measurement outcomes that result from applying $T$ alternating binary projective measurements $\ProjA,\ProjB,\ProjA,\ProjB,\ldots$ to the state $\sum_j \alpha_j \JorKetB{j}{1}$ have the following distribution:
\begin{enumerate}[nolistsep]
    \item sample $p_j$ with probability  $|\alpha_j|^2$;
    \item output $\MWDist{p_j}{T}$.
\end{enumerate}
\end{lemma}
Via a Chernoff bound, we then obtain the following very useful result.
\begin{claim}
\label{claim:estimate-concentration}
For $b_0,b_1,\ldots,b_{n} \in \bits$, define
\begin{equation*}
\NumReps(b_0,b_1,\dots,b_{n}) \coloneqq \frac{|\{ j \in \{1, \ldots, n\} : b_{j-1} = b_{j} \}|}{n}.
\end{equation*}
Fix $p \in [0,1], T \in \Z$. Let $X \coloneqq \NumReps(1,b_1,\ldots,b_T)$ for $b_1,\ldots,b_T \gets \MWDist{T}{p}$. Then $\E[X] = p$, and for any $\varepsilon, \delta \in [0,1]$, if $T \geq \left\lceil \ln(\frac{1}{2\delta})/(2\varepsilon^2) \right\rceil$,
\[\Pr[\abs{X - p} \leq \varepsilon] \geq 1-\delta.\]
\end{claim}

An important consequence of the above is the existence of an efficient $(\varepsilon,\delta)$-almost projective measurement related to the Jordan decomposition.

\begin{definition}[\cite{zhandry2020schrodinger}]
\label{def:almost-proj}
A real-valued measurement $\APMeas$ is \textbf{$(\varepsilon,\delta)$-almost-projective} if applying $\APMeas$ twice in sequence to any state $\rho$ produces measurement outcomes $p,p'$ where 
\begin{equation*}
\Pr[\abs{p-p'} \leq \varepsilon] \geq 1-\delta.
\end{equation*}
\end{definition}

\begin{lemma}
	\label{lemma:effjor}
        For any $\varepsilon, \delta > 0$, and binary projective measurements $(\ProjA,\ProjB)$,
	there is an $(\varepsilon,\delta)$-almost projective measurement $\EffJor_{\varepsilon,\delta}$ which applies $\ProjA,\ProjB$ a total of $O(\frac{1}{\varepsilon^2} \log\frac1\delta)$ times, with the following properties:
	\begin{itemize}
		\item if $\Tr(\ProjA \rho) = 1$ then $\E_{p \gets \EffJor(\rho)}[p] = \Tr(\ProjB \rho) = \Tr(\ProjA \ProjB \ProjA \rho)$;
		\item if $\Pr[\EffJor(\rho) \geq p] \geq \gamma$ then $\sum_{j, p_j \geq p - \varepsilon} \JorBraA{j}{1} \rho \JorKetA{j}{1} \geq \gamma - \delta$;
		\item similarly, if $\Pr[\EffJor(\rho) \leq p] \geq \gamma$ then $\sum_{j, p_j \leq p + \varepsilon} \JorBraA{j}{1} \rho \JorKetA{j}{1} \geq \gamma - \delta$;
	\end{itemize}
\end{lemma}
\begin{proof}[Proof sketch]
	Let $\Pi'_A \coloneqq \Pi_A \otimes \proj{00}$, $\Pi'_B = \Pi_B \otimes \proj{++} + I \otimes \proj{--}$. The algorithm works as follows:
	\begin{enumerate}
		\item Measure $\Pi_A'$ on $\rho \otimes \proj{00}$; if the outcome is $0$, abort.
		\item Alternate $\Pi_B',\Pi_A'$ $T = 2\left\lceil \ln(\frac{1}{\delta})/\varepsilon^2 \right\rceil$ times, obtaining outcomes $b_1,\ldots,b_T$.
		\item Continue alternating until $\Pi_A' \to 1$, or at most $k = 3\log \frac{2}{\delta}$ times.
		\item Output $4(\NumReps(1,b_1,\ldots,b_T)-\frac14)$.
	\end{enumerate}
	If $\JorKetA{j}{1}$ is an eigenvector of $\Pi_A \Pi_B \Pi_A$ with eigenvalue $p_j$, then $\JorKetA{j}{1} \otimes \ket{00}$ is an eigenvector of $\Pi'_A \Pi'_B \Pi'_A$ with eigenvalue $\frac{1}{4} p_j + \frac{1}{4} \in [\frac14,\frac12]$. Hence the transition probability when alternating $\Pi_A',\Pi_B'$ is between $\frac14$ and $\frac12$, so the probability that after $k$ applications we have not reached $\Pi_A'$ is at most $\left(\frac34\right)^k \leq \delta/2$. Combined with the usual analysis of alternating projectors \cite{zhandry2020schrodinger,chiesa2022post}, the lemma statement follows.
\end{proof}

\newcommand{\StateTrans}{\hyperref[alg:state-trans]{\mathsf{StateTrans}}}
\subsection{State transformation for almost-projective measurements}
\label{sec:sim-jordan-alg}

In this section, we describe an algorithm for the following problem. Let $\APMeas_0,\APMeas_1$ be $(\varepsilon,\delta)$-almost projective measurements (\Cref{def:almost-proj}). Given a state $\rho$ such that $\APMeas_0(\rho) \geq \alpha$ with high probability, and a target $\beta \in [0,1]$, efficiently prepare a state $\sigma$ so that either:
	\begin{enumerate}[label=(\roman*)]
		\item \label{case:statetrans-gtr} $\APMeas_1(\sigma) \gtrsim \beta$ with high probability, or
		\item \label{case:statetrans-less} both $\APMeas_0(\sigma) \gtrsim \alpha$ \emph{and} $\APMeas_1(\sigma) < \beta$ with high probability.
	\end{enumerate}
	That is, the algorithm either converts a ``good'' state with respect to $\APMeas_0$ into a ``good'' state with respect to $\APMeas_1$, or produces a ``good'' state with respect to $\APMeas_0$ which has small overlap with \emph{any} state that is ``good'' with respect to $\APMeas_1$.

Before giving the algorithm, we set up some preliminaries. By Naimark dilation, any measurement $\APMeas = (M_q)_{q \in \bits^n}$, can be implemented as a unitary $\MeasUnitary{}$ on $\RegH \otimes \RegW$ for some ancilla register $\RegW$ initialized to zero, followed by some projective measurement $(\Pi_q)_{q \in \bits^n}$ on $\RegW$, where the $\Pi_q$ are independent of $\APMeas$. Formally, for each $q \in \bits^n$, the unitary $\MeasUnitary{}$ and projector $\Pi_q$ satisfy $M_{q} \rho M_{q}^{\dagger} = \Tr_{\RegW}(\Pi_q \MeasUnitary{} (\rho \otimes \ketbra{0}{0}_{\RegW}) \MeasUnitary{}^{\dagger})$ for all states $\rho$. By ``black-box unitary access to $\APMeas$'', we mean access to $\MeasUnitary{}$,$\MeasUnitary{}^{\dagger}$, and access to the unitary $\sum_{q,q' \in \bits^n} \Pi_q \otimes \ketbra{q' \oplus q}{q'}$.

The main result of this section is the following lemma.

\begin{lemma}
	\label{lemma:unif-alg-correctness}
	Let $\APMeas_0,\APMeas_1$ be $(\varepsilon,\delta)$-almost projective measurements on the same system $\RegH$ for some $0 < \delta < 1$. There is an algorithm $\StateTrans$ such that for every state $\rho$, and real numbers $\alpha,\beta,\gamma \in [0,1]$, $\tau \in [2\varepsilon,1-\delta]$ satisfying $\Pr[\APMeas(\rho) \geq \alpha] \geq 1-\gamma$, letting $(\sigma,c) \gets \StateTrans_{\varepsilon,\delta,\tau}[\APMeas_0,\APMeas_1,\beta](\rho)$ and $q_b \gets \APMeas_b(\sigma)$, the following hold, for $K = \lceil \frac{2}{\tau} \ln \frac{1}{\delta} \rceil$:
    \begin{enumerate}
        \item $\Pr[c = \bot] \le 4K \sqrt{\delta},$
        \item $\Pr[c = 0 \wedge q_1 < \beta - \varepsilon] \le \delta,$
        \item $\Pr[c = 1 \wedge q_0 < \alpha - 2K\varepsilon] \leq \gamma + \sqrt{\tau + \varepsilon + \delta},$ and
        \item $\Pr[c = 1 \wedge q_1 > \beta] \leq \tau + \varepsilon + \delta.$
    \end{enumerate}
	Furthermore $\StateTrans_{\epsilon, \delta, \tau}$ runs in expected time $O(K)$, given black-box unitary access to $\APMeas_0,\APMeas_1$.
\end{lemma}

That is, $\StateTrans$ takes as input a state $\rho$ and outputs a state $\sigma$ and a flag $c$ such that if $c = 0$ then $\sigma$ satisfies condition \ref{case:statetrans-gtr} and if $c = 1$ then $\sigma$ satisfies condition \ref{case:statetrans-less}. ($c = \bot$ indicates failure.)

Let $\APMeas_0,\APMeas_1$ be almost projective measurements on the same register $\RegH$, and let $\RegW_0,\RegW_1$ be corresponding ancilla registers for the Naimark dilations of $\APMeas_0$ and $\APMeas_1$ respectively. We define the projectors:
\begin{align*}
\Pi_{\geq p}^b \coloneqq \sum_{q \geq p} \MeasUnitary{b}^{\dagger} \Pi_q \MeasUnitary{b} \qquad\text{and}\qquad \Pi_{< p}^b \coloneqq \id_{\RegH \otimes \RegW_b} - \Pi_{\geq p}^b\,.
\end{align*}
We define $\tilde{\Pi}_{\geq p}^b \coloneqq \Pi_{\geq p}^b \otimes \ketbra{0}{0}_{\RegW_{1-b}}$, and define $\tilde{\Pi}_{< p}^b \coloneqq \Pi_{< p}^b \otimes \ketbra{0}{0}_{\RegW_{1-b}}$.

\begin{longfbox}[breakable=false, padding=1em, margin-top=1em, margin-bottom=1em]
    \begin{algorithm}\label{alg:state-trans}
        $\mathsf{StateTrans}_{\varepsilon, \delta, \tau}[\APMeas_0, \APMeas_1, \beta]$.
    \end{algorithm}
    \noindent \textbf{Input: } Quantum register $\reg{H}$.
    \begin{enumerate}
        \item Let $K \coloneqq \lceil \frac{2}{\tau} \ln\frac{1}{\delta} \rceil$, $T \coloneqq \lceil 1/\sqrt{\delta} \rceil$.
        For $i = 1, \ldots, K$:
    \begin{enumerate}
        \item \label[step]{step:initial-meas} Apply the $(\varepsilon,\delta)$-almost-projective measurement $\APMeas_0$, obtaining outcome $\alpha_i$. 
        \item \label[step]{step:init-w} Initialize ancilla registers $\RegW_0,\RegW_1$ to $\ket{0}$.
        \item \label[step]{step:disturb} Apply the measurement $\APMeasProjExt{< \beta}{1}$, obtaining outcome $b_1$. If $b_1 = 0$, apply $\MeasUnitary{1}$ to $(\RegH, \RegW_1)$, discard the $\RegW$ registers, and output $c = 0$ along with the $\RegH$ register.
        \item \label[step]{step:estimate} Apply the measurements $\APMeasProjExt{\geq \alpha_i - \varepsilon}{0}$, $\APMeasProjExt{< \beta}{1}$ in an alternating fashion $K-1$ times, obtaining outcomes $b_2,\ldots,b_{2K-1}$.
        \item \label[step]{step:alg-repair} Apply the measurements $\tilde{\Pi}_{\geq \alpha_i - \varepsilon}^0$, $\tilde{\Pi}_{< \beta}^1$ in an alternating fashion until $\tilde{\Pi}_{\geq \alpha_i - \varepsilon}^0 \to 1$, or until $2TK+1$ measurements have been applied in this step ($2(T+1)K$ overall). In the latter case, abort (output $\bot$).
        \item If $\NumReps(b_1,\ldots,b_{2K-1}) \geq 1-\tau$, discard the $\RegW$ registers and output $c = 1$ along with the $\RegH$ register.
        \item Apply $\MeasUnitary{0}$ to $(\RegH, \RegW_0)$, and discard the $\RegW$ registers.
        \end{enumerate}
        \item If the algorithm does not terminate with output above, abort (output $\bot$).
    \end{enumerate}
\end{longfbox}

In the proof we will make use of the following lemma, due to \cite{lombardi2022post}:
\newcommand{\DMatrix}{\rho}
\begin{lemma}[Pseudoinverse lemma]
\label{lemma:pseudoinverse}
    Let $\ProjA,\ProjB$ be projectors, and $(\JorProj{j})_j \coloneqq \MeasJor[\ProjA,\ProjB]$ the corresponding Jordan measurement. Let $\DMatrix$ be a state such that $\ProjA \DMatrix = \DMatrix \ProjA$ and $\Tr(\ProjA \DMatrix) \geq 1 - \gamma$, and let $\Pi_0 \coloneqq \sum_{j, p_j = 0} \JorProj{j}$. Let $E \coloneqq \sum_{j, p_j > 0} \frac{1}{p_j} \JorProj{j}$. There exists a ``pseudoinverse'' state $\sigma$ with $\Tr(\ProjB \sigma) = 1$ such that all of the following are true:
    \begin{enumerate}[noitemsep]
        \item $\Tr(\ProjA \sigma) = \frac{1-\Tr(\Pi_0 \DMatrix)}{\Tr(E \DMatrix)} $,
        \item $\td\mparen{\DMatrix,\frac{\ProjA \sigma \ProjA}{\Tr(\ProjA \sigma)}} \leq \sqrt{\Tr(\Pi_0 \DMatrix)}$,
        \item for all $j$ such that $p_j > 0$ it holds that $\Tr(\JorProj{j} \sigma) = \frac{\Tr(\JorProj{j} \DMatrix)}{p_j \cdot \Tr(E \DMatrix)}$, and
        \item $\Tr(\Pi_0 \sigma) = 0$.
    \end{enumerate}
\end{lemma}

\begin{claim}
	\label{claim:rotate-to-subspace}
	Let $\rho$ be such that $\ProjA \rho = \rho \ProjA$ and $\Tr(\Pi_0 \rho) = 0$. Then there exists a state $\rho'$ such that $\Tr(E\rho') = \Tr(E\rho)$, $\Tr(\ProjA \rho') = 1$, and $\td(\rho,\rho') \leq 1 - \Tr(\ProjA \rho)$.
\end{claim}
\begin{proof}
	This is a variant of \cite[Claim 7.2]{lombardi2022post}. Let
	\[
		U \coloneqq \sum_{j,p_j \notin \{0,1\}} (\JorKetA{j}{1}\JorBraA{j}{0} + \JorKetA{j}{0}\JorBraA{j}{1}) + \sum_{j,p_j \in \{0,1\}} \JorKetA{j}{p_j} \JorBraA{j}{p_j}
	\]
	and let $\rho' \coloneqq \ProjA \rho + U(I-\ProjA)\rho U^\dagger$. Then $\td(\rho,\rho') = \frac12 \norm{(I-\ProjA)\rho - U(I-\ProjA)\rho U^\dagger}_1 \leq \Tr((I-\ProjA)\rho)$.
\end{proof}

\begin{claim}
	\label{claim:single-iteration}
	For all $\epsilon$, $\delta$, $\tau$ and all initial states in $\RegH$ at \cref{step:initial-meas}, every time the $\StateTrans_{\epsilon, \delta, \tau}$ reaches \cref{step:alg-repair}, it aborts with probability at most $2\sqrt{\delta}$ at that step and makes at most $9K$ measurements in expectation.  Moreover, for all $1 < i \leq K$, $\Pr[b_1 = 1 \text{ in iteration $i-1$ } \wedge \alpha_i < \alpha_{i-1} - 2\varepsilon] \le 3\sqrt{\delta}$.
	\end{claim}
	\begin{proof}
	This follows via an analysis similar to \cite[Lemma 4.10]{chiesa2022post}, modified to account for the additional measurements in \cref{step:estimate}.
	
	First note that, because $\APMeas_0$ is $(\varepsilon,\delta)$-almost projective, the probability that a measurement of $\APMeasProjExt{\geq \alpha_i - \varepsilon}{0}$ accepts immediately after \cref{step:init-w} is $1 - \delta$.
 Hence by \Cref{lemma:gentle-measurement}, post-selecting on $\APMeasProjExt{\geq \alpha_i - \varepsilon}{0} \to 1$ at this point disturbs the state by at most $\sqrt{\delta}$ in trace distance.
 Henceforth we therefore assume that the state immediately after \cref{step:init-w} is some $\ket{\phi} \in \img(\APMeasProjExt{\geq \alpha_i - \varepsilon}{0})$.
 (We without loss of generality argue this for all pure states $\ket\phi \in \img(\APMeasProjExt{\geq \alpha_i - \varepsilon}{0})$ and by linearity the argument also generalizes to mixed states.)
 This changes the probability of any event by at most an additive $\sqrt{\delta}$ and the expected running time by at most an additive $\sqrt{\delta} (2TK+1) \le \sqrt\delta(4K/\sqrt\delta + 1) \le 5K$.
	
	We first analyse the distribution of measurement outcomes in \cref{step:alg-repair}.
	Indeed, let $b_{2K},b_{2K+1},\ldots,b_{2TK}$ be the outcomes of those measurements. By \Cref{lemma:alt-meas-general-state}, we have that $b_1,\ldots,b_{2TK} \sim \sum_j a_j \MWDist{p_j}{2TK}$ for some $a_j \in [0,1], \sum_j a_j = 1$.
	
	Fix any $p \in [0,1]$, and consider $c_1,\ldots,c_{2TK} \gets \MWDist{p}{2TK}$. Let $r \coloneqq \Pr[c_{2K} = 0]$. Note that the variables $c_2,c_4,\ldots,c_{2TK}$ form a Markov chain with symmetric transition probabilities and initial state $c_0 = 1$. Hence, we have that $\Pr[c_{2iK} = 1 \mid c_{2(i-1)K} = 0] = r$ for all $2 \leq i \leq T$. It follows that
	\begin{equation}
		\label{eq:return-prob}
		\Pr\mbracket{\bigwedge_{i=1}^T c_{2iK} = 0} = \Pr[c_{2K} = 0] \prod_{i=2}^T \Pr[c_{2iK} = 0 \mid c_{2(i-1)K} = 0] = r (1-r)^{T-1} \le \frac{1}{T},
	\end{equation}
    where the last inequality follows by finding that LHS is maximized when $r = \frac1T$, giving maximum $\frac1T \cdot \paren{1 - \frac1T}^{T-1} \le \frac1T$.
	Next, let $D'$ be a random variable corresponding to the smallest $i \geq 1$ such that $c_{2iK} = 1$. Then
	\begin{equation}
		\label{eq:return-time}
		\E[D'] = \sum_{t=1}^{T} \Pr[D' \geq t] \leq 1 + r \sum_{t=0}^{\infty} (1 - r)^t = 2\,.
	\end{equation}
 
	Since the distribution of $b_1,\ldots,b_{2TK}$ is a convex combination of $\MWDist{p_j}{2TK}$, it holds by convexity and \eqref{eq:return-prob} that $\Pr\mbracket{\bigwedge_{i=1}^T b_{2iK} = 0} \leq 1/T$. If (in particular) $b_{2iK} = 1$ for any $1 \leq i \leq T$, then $\StateTrans$ does not abort; hence $A$ aborts with probability at most $1/T \le \sqrt{\delta}$. Next, let $D$ be the number of measurements applied in \cref{step:alg-repair}. By convexity, linearity of expectation and \Cref{eq:return-time}, $\E[D] \leq 4K$.
	
	Finally we bound the probability that $\alpha_i < \alpha_{i-1} - 2\varepsilon$.
 The post-measurement state after \cref{step:disturb}, conditioned on $b_1 = 1$, is $\APMeasProjExt{< \beta}{1} \ket{\phi}/\sqrt{q}$, where $q = \Pr[b_1 = 1] = \norm{\APMeasProjExt{< \beta}{1} \ket{\phi}}^2 = \sum_{j} q_j p_j$ for $q_j \in [0,1]$ such that $\sum_j q_j = 1$.
 $q > 0$ since otherwise we would reach \cref{step:disturb} with probability 0.
 Observe that $\norm{\JorProj{j} \ket{\phi}/\sqrt{q}}^2 = q_j p_j/q$ for each $j$.
	Let $\rho$ be the reduced density matrix on register $\reg H, \reg W_0, \reg W_1$ after \cref{step:alg-repair} (again conditioned on $b_1 = 1$). Let $\rho'$ be the state guaranteed by \Cref{claim:rotate-to-subspace}; note that $\td(\rho,\rho') \leq \sqrt{\delta}$ by \eqref{eq:return-prob}.
 Since $\JorProj{j}$ commutes with both $\APMeasProjExt{\geq \alpha_i - \varepsilon}{0}$ and $\APMeasProjExt{< \beta}{1}$, $\Tr(\JorProj{j} \rho) = q_j p_j/q$, so
 \begin{equation}
    \label{eq:uniform-trace-bound}
     \Tr(E \rho') = \Tr(E \rho) = \sum_{j, p_j > 0} \frac{q_j}q \leq \frac1q
 \end{equation}
 and $\Tr(\Pi_0 \rho') = \Tr(\Pi_0 \rho) = 0$ for $E,\Pi_0$ as defined in \Cref{lemma:pseudoinverse}.
	
    By \Cref{lemma:pseudoinverse} there is a pseudoinverse state $\sigma \in \img(\APMeasProjExt{< \beta}{1})$ with $\rho' = \APMeasProjExt{\geq \alpha_i - \varepsilon}{0} \sigma \APMeasProjExt{\geq \alpha_i - \varepsilon}{0}/q'$ for $q' = \frac{1}{\Tr(E \rho)} \geq q$ by \eqref{eq:uniform-trace-bound}. Since $\img(\APMeasProjExt{< \beta}{1}) \subseteq \img(I_{\RegH} \otimes \proj{0}_{\RegW_0} \otimes I_{\RegW_1})$, $\sigma = \sigma'_{\RegH,\RegW_1} \otimes \proj{0}_{\RegW_0}$ for some $\sigma'$. Hence $\Tr_{\RegW}(\MeasUnitary{0} \rho' \MeasUnitary{0}^{\dagger})$ is precisely the post-measurement state after applying $\APMeas_0$ to $\sigma'$ and post-selecting on obtaining an answer greater than $\alpha_i - \varepsilon$. It follows that $\Pr[\alpha_i < \alpha_{i-1} - 2\varepsilon \mid b_1 = 1] = \Pr[p' < \alpha_{i-1} - 2\varepsilon \mid p \geq \alpha_i - \varepsilon]$ where $p,p'$ are the results of applying $\APMeas_0$ twice in sequence to $\sigma$. By the definition of conditional probability,
    \begin{equation*}
    	\Pr[p' < \alpha_{i-1} - 2\varepsilon \mid p \geq \alpha_i - \varepsilon] = \frac{\Pr[p' < \alpha_{i-1} - 2\varepsilon \wedge p \geq \alpha_i - \varepsilon]}{\Pr[p \geq \alpha_i - \varepsilon]} \leq \frac{\delta}{q}
    \end{equation*}
    since $\Pr[p \geq \alpha_i - \varepsilon] = q' \geq q$ and $\APMeas_0$ is $(\varepsilon,\delta)$-almost projective. The claim follows since $\Pr[b_1 = 1] = q$.
\end{proof}

	\begin{claim}\label{claim:state_trans_abort_probability}
		$\StateTrans$ aborts with probability at most $4K \sqrt{\delta}$.
	\end{claim}
	\begin{proof}
    We show that for each $i$, the probability that the algorithm stops (outputting either $0$ or $1$) in the $i$-th iteration is at least $\tau$. 
    Consider the following experiment:
    
    \begin{longfbox}[breakable=false, padding=1em, margin-top=1em, margin-bottom=1em]
    \textbf{$\StateTrans$ $i$-th round abort experiment}
    \begin{enumerate}[noitemsep]
        \item Initialize ancilla registers $\RegW_0,\RegW_1$ to $\ket{0}$.
        \item[$\star$.] Apply the Jordan subspace measurement, obtaining a subspace label $j$.
        \item Apply the measurement $\APMeasProjExt{< \beta}{1}$, obtaining outcome $b_1$. If $b_1 = 0$, output \textsc{Yes}.
        \item Apply the measurements $\APMeasProjExt{\geq \alpha_i - \varepsilon}{0}$, $\APMeasProjExt{< \beta}{1}$ in an alternating fashion $K-1$ times, obtaining outcomes $b_2,\ldots,b_K$. If $\NumReps(b_1,\ldots,b_K) \geq 1-\tau$, output \textsc{Yes}.
        \item Otherwise, output \textsc{No}.
    \end{enumerate}
    \end{longfbox}
    Observe that $\StateTrans$ stops at the $i$-th iteration without aborting if and only if this experiment, without step $\star$, outputs $\textsc{Yes}$. Since step $\star$ commutes with $\APMeasProjExt{\geq \alpha_i - \varepsilon}{0}$ and $\APMeasProjExt{< \beta}{1}$, inserting $\star$ does not change the outcome probabilities. Note that given outcome $j$ from step $\star$, the probability that step 2 outputs \textsc{Yes} is $1-p_j$, and that $\NumReps(b_1,\ldots,b_K) \sim \Bin(K,p_j)/K$. Hence there are two cases:
    \begin{itemize}[nolistsep]
        \item if $p_j < 1 - \tau + \varepsilon$, the probability that step $2$ outputs \textsc{Yes} is at least $\tau - \varepsilon$;
        \item if $p_j \geq 1 - \tau + \varepsilon$, the probability that step $3$ outputs \textsc{Yes} is at least $1-\delta \geq \tau-\varepsilon$.
    \end{itemize}
    Since this lower bound holds regardless of the initial state at iteration $i$, it follows that the probability that the procedure fails to terminate with output within $K$ steps is at most $(1-(\tau-\varepsilon))^K \le (1 - \tau/2)^K \leq e^{-\tau K/2} \leq \delta$.

    Then the overall probability, using union bound and \Cref{claim:single-iteration}, is at most $3K\sqrt\delta + \delta \le 4K\sqrt\delta$.
\end{proof}

Next we show that, for any $\alpha,\beta \in [0,1]$, a state with large overlap with both $\APMeasProjExt{\geq \alpha}{0}$ and $\APMeasProjExt{< \beta}{1}$ is likely to return an outcome $\geq \alpha$ when measured with $\APMeas_0$ and $< \beta$ when measured with $\APMeas_1$. (Note this is not trivial due to the presence of ancillas.) 

\begin{claim}
    	\label{claim:overlapping-projectors}
    	For any state $\rho$ satisfying $\Tr(\APMeasProjExt{\geq \alpha}{0} \rho) = 1$ and $\Tr(\APMeasProjExt{< \beta}{1}\rho) \geq 1 - \gamma$, it holds that
    	\begin{equation*}
    		\Pr[\APMeas_0(\rho') \geq \alpha] \geq 1 - \sqrt{\gamma} \qquad \text{and} \qquad \Pr[\APMeas_1(\rho') < \beta] \geq 1 - \gamma\,,
    	\end{equation*}
    	where $\rho' \coloneqq \Tr_{\RegW_0,\RegW_1}(\rho)$.
    \end{claim}
    \begin{proof}
    	We have that $\Pr[\APMeas_1(\rho') < \beta] = \Tr(\Pi_{< \beta}^1 (\rho' \otimes \ketbra{0}{0}_{\RegW_1}))$. Since $\Tr(\APMeasProjExt{\geq \alpha}{0} \rho) = 1$, $\rho' \otimes \ketbra{0}{0}_{\RegW_1} = \Tr_{\RegW_0}(\rho)$, and so $\Pr[\APMeas_1(\rho') < \beta] = \Tr(\APMeasProjExt{< \beta}{1} \rho) = 1 - \gamma$.
    
    	Similarly, $\Pr[\APMeas_0(\rho') \geq \alpha] = \Tr(\APMeasProjExt{\geq \alpha}{0} (\rho' \otimes \ketbra{0}{0}_{\RegW_0}))$. Since $\Tr(\APMeasProjExt{< \beta}{1} \rho) \geq 1 - \gamma$, the states $\Tr_{\RegW_1}(\rho)$ and $\rho' \otimes \ketbra00_{\RegW_0}$ are $\sqrt{\gamma}$-close in trace distance by \Cref{lemma:gentle-measurement}. Hence $\Pr[\APMeas_0(\rho') \geq \alpha] \geq \Tr(\Pi_{\geq \alpha}^0 \rho) - \sqrt{\gamma} = 1 - \sqrt{\gamma}$.
\end{proof}
We use this claim to prove that the algorithm outputs the correct state when it does not abort.  
\begin{claim}\label{claim:state_trans_correctness_non_abort}
        For all quantum states $\sigma$ and $\varepsilon, \delta, \tau$, let $(c, \rho) \gets \StateTrans_{\varepsilon, \delta, \tau}[\APMeas_0, \APMeas_1, \beta](\sigma)$ be the output of the algorithm.  Then the following hold, where $i$ is (a random variable corresponding to) the last iteration of the algorithm:
        \begin{enumerate}
        \item $\Pr[c = 0 \wedge \APMeas_1(\rho) < \beta - \varepsilon] \leq \delta,$
        \item $\Pr[c = 1 \wedge \APMeas_0(\rho) < \alpha_i - \varepsilon] \leq \sqrt{\tau + \varepsilon + \delta}$, and
        \item $\Pr[c = 1 \wedge \APMeas_1(\rho) \geq \beta] \leq \tau + \varepsilon + \delta.$
        \end{enumerate}
    \end{claim}
	\begin{proof}
	By the description of $\StateTrans$, $c = 0$ when the measurement $\APMeasProjExt{<\beta}{1}$ yields outcome $0$ in \cref{step:disturb}. By the definition of an almost projective measurement, the probability that applying $\APMeas_1$ subsequently yields outcome $<\beta - \varepsilon$ is less than $\delta$, which proves the first inequality.  
	
	Consider inserting the Jordan subspace measurement, $\MeasJor$, after \cref{step:alg-repair}, obtaining outcome $j$.
 Since $\MeasJor$ commutes with $\APMeasProjExt{\geq \alpha - \varepsilon}{0}$ and $\APMeasProjExt{< \beta}{1}$, we can equivalently insert the Jordan subspace measurement before \Cref{step:estimate}.
 Suppose that the outcome $j$ satisfies $p_j < 1 - \tau - \varepsilon$; conditioned on receiving this outcome, by \Cref{claim:estimate-concentration}, $\Pr[\NumReps(b_1,\ldots,b_{2K-1}) \geq 1 - \tau] \leq \delta$.

 By the definition of conditional probability we have that
    \begin{equation}\label{eqn:num_reps_prob_upper_bound}
    	\Pr[\NumReps(b_1,\ldots,b_{2K-1}) \geq 1 - \tau \wedge p_j < 1 - \tau - \varepsilon] \leq \delta\,.
    \end{equation}
    
    Let $\rho$ be the output state of $\StateTrans$, conditioned on halting with output $c = 1$; let $q \coloneqq \Pr[c = 1]$. Then by construction, $\Tr(\APMeasProjExt{\geq \alpha_i - \varepsilon}{0} \rho) = 1$, and by \cref{eqn:num_reps_prob_upper_bound}, $\Tr(\sum_{j, p_j < 1 - \tau - \varepsilon} \JorProj{j} \rho) \leq \delta/q$. Then $\Tr(\APMeasProjExt{< \beta}{1} \rho) \ge (1 - \tau - \varepsilon)(1 - \delta/q) \geq 1 - \tau - \varepsilon - \delta/q$.
    It follows from \Cref{claim:overlapping-projectors} that
		\begin{equation*}
    		\Pr[\APMeas_0(\rho) \geq \alpha_i - \varepsilon] \geq 1 - \sqrt{\tau + \varepsilon + \delta/q} \text{ and } \Pr[\APMeas_1(\rho) < \beta] \geq 1 - \tau - \varepsilon - \delta/q~.
    	\end{equation*}
    	The second two inequalities in the claim then follow by definition of the conditional  probability and since $q \le 1$.
    \end{proof}

We are now ready to put together the proof of \Cref{lemma:unif-alg-correctness}.
    
\begin{proof}[Proof of \Cref{lemma:unif-alg-correctness}]
	\Cref{claim:state_trans_abort_probability} shows that $\Pr[c = \bot] \le 4K\sqrt{\delta}$.
 \Cref{claim:state_trans_correctness_non_abort} shows the correctness of the algorithm when the algorithm does not abort.
 \Cref{claim:single-iteration} shows that $\StateTrans$ makes $O(K)$ measurements in expectation.
\end{proof}

\subsection{Proof of \texorpdfstring{\Cref{lemma:uniform-rep}}{Lemma 5.2}}

We are now ready to prove \Cref{lemma:uniform-rep}.
First recall the definition of $G_i$ from \cref{eqn:game_projector}:

\begin{equation*}
    G_{i} = (C^{\otimes k})^{\dagger} A^{\dagger} \left(D^{(\leq i)}\right) A \cdot C^{\otimes k}.
\end{equation*}
Let $\epsilon_0 \coloneqq \epsilon \delta^k/4$, $\tau \coloneqq \epsilon_0^2/100, \hat{\delta} \coloneqq \frac{\tau^6}{8k^3}, \hat{\varepsilon} \coloneqq \frac{\epsilon_0^2 \tau}{10\log 1/\delta}$. For all $i \in [k]$, define $\APMeas_i \coloneqq \EffJor_{\hat{\varepsilon},\hat{\delta}}[\proj{0}_{\RegW_{0}^{\leq i}},G_i]$. The algorithm $\UnifAmplify$ is defined as follows.

\begin{longfbox}[breakable=false, padding=1em, margin-top=1em, margin-bottom=1em]
\begin{algorithm}\label{alg:single-fold-adversary}
    Uniform adversary $\UnifAmplify$ for the $1$-fold protocol.
\end{algorithm}
\noindent \textbf{Input: } registers $(\reg{A_j})_{1 \leq j \leq t}$, unitary oracles $A,C,D$, $\delta,\epsilon \in [0,1], k \in \N$.
\begin{enumerate}
	\item For $j = 1, \ldots, t$, initialise a fresh register $\reg W_{0,j},\reg M_{0,j}$ to $\ket{0}$ and measure $\reg{A}_j,\reg W_{0,j},\reg M_{0,j}$ with $\APMeas_k$, obtaining outcome $\gamma_j$.
	\item \label[step]{step:choose-j} Let $j$ be such that $\gamma_j \geq \delta^k - \hat{\varepsilon}$; if no such $j$ exists, abort. Set $\reg{A} \coloneqq \reg{A}_j, \reg{W}_0 \coloneqq \reg{W}_{0,j}$.
	\item \label[step]{step:unif-transform-state} For $i = k, \ldots, 2$, apply $\StateTrans_{\hat{\varepsilon},\hat{\delta},\tau}[\APMeas_{i},\APMeas_{i-1},\delta^{i-1}-\hat{\varepsilon}]$ to $(\reg{A},\RegW_0)$, obtaining outcome $c_i$. If $c_i = \bot$, abort. If $c_i = 1$, stop and proceed to the next step.
	\item If $c_i = 0$ for all $i = 2, \ldots, k$, set $c_1 \coloneqq 1$. Let $i^*$ be the (unique) index for which $c_{i^*} = 1$.
	\item Run $\Amplify_{\tau,\mu}$ (\Cref{prot:non-uniform-adv}) on input $i^*, \delta$ and registers $\reg A \reg W_{0}^{\geq i+1} \reg M_0^{\leq k}$.
\end{enumerate}
\end{longfbox}
We first bound the probability that \Cref{alg:single-fold-adversary} aborts.
By \Cref{lemma:effjor} and assumption, for all $j$, $\E[p_j] = \Tr(\wt{G}_k \rho_j) \geq \delta^k$.
Since $p_j \le 1$, by an averaging argument we get $\Pr[p_j \geq \delta^k - \hat{\varepsilon}] \ge \frac{\hat\varepsilon}{1 - (\delta^k - \hat\varepsilon)} \geq \hat{\varepsilon}$, and so the probability that \Cref{alg:single-fold-adversary} aborts in \cref{step:choose-j} is at most $(1-\hat{\varepsilon})^t \leq \epsilon_0/10$ if $t =\frac1{\hat\varepsilon}\ln\frac{10}{\epsilon_0}$.
By \Cref{lemma:unif-alg-correctness}, the probability that $c_i = \bot$ for any $i$ is at most $4kK\sqrt{\hat\delta} = 4k\sqrt{\hat\delta} \cdot \lceil\frac2\tau\ln\frac1{\hat\delta}\rceil \le \frac{9k}\tau \cdot \sqrt{\hat\delta}\ln\frac1{\hat\delta} \le \frac{20k}\tau \cdot \sqrt[3]{\hat\delta} = \epsilon_0/10$.

By \Cref{lemma:unif-alg-correctness}, the state $\rho$ at the beginning of iteration $i^*$ of \Cref{step:unif-transform-state} has $\Pr[\APMeas_{i^*}(\rho) < \delta^{i^*} - 2\hat{\varepsilon}] = O(K\delta) \leq \epsilon_0/10$.
Then, again by \Cref{lemma:unif-alg-correctness}, the state $\rho'$ at the end of iteration $i^*$ of \Cref{step:unif-transform-state} has $\APMeas_{i^*}(\rho') \geq \delta^{i} - 2K\hat{\varepsilon}$ with probability $1 - \epsilon_0/10 - \sqrt{\tau + \hat{\varepsilon} + \hat{\delta}} \geq 1 - \epsilon_0/5$.
It then follows by \Cref{lemma:effjor} that $\Tr(\wt{G}_{i^*} \rho') \geq \delta^{i^*} - 2K\hat{\varepsilon} - \epsilon_0/5 \geq \delta^{i^*} - \epsilon_0/4$. We also have that $\APMeas_{i^*-1}(\rho') \leq \delta^{i^*-1} - \hat{\varepsilon}$ with probability $1 - \tau - \hat{\varepsilon} - \hat{\delta}$, from which it follows that $\Tr(\HProj{\wt{G}_{i^*-1}}{> \delta^{i^*-1}} \rho') \leq \tau + \hat{\varepsilon} + 2\hat{\delta} \leq \epsilon_0^2/16$.

Let $\sigma \coloneqq \HProj{\wt{G}_{i^*-1}}{> \delta^{i^*-1}} \rho' \HProj{\wt{G}_{i^*-1}}{> \delta^{i^*-1}} / \Tr(\HProj{\wt{G}_{i^*-1}}{> \delta^{i^*-1}})$. By gentle measurement, $\td(\rho',\sigma) \leq \epsilon_0/4$.
By \Cref{theorem:non-uniform-rep}, applying $\Amplify$ (\Cref{prot:non-uniform-adv}) to $\sigma$ with parameters $\tau = 1-\frac{\epsilon_0}{\delta^{i^*}} \geq 1-\epsilon/4$, $\mu = \epsilon/4$ yields an adversary that succeeds with probability at least $(1-\epsilon/4)^3 \delta \geq (1 - \epsilon)\delta \geq \delta - \epsilon$, and runs in time $\poly(k,\frac{1}{\delta^k},\frac{1}{\epsilon})$.

\section{Barriers to parallel repetition beyond \texorpdfstring{$3$}{3}-message protocols}
\label{sec:failure}

In this section we show that, under a cryptographic assumption, for every $k$ there exists a constant round quantum interactive protocol such that the $k$-fold parallel repetition of the protocol has the \emph{same} soundness as the original protocol.  Before introducing the interactive quantum protocol, we need to define post-quantum bit commitment schemes and non-malleability.

\subsection{Post-quantum bit commitments}

We begin by describing post-quantum bit commitments similar to how they are described in \cite{chia2021black}.  At a high level, interactive quantum bit commitments are quantum interactive protocols executed between a sender and receiver.  Post-quantum bit commitments describe a special form of general quantum bit commitments where the protocol is entirely classical but security holds against quantum adversaries.  Since the messages are entirely classical we will not use the quantum register notation used in the rest of the paper for the sake of simplicity.  To talk about the security of post-quantum bit commitments, we must first define computational indistinguishability.  We say that two families of random variables, indexed by $\lambda$, $\mathcal{X} = \{\mathcal{X}_{\lambda}\}$ and $\mathcal{Y} = \{\mathcal{Y}_{\lambda}\}$ are \emph{computationally indistinguishable}, denoted $\mathcal{X} \approx_{c} \mathcal{Y}$ if for all polynomial-time quantum adversaries $(A_{\lambda})_{\lambda}$, 
\begin{equation*}
   \left|\Pr_{x \sim \mathcal{X}_{\lambda}}[A_{\lambda}(x)\text{ accepts}] - \Pr_{y \sim \mathcal{Y}_{\lambda}}[A_{\lambda}(y)\text{ accepts}] \right| \leq \negl(\lambda)\,.
\end{equation*}

For a post-quantum commitment scheme between a sender $A = (A_{\lambda})_{\lambda}$ and receiver $B = (B_{\lambda})_{\lambda}$, denote by $\OUT_{S}(m, A, B)$ and $\OUT_{R}(m, A, B)$ the random variables corresponding to the classical strings held in the private workspace registers at the end of the commit phase when the sender is committing to the message $m$, where the randomness is over the internal randomness of the algorithm.  Let $\tau(m, A, B)$ be the public transcript (i.e. an ordered list of messages sent).  We now define the hiding and binding properties of post-quantum commitments.

\begin{definition}[Hiding property of post-quantum bit commitments]
    Let $A$ be the algorithm that performs an honest execution of the sender in a post-quantum bit commitment.  The post-quantum bit commitment is computationally hiding if for all polynomial-time quantum adversaries $B = (B_{\lambda})_{\lambda}$, 
    \begin{equation*}
        \OUT_{R}(0, A, B) \approx_{c} \OUT_{R}(1, A, B)\,.
    \end{equation*}
\end{definition}

At a high level, the hiding property means that any bounded receiver can not tell if they are receiving a commitment to $0$ or $1$.  

\begin{definition}[Binding property of post-quantum commitments]
    Let $\OUT_{R}^{b}$ and $\tau^{b}$ be the private receiver workspace and public transcript after an execution of a post-quantum commitment scheme where the sender committed to $b$.  Let $\ACCEPT_{\lambda}(\OUT_{R}^{b}, \tau^{b}, A)$ be the bit corresponding to whether or an honest receiver accepts in the reveal phase when the sender acts according to the algorithm $A$, conditioned on the receiver having private workspace distributed as $\OUT_{R}^{b}$, and public transcript $\tau^{b}$ in the commit phase, and $\REVEAL(\OUT_{R}^{b}, \tau^{b}, A)$ be the bit that is revealed to the receiver in the same context.  The post-quantum commitment scheme is \emph{binding} if for all polynomial time non-uniform quantum adversaries $A$, 
    \begin{equation}
        \Pr[\REVEAL(\OUT_{R}^{b}, \tau^{b}, A) = 1 - b \text{ and } \ACCEPT(\OUT_{R}^{b}, \tau^{b}, A) = 1] \leq \negl(\lambda)\,.
    \end{equation}
\end{definition}

At a high level, the binding property means that a bounded adversary can not switch their commitment from $b$ to $1-b$ after the commit phase.

\begin{figure}
    \centering
    \tikzset{every picture/.style={line width=0.75pt}} %

\begin{tikzpicture}[x=0.75pt,y=0.75pt,yscale=-1,xscale=1]

\draw    (77,126) -- (201,126) ;
\draw [shift={(203,126)}, rotate = 180] [color={rgb, 255:red, 0; green, 0; blue, 0 }  ][line width=0.75]    (10.93,-3.29) .. controls (6.95,-1.4) and (3.31,-0.3) .. (0,0) .. controls (3.31,0.3) and (6.95,1.4) .. (10.93,3.29)   ;
\draw    (203,171) -- (80,171) ;
\draw [shift={(78,171)}, rotate = 360] [color={rgb, 255:red, 0; green, 0; blue, 0 }  ][line width=0.75]    (10.93,-3.29) .. controls (6.95,-1.4) and (3.31,-0.3) .. (0,0) .. controls (3.31,0.3) and (6.95,1.4) .. (10.93,3.29)   ;
\draw  [fill={rgb, 255:red, 0; green, 0; blue, 0 }  ,fill opacity=1 ] (89,215) .. controls (89,213.9) and (89.9,213) .. (91,213) .. controls (92.1,213) and (93,213.9) .. (93,215) .. controls (93,216.1) and (92.1,217) .. (91,217) .. controls (89.9,217) and (89,216.1) .. (89,215) -- cycle ;
\draw  [fill={rgb, 255:red, 0; green, 0; blue, 0 }  ,fill opacity=1 ] (140,215) .. controls (140,213.9) and (140.9,213) .. (142,213) .. controls (143.1,213) and (144,213.9) .. (144,215) .. controls (144,216.1) and (143.1,217) .. (142,217) .. controls (140.9,217) and (140,216.1) .. (140,215) -- cycle ;
\draw  [fill={rgb, 255:red, 0; green, 0; blue, 0 }  ,fill opacity=1 ] (190,215) .. controls (190,213.9) and (190.9,213) .. (192,213) .. controls (193.1,213) and (194,213.9) .. (194,215) .. controls (194,216.1) and (193.1,217) .. (192,217) .. controls (190.9,217) and (190,216.1) .. (190,215) -- cycle ;
\draw    (77,263) -- (201,263) ;
\draw [shift={(203,263)}, rotate = 180] [color={rgb, 255:red, 0; green, 0; blue, 0 }  ][line width=0.75]    (10.93,-3.29) .. controls (6.95,-1.4) and (3.31,-0.3) .. (0,0) .. controls (3.31,0.3) and (6.95,1.4) .. (10.93,3.29)   ;
\draw    (380,126) -- (504,126) ;
\draw [shift={(506,126)}, rotate = 180] [color={rgb, 255:red, 0; green, 0; blue, 0 }  ][line width=0.75]    (10.93,-3.29) .. controls (6.95,-1.4) and (3.31,-0.3) .. (0,0) .. controls (3.31,0.3) and (6.95,1.4) .. (10.93,3.29)   ;
\draw    (506,174) -- (383,174) ;
\draw [shift={(381,174)}, rotate = 360] [color={rgb, 255:red, 0; green, 0; blue, 0 }  ][line width=0.75]    (10.93,-3.29) .. controls (6.95,-1.4) and (3.31,-0.3) .. (0,0) .. controls (3.31,0.3) and (6.95,1.4) .. (10.93,3.29)   ;
\draw  [fill={rgb, 255:red, 0; green, 0; blue, 0 }  ,fill opacity=1 ] (392,214) .. controls (392,212.9) and (392.9,212) .. (394,212) .. controls (395.1,212) and (396,212.9) .. (396,214) .. controls (396,215.1) and (395.1,216) .. (394,216) .. controls (392.9,216) and (392,215.1) .. (392,214) -- cycle ;
\draw  [fill={rgb, 255:red, 0; green, 0; blue, 0 }  ,fill opacity=1 ] (443,214) .. controls (443,212.9) and (443.9,212) .. (445,212) .. controls (446.1,212) and (447,212.9) .. (447,214) .. controls (447,215.1) and (446.1,216) .. (445,216) .. controls (443.9,216) and (443,215.1) .. (443,214) -- cycle ;
\draw  [fill={rgb, 255:red, 0; green, 0; blue, 0 }  ,fill opacity=1 ] (493,214) .. controls (493,212.9) and (493.9,212) .. (495,212) .. controls (496.1,212) and (497,212.9) .. (497,214) .. controls (497,215.1) and (496.1,216) .. (495,216) .. controls (493.9,216) and (493,215.1) .. (493,214) -- cycle ;
\draw    (386,263) -- (510,263) ;
\draw [shift={(512,263)}, rotate = 180] [color={rgb, 255:red, 0; green, 0; blue, 0 }  ][line width=0.75]    (10.93,-3.29) .. controls (6.95,-1.4) and (3.31,-0.3) .. (0,0) .. controls (3.31,0.3) and (6.95,1.4) .. (10.93,3.29)   ;
\draw    (276,15) -- (276,276) ;

\draw (5,64) node [anchor=north west][inner sep=0.75pt]   [align=left] {Sender};
\draw (91,18) node [anchor=north west][inner sep=0.75pt]   [align=left] {Commit Stage};
\draw (191,64) node [anchor=north west][inner sep=0.75pt]   [align=left] {Receiver};
\draw (54,62) node [anchor=north west][inner sep=0.75pt]   [align=left] {$\displaystyle \left(\ket{\mathrm{aux}}\right)$};
\draw (126,105) node [anchor=north west][inner sep=0.75pt]   [align=left] {$\displaystyle m_{0}$};
\draw (126,150) node [anchor=north west][inner sep=0.75pt]   [align=left] {$\displaystyle r_{0}$};
\draw (126,242) node [anchor=north west][inner sep=0.75pt]   [align=left] {$\displaystyle m_{c}$};
\draw (306,64) node [anchor=north west][inner sep=0.75pt]   [align=left] {Sender};
\draw (394,18) node [anchor=north west][inner sep=0.75pt]   [align=left] {Reveal Stage};
\draw (500,64) node [anchor=north west][inner sep=0.75pt]   [align=left] {Receiver};
\draw (355,62) node [anchor=north west][inner sep=0.75pt]   [align=left] {$\displaystyle \left(\ket{\mathrm{aux}}\right)$};
\draw (429,105) node [anchor=north west][inner sep=0.75pt]   [align=left] {$\displaystyle m_{0}$};
\draw (429,153) node [anchor=north west][inner sep=0.75pt]   [align=left] {$\displaystyle r_{0}$};
\draw (429,242) node [anchor=north west][inner sep=0.75pt]   [align=left] {$\displaystyle m_{r}$};
\draw (523,253) node [anchor=north west][inner sep=0.75pt]   [align=left] {$\displaystyle \{\mathrm{Accept} ,\ \mathrm{Reject}\}$};

\end{tikzpicture}
    \caption{An $c$, $r$-message post-quantum bit commitment scheme.  In a post-quantum bit commitment, all messages are classical, but the sender might be quantum.}
    \label{fig:quantum_bit_commitment}
\end{figure}
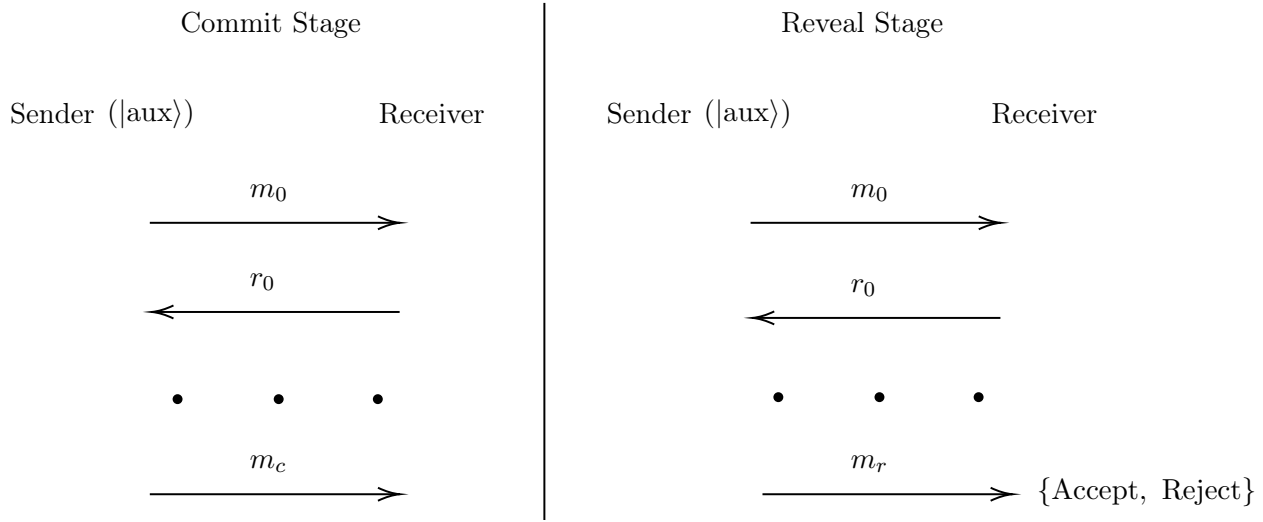

Now we turn our attention to post-quantum interactive protocols with a specific type of security.  The notion of security that we will care about is called \emph{non-malleability}~\cite{dolev1991non}.  At a high level, a post-quantum commitment scheme is non-malleable if any bounded man-in-the-middle adversary can not modify a commitment into another valid commitment in a systematic way (i.e. they may be able to output commitments to random bits, or forward commitments they receive).  We say that the man-in-the-middle adversary receives commitments on the left and sends their commitments on the right.  When discussing man-in-the-middle attacks against interactive protocols, there are two additional concepts we need to introduce: tags and schedules.  

A \emph{tag-based} commitment scheme is one where in addition to $\lambda$, the sender receives a classical tag $\tg_{\lambda} \in \{0, 1\}^{t(\lambda)}$, which we assume (w.l.o.g.) is contained in the sender's initial workspace register (i.e. the sender starts with $\ket{b}\otimes \ket{\tg_{\lambda}} \otimes \ket{0}$), and the receiver's private workspace register after the commit phase.  For a fixed sequence of tags $\tg = \{\tg_{\lambda}\}_{\lambda}$, we require that the corresponding family of commitment schemes satisfy hiding and binding.  At a high level, a tag is meant to prevent an adversary from forwarding communication from the left to the right, in the sense that we will make the definitions such that in order for an adversary to break the non-malleable property of a commitment scheme, we will require that the man in the middle uses different tags in the left and right commitments. 

For man-in-the-middle adversaries, a schedule refers to the sequence of messages sent in the left and right that the adversary accepts.  In this paper we define non-malleability with respect to a specific schedule, which we call $\sch$, defined by \Cref{fig:sync_comp_sec_expt}.  Certain post-quantum bit commitments might be secure against more general schedules, or even all schedules, but any scheme that is secure against this schedule will suffice for our result.  

\begin{definition}[Many-to-many synchronous non-malleable property of commitment scheme]
    Let $\mathcal{C}$ be a post-quantum, tagged, commitment scheme.  For a man-in-the-middle quantum adversary $A = (A_{\lambda})_{\lambda}$, that receives $l$ commitments on the left and outputs commitments to the $l'$ bits on the right according to the schedule $\sch$, let $\mim^{A}_{\mathcal{C}}(\lambda, x)$ denote the random variable over the private workspace of the adversary and the message bits it outputs, $(\OUT_{S}(m_i, A, R), m)_{i \leq l'}$, where $m$ is the $l'$ bits that the adversary commits to on the right, where $R$ is an honest execution of the receiver, whenever the $l$ bits being committed to on the left are $x$.  Let $\mim^{A}_{\mathcal{C}}(\lambda, x) = \bot$ whenever the adversary uses the tag that the original commitment used.  We assume that the adversary commits to each $m_i$ individually using the bit commitment scheme.  A post-quantum bit commitment scheme is \emph{one-to-many non-malleable} if 
    \begin{equation*}
       \mim_{\mathcal{C}}^{A}(\lambda, x_1) \approx_{c} \mim_{\mathcal{C}}^{A}(\lambda, x_2)
    \end{equation*}
    for all message on the left, $x_1$ and $x_2$.  
\end{definition}

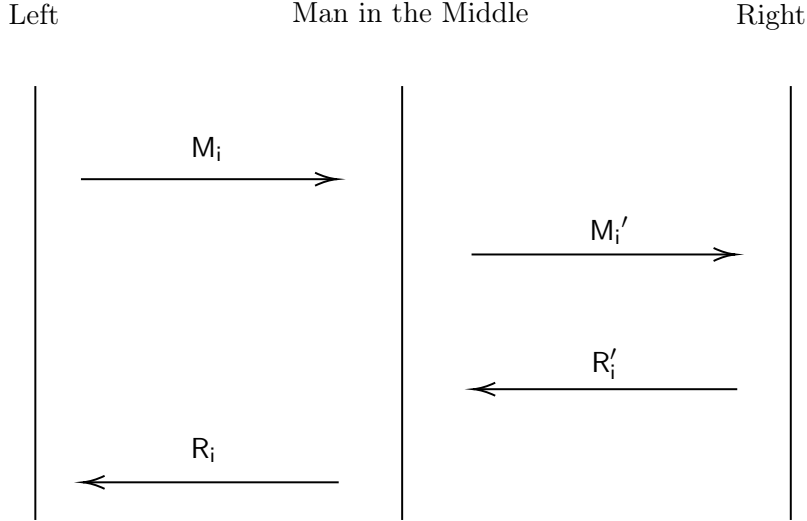
\begin{figure}
    \centering
    \tikzset{every picture/.style={line width=0.75pt}} %

\begin{tikzpicture}[x=0.75pt,y=0.75pt,yscale=-1,xscale=1]

\draw    (321,81) -- (321,300) ;
\draw    (159,128) -- (286,128) ;
\draw [shift={(288,128)}, rotate = 180] [color={rgb, 255:red, 0; green, 0; blue, 0 }  ][line width=0.75]    (10.93,-3.29) .. controls (6.95,-1.4) and (3.31,-0.3) .. (0,0) .. controls (3.31,0.3) and (6.95,1.4) .. (10.93,3.29)   ;
\draw    (356,166) -- (487,166) ;
\draw [shift={(489,166)}, rotate = 180] [color={rgb, 255:red, 0; green, 0; blue, 0 }  ][line width=0.75]    (10.93,-3.29) .. controls (6.95,-1.4) and (3.31,-0.3) .. (0,0) .. controls (3.31,0.3) and (6.95,1.4) .. (10.93,3.29)   ;
\draw    (490,234) -- (359,234) ;
\draw [shift={(357,234)}, rotate = 360] [color={rgb, 255:red, 0; green, 0; blue, 0 }  ][line width=0.75]    (10.93,-3.29) .. controls (6.95,-1.4) and (3.31,-0.3) .. (0,0) .. controls (3.31,0.3) and (6.95,1.4) .. (10.93,3.29)   ;
\draw    (289,281) -- (162,281) ;
\draw [shift={(160,281)}, rotate = 360] [color={rgb, 255:red, 0; green, 0; blue, 0 }  ][line width=0.75]    (10.93,-3.29) .. controls (6.95,-1.4) and (3.31,-0.3) .. (0,0) .. controls (3.31,0.3) and (6.95,1.4) .. (10.93,3.29)   ;
\draw    (136,81) -- (136,300) ;
\draw    (517,81) -- (517,300) ;

\draw (264,37) node [anchor=north west][inner sep=0.75pt]   [align=left] {Man in the Middle};
\draw (121,38) node [anchor=north west][inner sep=0.75pt]   [align=left] {Left};
\draw (488,38) node [anchor=north west][inner sep=0.75pt]   [align=left] {Right};
\draw (213,105) node [anchor=north west][inner sep=0.75pt]   [align=left] {$\displaystyle \mathsf{M_{i}} \ $};
\draw (414,145) node [anchor=north west][inner sep=0.75pt]   [align=left] {$\displaystyle \mathsf{M_{i}} '$};
\draw (415,212) node [anchor=north west][inner sep=0.75pt]   [align=left] {$\displaystyle \mathsf{R_{i} '}$};
\draw (213,257) node [anchor=north west][inner sep=0.75pt]   [align=left] {$\displaystyle \mathsf{R_{i}} \ $};

\end{tikzpicture}
    \caption{Schedule of messages sent in $\sch$.  The schedule is chosen specifically so that an adversary might forward commitments to other challengers, in particular the left and right always send their messages \emph{first}, before expecting a response.  In \Cref{prot:k-fold-unrepeatable}, both the left and right are executed by the challenger, and the adversary executes the man-in-the-middle strategy.  }
    \label{fig:sync_comp_sec_expt}
\end{figure}

The use of the word many-to-many indicates that the adversary can receive any number of commitment on the left, and commits to many bits on the right.  Being non-malleable means no matter what joint distribution over bits the adversary commits to on the right, the state of the adversary and bits committed to on the right is indistinguishable when the left commits to $0$ versus $1$.
We note that we define $\mim$ to take the value $\bot$ whenever the adversary re-uses the tag to allow for the ``un-interesting'' case when the adversary does not try to tamper with the commitment. %

\subsection{Parallel repetition fails for \texorpdfstring{$4$}{4}-message quantum interactive protocols}

We show that, if there exists a $(c+r)$-message post-quantum interactive bit commitment scheme that satisfies one-to-many non-malleability, then the following protocol is a $2(c+r)$-message quantum interactive protocol such that the soundness does not change after $k$ repetitions.  The protocol involves the challenger and adversary simultaneously making commitments to each other, and then revealing their commitments, with the adversary winning if they can ``flip'' the challengers commitment.  At a high level, the challenger and adversary will do the following: The challenger will begin by sending a message corresponding to the next message of commitment to $k-1$ copies of a bit $b$ from the challenger to the adversary.  Before the adversary sends to response to this message, the adversary and challenger will compute and send (in parallel) the next message of $k-1$ commitments from the adversary to the challenger.  After doing this, the adversary will send their response to the challenger's initial commitment.  Formally, assume that a $(c+r)$-message commitment scheme exists and consider the following protocol:

\begin{longfbox}[breakable=false, padding=1em, margin-top=1em, margin-bottom=1em]
\begin{algorithm}\label{prot:k-fold-unrepeatable}
    {\bf $k$-fold unrepeatable challenger}
\end{algorithm}

\noindent \textbf{Input: } Security parameter $\lambda$, tag $\tg$.
\begin{enumerate}
    \item Sample a bit $b$ uniformly at random.
    \item Repeat $\lceil c/2 \rceil$ times: 
    \begin{enumerate}
        \item The challenger sends the next message in $k-1$ commitments to $b$. \label{step:challenger_commit}
        \begin{enumerate}
            \item The adversary sends the next message in the commitments to $k-1$ many bits.\label{step:adversary_commit}
            \item The challenger sends their response to the $k-1$ commitments from the adversary.\label{step:challenger_response}
        \end{enumerate}
        \item The adversary sends the responses to the challengers commitments to $b$.\label{step:adversary_response}
    \end{enumerate}
    \item Check that the adversary does not use the tag $\tg$. If not, reject.
    \item Repeat $\lceil r/2 \rceil$ times:
    \begin{enumerate}
        \item The challenger sends the next message in $k-1$ reveals to $b$.
        \begin{enumerate}
            \item The adversary sends the next message in the reveal to $k-1$ many bits.
            \item The challenger sends their response to the $k-1$ reveals from the adversary.
        \end{enumerate}
        \item The adversary sends the response to the challengers reveals to $b$.
    \end{enumerate} 
    \item Let $\{c_i\}_{i = 0}^{k-1}$ be the bits that the adversary revealed.  Accept if 
    \begin{enumerate}
        \item $\oplus_{i = 0}^{k-1} c_i \neq b$ and
        \item Every commitment sent by the adversary is accepted.  
    \end{enumerate}
\end{enumerate}
\end{longfbox}

The use of the doubly indented bullet points is meant to highlight that steps \ref{step:challenger_commit} and \ref{step:adversary_response} correspond to an execution of the commitment scheme, and steps \ref{step:adversary_commit} and \ref{step:challenger_response} correspond to another execution of the commitment scheme, interwoven with the first.  Note that in the final iteration of each loop, the response might be an empty message (for example, if the final message the commitment or reveal stages has already been sent).  With this, it is clear that \Cref{prot:k-fold-unrepeatable} is a $2(c+r)$-message protocol if the commitment scheme involves exchanging $c$ messages in the commit stage and $r$ messages in the reveal stage.
We will show that both \Cref{prot:k-fold-unrepeatable} and the $k$-fold parallel repetition of \Cref{prot:k-fold-unrepeatable} have soundness $1/2$, assuming that the commitment scheme satisfies certain properties.  

\begin{lemma}
    Assume that $\mathcal{C}$ is a post-quantum commitment scheme satisfying one-to-many non-malleability.  Then for every $k$, and every polynomial-time quantum adversary $A$, there exists a negligible function $\epsilon$ such that the probability that $A$ is accepted by the challenger executing \Cref{prot:k-fold-unrepeatable} is at most
    \begin{equation}
        \frac{1}{2} + \epsilon(\lambda)\,.
    \end{equation}
\end{lemma}
\begin{proof}
    Assume for the sake of contradiction that there is an adversary $(A_{\lambda})_{\lambda}$ that is accepted by the challenger with advantage $1 / p(\lambda)$ for some polynomial $p(\cdot)$.  We are going to construct an adversary for the non-malleability property of $\mathcal{C}$ as follows: The challenger on the left commits to the string $b^{k-1}$ using $k-1$ copies of the commitment scheme (here $b$ can be either $0$ or $1$).  The adversary then runs $A$ on the commitments from the left to receive a commitment to a string of length $k-1$, denoted $\wt{m}$.  By step $3$, we can assume that all of the commitments in used to commit to $\wt{m}$ use a different tag than the challenger's.

    By the assumption about $A$, $\wt{m}$ has parity $1 - b$ with probability $1/2 + 1/p(\lambda)$.  Thus, with probability $1/2 + 1/p(\lambda)$, the message string in $\mim^{A}_{\mathcal{C}}(0^{k-1}, \lambda)$ has XOR $1$, and the message string in $\mim^{A}_{\mathcal{C}}(1^{k-1}, \lambda)$ has XOR $0$.  Thus, the two distributions are distinguishable with non-negligible advantage by an polynomial-time adversary that accepts if the messages sent to the right have even parity.  Thus, if there is an adversary that succeeds in \Cref{prot:k-fold-unrepeatable} with non-negligible advantage, $\mathcal{C}$ is not many-to-many non-malleable, a contradiction.  
\end{proof}

Next we show that there is a adversary that is accepted with probability $1/2$ in the $k$-fold parallel repetition of \Cref{prot:k-fold-unrepeatable}.  At a high level, for each challenger in the $k$-fold parallel repetition, the adversary will forward commitments from all $k-1$ other challengers to them.  If the XOR of all of the committed bits is equal to $1$, every bit that the adversary commits to will be $b \oplus 1$, and the adversary will win every game in the $k$-fold repetition.  
\begin{lemma}
    For every commitment scheme $\mathcal{C}$, there exists an adversary that accepted with probability $1/2$ in the $k$-fold parallel repetition of \Cref{prot:k-fold-unrepeatable}, when all $k$ repetitions are given unique tags $\tg$.
\end{lemma}
\begin{proof}
In step \ref{step:challenger_commit} of \Cref{prot:k-fold-unrepeatable}, the challenger sends $k-1$ messages to the adversary.  In the $k$-fold parallel repetition, we will denote by $m^{i, \alpha}_{l}$ the $\alpha^{th}$ commitment sent from the $i^{th}$ repetition of the protocol, and similarly for responses.  Consider the following adversary for the $k$-fold parallel repetition of \Cref{prot:k-fold-unrepeatable}.

\begin{longfbox}[breakable=false, padding=1em, margin-top=1em, margin-bottom=1em]
\begin{algorithm}\label{prot:k-fold-adversary}
    {\bf Adversary for $k$-fold repetition of \Cref{prot:k-fold-unrepeatable}.}
\end{algorithm}
\noindent \textbf{Input:} Security parameter $\lambda$.
\begin{enumerate}
	\item For $l$ from $0$ to $\lfloor c/2 \rfloor$:
	\begin{enumerate}
        \item Receive $k$ messages, $\{m^{0, \alpha}_{l}\}_{\alpha \in [k-1]}, \ldots, \{m^{k-1, \alpha}_{l}\}_{\alpha \in [k-1]}$ from the challengers.
        \begin{enumerate}
            \item To every challenger $i$, send the $k-1$ messages $\{m^{j+i \mod k, j+1 \mod k}_{l}\}_{j \in [1:k]}$.  
            \item From challenger $i$, receive $k-1$ responses, $\{r^{j+i \mod k, j+1 \mod k}_{l}\}_{j \in [1:k]}$ from the challengers.
        \end{enumerate} 
        \item To every challenger $i$, send the $k-1$ responses $\{r^{i, \alpha}_{l}\}_{\alpha \in [k-1]}$.
	\end{enumerate}
     \texttt{(The challengers reveal bits $b_0, \ldots b_k$.  To each challenger $i$ the adversary reveal $\{b_j\}_{j \neq i}$.)}
    \item For $l$ from $0$ to $\lfloor r/2 \rfloor$:
    \begin{enumerate}
        \item Receive $k$ messages, $\{m^{0, \alpha}_{l}\}_{\alpha \in [k-1]}, \ldots, \{m^{k-1, \alpha}_{l}\}_{\alpha \in [k-1]}$ from the challengers.
        \begin{enumerate}
            \item To every challenger $i$, send the $k-1$ messages $\{m^{j+i \mod k, j+1 \mod k}_{l}\}_{j \in [1:k]}$.  
            \item From challenger $i$, receive $k-1$ responses, $\{r^{j+i \mod k, j+1 \mod k}_{l}\}_{j \in [1:k]}$ from the challengers.
        \end{enumerate} 
        \item To every challenger $i$, send the $k-1$ responses $\{r^{i, \alpha}_{l}\}_{\alpha \in [k-1]}$.
    \end{enumerate}
\end{enumerate}
\end{longfbox}
We first note that \Cref{prot:k-fold-adversary} runs in polynomial time and always produces valid commitments to the bits that it announces, and assuming that all of the challengers used unique tags, every challenger sees commitments that use unique tags.

Let $B = \oplus_{j = 0}^{k-1} b_j$ be the XOR of all of the bits that the challengers committed to.
Because the adversary announces to every challenger, $j$, the bits of the other challengers (not including $j$), the adversary announces bits that XOR to $\left(\bigoplus_{l \neq j} b_l\right) = b_j \oplus B$, to the $j^{th}$ repetition of the protocol.  If $B = 1$, \Cref{prot:k-fold-adversary} wins all of the repetitions of \Cref{prot:k-fold-unrepeatable}.  Since every challenger samples a bit uniformly at random, this happens with probability $1/2$.  Thus the adversary is accepted in the $k$-fold repetition of \Cref{prot:k-fold-unrepeatable} with probability $1/2$.  
\end{proof}

If $c = r = 1$, then we say that a quantum bit commitment is \emph{non-interactive}.  Classical non-interactive, non-malleable bit commitments are known to exist in a number of settings and from a wide variety of assumptions~\cite{lin2020two}, and there is no formal evidence ruling out the existence of post-quantum (or fully quantum) non-interactive non-malleable bit commitments.  Under this assumption, combined with the previous claims, we have the following corollary.

\begin{corollary}
    If there exists a non-interactive post-quantum bit commitment scheme satisfying one-to-many non-malleability, then there is a $4$-message quantum interactive protocol for which parallel repetition fails.
\end{corollary}

Thus, it can not be the case that both non-interactive one-to-many non-malleable commitments exist, and that strong parallel repetition, where the soundness for the $k$-fold parallel repetition goes down as $\epsilon^k$, holds $4$-message quantum interactive protocols.  If one believes that these kinds of commitments do exist, then the $3$-message parallel repetition theorem proved here is tight.  However this argument only shows that for every fixed $k$, there exists a protocol for which soundness of the (exactly) $k$-fold parallel repetition stays high.  It could still be the case that for every $4$-message protocol, there exists a large $k$ for which the soundness of the $k$-fold parallel repetition of the protocol has negligible soundness.  We leave the task of ruling out this (weaker) form of parallel repetition as an open question.

\section{Round compression for quantum argument systems}
\label{sec:compression}

A quantum interactive \emph{argument} 
is a special type of quantum interactive proof where the completeness and soundness conditions hold with respect to computationally efficient provers. Formally, a promise decision problem $L = (L_{yes},L_{no})$ admits a quantum interactive argument with completeness $c$ and soundness $s$ if there exists a polynomial-time quantum verifier $V$ and a polynomial-time quantum ``honest'' prover $P$ such that
\begin{enumerate}
    \item For all $x \in L_{yes}$ there exists a quantum advice state $\ket{\psi_x}$ such that $P$ given $\ket{\psi_x}$ interacting with $V$ on input $x$ gets accepted with probability at least $c$;
    \item For all polynomial-time ``malicious'' provers $P^*$, for all $x \in L_{no}$, for all quantum advice states $\ket{\psi_x}$, the prover $P^*$ interacting with $V$ on input $x$ gets accepted with probability at most $s$.
\end{enumerate}
In other words, the quantum protocol that the verifier $V$ engages in is $s$-computationally secure when $V$ is given ``no'' instances from $L_{no}$.

In this section we prove a \emph{round compression} result for quantum interactive arguments. It was proved by Kitaev and Watrous in~\cite{kitaev2000parallelization} that every $m$-round quantum interactive \emph{proof} (i.e. were the prover has unbounded computational power) for a (promise) language $L$ can be \emph{compressed} into another interactive proof for $L$ that only has $3$ messages and the soundness is worsened by a factor $\poly(m)$. By parallel repetition, the soundness can be improved back to a constant~\cite{kitaev2000parallelization}. This stands in contrast to the classical case, where an analogous round compression technique for interactive proofs is unknown and considered unlikely. 

We prove an analogue of the Kitaev--Watrous compression procedure for interactive arguments:

\begin{theorem}[Round compression for quantum interactive arguments]
\label{thm:compression}
    Let $L$ be a promise language with an $m(n)$-message quantum interactive argument with completeness $c$ and soundness $s$ for $m \ge 3$.
    Then there exists a $3$-message quantum interactive argument for $L$ with completeness $1 - \frac{2(1 - c)}{m - 1}$ and soundness $1 - \frac{1 - s}{(m - 1)^4}$.
    Furthermore, if the original prover has complexity $t_P$ and verifier has complexity $t_V$, then the compressed protocol has prover complexity $m^{O(1)}(t_P + t_V)$ and verifier complexity $m^{O(1)}t_V$.
\end{theorem}

We follow the proof approach of Kempe et al.\ \cite{kempe2007using} (see also~\cite{vidick2016quantum}) who gave an iterative procedure to round-compress quantum proof systems (whereas the original Kitaev-Watrous proposal achieved the compression in one step). The idea as follows: given a $(2r+1)$-message ``original'' protocol, one can obtain an $(r+1)$-message ``compressed'' protocol where the prover first sends the intermediate state of the verifier at round $r+1$ of the original protocol (i.e., the midpoint). The verifier in the compressed protocol randomly decides whether to play the original protocol forwards or backwards to check respectively whether the original verifier would have accepted, or was initialized properly. 
Our protocol is similar in spirit to this reduction, except that we must slightly tweak the reduction and analysis to ensure that the reduction is efficient.

\begin{remark}
    In the quantum interactive proof setting, the assumption of perfect completeness on the original protocol is not necessary, because every quantum interactive proof can be transformed to have perfect completeness (see~\cite{kitaev2000parallelization} and~\cite{vidick2016quantum}). However we do not know whether quantum arguments can be generically transformed to have perfect completeness.
\end{remark}

\begin{remark}
    To maintain consistency with existing work regarding interactive proofs and arguments, in this section we adopt the terminology ``verifier'', ``honest prover'', and ``adversary'' (or malicious prover).  The verifier corresponds to the challenger in an interactive argument.  When talking about completeness, the honest prover corresponds to the adversary, and when talking about soundness we still use the term adversary.  
\end{remark}

\begin{remark}
    When we refer to the \emph{total run-time} of an efficient adversary in an interactive argument, as the total time it takes to execute the entire interaction between the adversary and the challenger.  If every action the adversary takes is polynomial time, then the total run-time will be a polynomial in $n$ as well.  
\end{remark}

We first describe a ``round-halving'' compression procedure that transforms a verifier for a $(2r+1)$-message protocol to a verifier for an equivalent $(r+1)$-message protocol (in the completeness and soundness sense).  Our reduction adds at most a constant overhead to the gate complexity of the verifier. At the end, we will iterate this procedure logarithmically many times to obtain a $3$-message protocol.

Let $\{C_i\}_{i \leq r}$ be a verifier in a $(2r+1)$-message quantum interactive protocol.
We describe a $(r + 1)$-message protocol and analyse its completeness and soundness.  We briefly recall the notation used to describe quantum interactive arguments.  $\reg{A_i}$ denotes the adversaries private workspace in round $i$, $\reg{M}_i$ denotes the message register sent from the adversary to the challenger, and $\reg{R_i}$ denotes the response register sent from the challenger to the adversary.  The adversary applies the unitary $A_i$ on $\reg{A_iR_{i-1}}$ to get a state on $\reg{A_{i+1}M_i}$, and the verifier applies a unitary $C_i$ on $\reg{W_iM_i}$ to get a state on $\reg{W_{i+1}R_i}$.  In the final round, the challenger applies a POVM $\{D, \id - D\}$ on $\reg{W_{r}M_{r}}$ to decide whether to accept or reject.  

\begin{longfbox}[breakable=false, padding=1em, margin-top=1em, margin-bottom=1em]
\begin{algorithm} {\bf Verifier for $(r+1)$-message compressed protocol} \label{prot:quantum_round_compression} 
\end{algorithm}
\begin{enumerate}
    \item Receive registers $\reg{M}_{r/2}\reg{W}_{r/2}$ from the adversary.  Flip a unbiased coin.  
    \item If the outcome of the coin flip is heads:
    \begin{enumerate}
        \item Run the original $2r+1$-message protocol starting from interaction $r/2$, i.e. apply unitary $C_{r/2}$ to registers $\reg{M}_{r/2}\reg{W}_{r/2}$, send the result of the coin flip and $\reg{R}_{r/2}$ to the adversary, and continue for interactions $r/2+1$ through $r$.\label[step]{step:forward_compression}
        \item After $r/2$ many interactions, receive register $\reg{M_{r}}$, measure $\Pi$ on registers $M_{r}\reg{W}_{r}$.  Accept if the measurement accepts.
    \end{enumerate}
    \item If outcome of the coin flip is tails:
    \begin{enumerate}
        \item Run the original $2r+1$-message protocol in reverse, i.e. send the result of the coin flip and $\reg{M}_{r/2}$ back to the adversary, receive register $\reg{R}_{r/2-1}$ from the adversary, apply $(C_{r/2-1})^{\dagger}$ to $\reg{R}_{r/2-1}\reg{W}_{r/2}$, and continue for interactions $r/2-1$ through $1$.  \label[step]{step:backward_compression} 
        \item After $r/2$ many interactions, receive register $\reg{R}_{1}$.  Apply $(C_{1})^{\dagger}$ on registers $\reg{R}_1\reg{W}_2$ to get registers $\reg{M}_0\reg{W}_0$, and measure $\proj{0}_{\reg{W}_0}$.  Accept if the measurement accepts.
    \end{enumerate}
\end{enumerate}
\end{longfbox}
\begin{claim}[Completeness]\label{claim:round_compression_completeness}
    If there is an non-uniform (resp. uniform) honest prover (denoted the \emph{original} honest prover) that succeeds in the original $(2r+1)$-message protocol with probability $1 - \epsilon$, then there is a non-uniform (resp. uniform) honest prover that succeeds in the $(r+1)$-message compressed protocol with probability $1 - \epsilon/2$.
    Furthermore, if the original protocol has running time $t_P, t_V$ for the prover and the verifier, then the new protocol has running time $O(t_P + t_V)$ for the prover and $O(t_V)$ for the verifier.
\end{claim}
\begin{proof}
    Let $A_{i}$ be unitary operations that the honest prover implements in the $i$'th interaction.  Specifically, the honest prover applies $A_{i}$ to registers $\reg{R}_{i-1}\reg{A_i}$ to get a message register $\reg{M_i}$ and new private register $\reg{A_{i+1}}$, and then sends $\reg{M}_i$ to the verifier.  Initially, $\reg{R}_{0}$ is an empty register, and $\reg{A_0}$ is initialized to some state $\proj{\aux}_{\reg{A}_0}$ (in the non-uniform case this might be an advice state, otherwise $\proj{0}$).

    The honest prover for the compressed protocol will essentially implement the honest strategy, applying the unitaries of the un-compressed honest prover forward or backwards depending on the outcome of the verifier's coin flip.  The honest prover begins by performing preparing a register $\reg{W_0}$ in the state $\ket{0}_{\reg{W_0}}$, and then performs $A_{0}, C_0, A_{1}, C_1, \ldots, A_{r/2}$ to the appropriate registers to get a state over registers $\reg{A_{r/2+1}}\reg{M_{r/2}}\reg{W_{r/2}}$.  The honest prover then sends $\reg{M_{r/2}}\reg{W_{r/2}}$ to the challenger as their initial message.  

    If the outcome of the verifier's coin flip is heads, the honest prover applies $A_{r/2+1}$ to the verifier's response, and continues as the un-compressed honest prover would starting from round $(r/2)+1$.  After all $r/2$ interactions, the verifier and honest prover hold the state after implementing the unitaries $A_0, C_0, \ldots, A_{r}$ on the initial state $\proj{\aux}_{\reg{A_0}}\otimes \proj{0}_{\reg{W_0}}$ and then measuring $D$ on registers $\reg{M_{r}}\reg{W_{r}}$.  Since the original protocol accepted with probability $1-\epsilon$, we have that
    \begin{equation*}
        \Tr(D (A_{r} C_{r-1} \ldots C_0 A_0) (\proj{\aux}_{\reg{A}_0} \otimes  \proj{0}_{\reg{W_0}}) (A_0^\dagger C_0^{\dagger} \ldots C_{r-1}^{\dagger} A_{r}^{\dagger})) = 1 - \epsilon\,.
    \end{equation*}
    Thus, the probability that the honest prover is accepted is also $1 - \epsilon$.  

    If the outcome of the verifier's coin flip is tails, when the honest prover receives register $\reg{M_i}$, the honest prover applies $(A_{i})^{\dagger}$ to $\reg{M_i}\reg{A_{i+1}}$ and sends registers $\reg{R}_{i-1}$ to the verifier, and continues by applying the inverse of the un-compressed honest prover unitary, going backwards in round number.  When the protocol reaches the final message, the honest prover and verifier will have performs the inverse of $C_0, A_1, C_1, \ldots, A_{r/2}$ to yield some state $\ket{\psi}_{\reg{M_0}\reg{A_1}} \otimes \ket{0}_{\reg{W_0}}$.  When the verifier measures $\proj{0}$, the measurement will accept with probability $1$.  

    So with probability $1/2$, the honest prover is accepted with probability $1 - \epsilon$, and with probability $1/2$ the honest prover is accepted with probability $1$.  Thus, the overall probability that the honest prover is accepted is $1 - \epsilon/2$.  

    To analyze the runtime of the honest prover, note that the honest prover for the compressed protocol applies every unitary $A_i$ at most twice, and runs the first verifier's unitaries $C_i$ at most once.  The verifier runs every $C_i$ at most once, so together they run every $A_i$ and $C_i$ at most twice, so its total running time of the protocol is at most twice the total running time of the original honest prover.  We note that the compressed honest prover can be implemented by controlling the $A_i$ and $C_i$ on the message register (so that they act correct controlled on the outcome of the verifier's coin flip).  If the honest prover can only use $2$-qubit gates, this might result in a constant multiplicative overhead to the gate complexity of the honest prover.  
\end{proof}

\begin{claim}[Soundness]\label{claim:round_compression_soundness}
    If there is a adversary that has total run-time $t$ that succeeds in the $(r+1)$-message compressed protocol with probability $1 - \epsilon$, then there is an adversary that has total run-time $O(t + t_V)$ and succeeds in the original $(2r+1)$-message protocol with probability $1 - 16\epsilon$.  
\end{claim}
\begin{proof}
    We can assume that $\epsilon \leq 1/16$, otherwise the claim is trivially true.  Assume that the adversary for the $(r+1)$-message compressed protocol implements an initial unitary $A_{0}$ on register $\reg{A_0}$, which is initialized in state $\ket{\aux}_{\reg{A_0}}$.  Call the output registers of this unitary $\reg{A_{r/2 + 1}}\reg{M_{r/2}}\reg{W_{r/2}}$, of which $\reg{M_{r/2}W_{r/2}}$ are sent to the challenger as the adversary's first message.  After seeing the result of the coin flip, $b \in \{H, T\}$, from the challenger, we denote by $A^{b}_{i}$ the unitaries that the adversary implements in interactions $i = 1$ through $r/2$.  We say that the unitary $A^{H}_{i}$ acts on registers $\reg{R_{r/2 + i - 1}}\reg{A_{r/2 + i}}$ (so that the first one acts on registers $\reg{R_{r/2}}\reg{A_{r/2 + 1}}$ and counts up from there) and the unitary $A^{T}_{i}$ acts on registers $\reg{M_{r/2 - i + 1}}\reg{A_{r/2 - i + 2}}$ (so that the first one acts on registers $\reg{M_{r/2}}\reg{A_{r/2 + 1}}$).  The challengers actions in the game are specified in \Cref{prot:quantum_round_compression}.  Define the following projectors.
    \begin{align*}
        \Pi_{H} &= (A_0^{\dagger}  C_{r/2}^{\dagger} {A^{H}_{1}}^{\dagger} C_{r/2+1}^{\dagger} \ldots A^{H \dagger}_{r/2}) D_{\reg{M_{r}W_{r}}} (A^{H}_{r/2} \ldots C_{r/2 + 1} A^{H}_{1} C_{r/2}  A_0)\,,\\
        \Pi_{T} &= (A_0^{\dagger} {A^{T}_{1}}^{\dagger} C_{r/2-1} \ldots A^{T \dagger}_{r/2} C^{\dagger}_{0}) (\id \otimes \proj{0}_{\reg{W_0}}) (C_{0}^{\dagger} A^{T}_{r/2}\ldots C_{r/2-1}^{\dagger} A^{T}_{1} A_0)\,.
    \end{align*}

    It is clear that $\Tr(\Pi_{H}\proj{\aux})$ and $\Tr(\Pi_{T}\proj{\aux})$ are the probabilities that the adversary is accepted when the result of the challenger's coin flip is $H$ and $T$ respectively.  By assumption, the adversary is accepted in the $r+1$-message compressed protocol with probability $1 - \epsilon$, so the adversary must be accepted in the protocol if the coin flip is fixed to being either heads or tails with probability at least $1-2\epsilon$.  Therefore we have the inequalities
    \begin{align*}
        \Tr(\Pi_{H}\proj{\aux}) \geq 1 - 2\epsilon\,,\\
        \Tr(\Pi_{T}\proj{\aux}) \geq 1 - 2\epsilon\,.
    \end{align*}
    
    Now we construct an adversary for the original $(2r+1)$-message protocol as follows:
    
    \begin{longfbox}[breakable=false, padding=1em, margin-top=1em, margin-bottom=1em]
    \begin{algorithm} {\bf Adversary for the $(2r+1)$-message protocol} \label{prot:prover_from_compressed_prover} 
    \end{algorithm}
    \noindent\textbf{Input: } Initial state $\ket{\aux}$.
    \begin{enumerate}
        \item Run $A_0$ to receive a state on registers $\reg{A_{r/2+1}M_{r/2}W_{r/2}}$.  Implement unitaries $A^{T}_{1}, C^{\dagger}_{r/2 - 1} \ldots, A^{T}_{r/2}, C_0^{\dagger}$ to receive a state on registers $\reg{M_0}\reg{W_0}$.  Measure $\id \otimes \proj{0}_{\reg{W}_0}$, if the measurement fails, abort.  If it accepts, send register $\reg{M_0}$ to the verifier.
        \item For steps $i = 1$ through $r/2$:
        \begin{enumerate}
            \item Run $(A^{T}_{r/2 - i})^{\dagger}$ to get the next message register and send it to the verifier. 
        \end{enumerate}
        \item For steps $i = r/2+1$ through $r$:
        \begin{enumerate}
            \item Run $A^{H}_{i - r/2}$ to get the next message register and send it to the verifier.
        \end{enumerate}
    \end{enumerate}
    \end{longfbox}
    Let $\rho_T$ be the state of the adversary after measuring $\Pi_T$ on $\proj{\aux}$ and accepting.  Let $\rho_H$ be the state after measuring $\Pi_H$ on $\proj{\aux}$ and accepting.  By the gentle measurement lemma (\Cref{lemma:gentle-measurement}) and the definition of the squared Bures distance, we have that the the following inequalities
    \begin{align*}
        \buresdist(\rho_T, \proj{\aux}) &\leq 2(1 - \sqrt{1 - 2\epsilon}) \leq 3\epsilon\,,\\
        \buresdist(\rho_H, \proj{\aux}) &\leq 3\epsilon\,.
    \end{align*}
    Here the first inequality holds for any $0 \le \epsilon \le \frac49$. %
    Using the weak triangle inequality for the squared Bures distance (\Cref{lem:weak_triangle_inequality}), we have that
    \begin{equation*}
        \buresdist(\rho_H, \rho_T) \leq 12 \epsilon\,.
    \end{equation*}
    By the definition of the squared Bures distance, we have that 
    \begin{equation}\label{eqn:heads_tails_fidelity}
        \fidelity(\rho_H, \rho_T) \geq (1 - 6\epsilon)^2 \geq 1 - 12\epsilon\,,
    \end{equation}   
    where we applied the Bernoulli inequality here.  Now we examine the state of the adversary in the un-compressed game.  If the un-compressed adversary does not abort in step $1$, after step $2$, the adversary is left with exactly $\rho_T$, because they condition on $\proj{0}_{\reg{W_0}}$ accepting.  At the end of the protocol, the challenger will have implemented the measurement $\Pi_{H}$ on the adversary's state after step $2$, $\rho_{T}$.  Since $\Pi_{H}$ is a projector and $\rho_{H}$ is the post-measurement state of some initial state after measuring $\Pi_{H}$, we have that $\Tr(\Pi_{H} \rho_{H}) = 1$.  Applying \Cref{lem:projector_to_fidelity} and \Cref{eqn:heads_tails_fidelity}, we have that $\Tr(\Pi_{H}\rho_{T}) \geq \fidelity(\rho_{H}, \rho_{T}) \geq 1 - 16\epsilon$.  Thus, the adversary is accepted with probability at least $1 - 16\epsilon$.  By assumption, the first measurement in the step succeeds with probability $(1 - 2\epsilon)$, thus the total probability of failure is given by
    \begin{equation*}
        2\epsilon + (1 - 2\epsilon) \cdot 12 \epsilon \leq 16\epsilon\,.
    \end{equation*}

    The un-compressed adversary applies every the unitaries $A_i^{T}$ (and $(A_i^{T})^{\dagger}$) two times, the unitaries $C_{i}^{\dagger}$ once, and the unitaries $A_i^{H}$ once.
    Thus the total run time of the adversary is $O(t + t_V)$ accounting for constant multiplicative overheads in implementing controlled unitaries.
\end{proof}

\begin{proof}[Proof of \Cref{thm:compression}]
To prove the theorem, we iterate the round-halving procedure. Let $L = (L_{yes},L_{no})$ be a (promise) decision language with a quantum interactive proof $\proto$ where for inputs of length $n$, the protocol has $m(n)$ messages. We assume without loss of generality that $m(n)$ is of the form $2^{k(n)}+1$; the adversary and challenger can send empty messages at the beginning of the protocol.
In other words, $k = \lceil\log(m - 1)\rceil$ and $2^{k - 1} < m - 1 \le 2^k$.

Let $\proto'$ denote protocol that applies the round-halving procedure described above $k - 1$ times. In other words, when the challenger receives input $x$, it first computes the description of the challenger for the $(2^{k-1}+1)$-message protocol (where $n = |x|$), and then based on this computes the description of the challenger for the $(2^{k-2}+1)$-message protocol, and so forth, until it obtains a challenger for a $3$-message protocol. Since $m(n)$ is polynomial in $n$, $k = O(\log n)$.
The final honest prover runs in time $m^{O(1)} (t_P + kt_V)$ and the verifier runs in time $m^{O(1)} t_V$, where the $O(1)$ factors here are exactly the constant multiplicative overhead above.

We now analyze the completeness and soundness of the $3$-message protocol. Given a ``yes'' instance $x \in L_{yes}$, by definition there is an adversary that is accepted with probability $c = 1 - \delta$.  By \Cref{claim:round_compression_completeness}, there is an adversary for the $2^{k(n) - 1}$+1-message protocol that is accepted with probability $1 - \delta/2$.  Repeating the compression $k - 1$ times, there exists an adversary that is accepted by the $3$-message protocol with probability 
\begin{equation*}
    1 - \frac{\delta}{ 2^{k - 1}} \geq 1 - \frac{2\delta}{m(n) - 1}.
\end{equation*}

Given a ``no'' instance $x \in L_{no}$, the original protocol accepts with probability at most $s = 1 - \eps$ for all efficient adversaries. Then by \Cref{claim:round_compression_soundness}, the soundness of the $2^{k(n)-1}+1$-message protocol is at most $1 - \eps/16$.
Similarly the  soundness of the $2^{k(n)-1}+1$-message protocol is at most $1 - \eps/16^2$. Iterating this we see that the $3$-message protocol has soundness at most 
\[
    1 - \frac{\eps}{16^{k - 1}} \leq 1 - \frac{\eps}{(m - 1)^4}.
\]

We note that all of the reductions are efficient, even after $k - 1$ iterations of round collapse.  Specifically, for ``no'' instances, if there is an adversary for the compressed protocol that has total run-time $t(n)$, then there is an adversary for the original protocol that has total run-time $m^{O(1)} (t + kt_V)$, which is efficient.
\end{proof}

We can slightly extend our round compression to compile any $3$-message quantum interactive argument into a public coin argument using the same strategy as \cite[Theorem 5.4]{Marriott2005}.  A \emph{public coin} quantum interactive argument is a quantum interactive argument where all of the challenger messages are uniformly random coin flips.  At a high level, if we were to apply round compression starting with a $3$-message protocol, the challenger only ever needs to receive one register (either $\reg{R}_0$ to go backwards, or $\reg{M}_1$ to go forwards) from the adversary, so they do not need to send back a quantum register as in round compression for $m > 3$.  Formally, we have the following.

\begin{theorem}[Compilation to public coin]\label{thm:public_coin}
    Let $L$ be a promise language with a $3$-message quantum interactive argument with completeness $c$ and soundness $s$.  Then there exists a $3$-message public coin quantum interactive argument for $L$ with completeness $1 - \frac{1 - c}{2}$ and soundness $1 - \frac{1 - s}{16}$.
\end{theorem}
\begin{proof}
Let $(C_0, \Pi)$ be a challenger for a $3$-message quantum interactive protocol.  Consider the following challenger for a $3$-message protocol.

\begin{longfbox}[breakable=false, padding=1em, margin-top=1em, margin-bottom=1em]
\begin{algorithm} {\bf Challenger for $3$-message public coin protocol} \label{prot:quantum_public_coin} 
\end{algorithm}
\begin{enumerate}
    \item Receive registers $\reg{W}_{1}$ from the adversary.  Flip a unbiased coin and send it to the adversary.
    \item If the outcome of the coin flip is heads:
    \begin{enumerate}
        \item The adversary sends register $\reg{M}_1$.
        \item Measure $\Pi$ on $\reg{M}_1\reg{W}_1$, accept if the measurement accepts.
    \end{enumerate}
    \item If outcome of the coin flip is tails:
    \begin{enumerate}
        \item The adversary sends register $\reg{R}_0$.
        \item Perform $C_{0}^{\dagger}$, and measure $\proj{0}$ on registers $\reg{W}_0$, accept if the measurement accepts.
    \end{enumerate}
\end{enumerate}
\end{longfbox}

This protocol is clearly public coin, as the challenger only sends the outcome of a single coin flip to the adversary.  An honest adversary runs the original $3$-message protocol up to the first two messages and sends the challenger's private workspace to the challenger.  Upon seeing heads they perform the honest adversary unitary and send $\reg{M}_1$, and upon seeing tails they send back $\reg{R}_0$ to the challenger.  A similar argument as \Cref{claim:round_compression_completeness} shows that the completeness behaves the same way as applying a single round of round compression.

Similarly, given an adversary for the public coin protocol, a malicious adversary for the original protocol can simulate the interaction between the public coin adversary and the challenger, and measure the challenger's private workspace to get $\proj{0}_{\reg{W}_0}$, and then get a suitable state on $\reg{M}_0$ to send to the challenger.  From there they run the public coin adversary, conditioned on seeing tails, in reverse for one step, and then the public coin adversary, conditioned on seeing heads, forward for one step.  The same argument as in \Cref{claim:round_compression_soundness} shows that the soundness of the public coin protocol is equivalent to having applied round compression one additional step.  

The run-time of the adversary in both directions is multiplied by a constant factor in the reduction, so the reduction is efficient.  
\end{proof}

Finally, we combine \Cref{thm:compression,thm:public_coin} with our hardness amplification result to show that all polynomial-message quantum interactive arguments can be compressed into $3$-message quantum interactive arguments with negligible soundness.

\begin{corollary}[Parallelization and amplification for quantum interactive arguments]
	Let $L$ be a (promise) language that has a polynomial-message quantum interactive argument with completeness $c \geq 1 - \negl(n)$ (or 1, respectively) and soundness $s \leq 1 - \frac{1}{\poly(n)}$.  Then there is a $3$-message public coin quantum interactive argument for $L$ with completeness $1 - \negl(n)$ (or 1, respectively) and soundness $\negl(n)$.
\end{corollary}

We remark that this parallel repetition is not as efficient as what is possible classically.
Classically, to amplify any interactive argument (say of soundness $\frac12$ to soundness $\frac14$) while preserving round complexity $r$, the state of the art incurs a multiplicative cost of order either $r$ \cite{berman2021tight} or $\lambda$ \cite{chung2010parallel}, where $\lambda$ is a security parameter.
Turning back to quantum interactive arguments, we note that we can first round collapse the argument into three messages (if it has more than three messages), which only incurs a constant multiplicative overhead, and then apply the three-message soundness amplification.
In the end, the cost of the overall compiled protocol is $\Omega(r^4)$, since we need to make up for the loss in soundness in the round collapse theorem.
A more careful analysis of \Cref{claim:round_compression_soundness} could improve the exponent from $\log 16 = 4$ down to $\log 10 \approx 3.32$ but we leave as future work to further improve the amplification efficiency.

\section{Applications}

The parallel repetition theorem has many immediate applications in cryptography, from boosting the security of commitments, to round reduction for zero-knowledge proofs.

\subsection{Strong amplification of quantum bit commitment schemes}
\label{sec:commitments}
Here we show that the security of canonical quantum bit commitment schemes can be amplified through parallel repetition.  While we have already defined bit commitments and notions of post-quantum security, here we review \emph{canonical} quantum bit commitment schemes, and the standard notions of security they are described by, honest hiding and binding

\paragraph{Canonical quantum bit commitments.}
A canonical quantum bit commitment is a kind of non-interactive bit commitment scheme. The scheme consists of two quantum circuits $C_0,C_1$, to commit to be a bit $b$, the sender generates the bipartite pure state $\ket{\psi_b}_{\reg{CR}} = C_b \ket{0 \ldots 0}$ and sends the $\reg{C}$ register of $\ket{\psi_b}$. To reveal, the sender sends $\reg{R}$ to the receiver, along with the bit $b$. The receiver can then verify the original bit $b$ by applying the circuit $C_b^{\dagger}$ and measuring $\ket{0\ldots 0}$.
Importantly, Yan~\cite{yan2023general} showed that all quantum commitment schemes can be compiled to this canonical form while preserving honest hiding and binding security, and honest security for canonical form commitments is equivalent to the stronger standard security notions.
Throughout this section we focus on canonical quantum commitment schemes.

The honest hiding property of a commitment scheme guarantees that an adversarial receiver can not reveal the bit committed to before the reveal phase. Formally,

\begin{definition}[Honest hiding property of commitment scheme]\label{def:hiding_commitment}
    Let $\eps(\lambda)$ denote a function. We say that a commitment scheme  $(\pi_{\lambda})_{\lambda}$ satisfies \emph{$\eps$-computational (resp.\ $\eps$-statistical) honest hiding} if for all non-uniform polynomial-time algorithms  (resp.\ for non-uniform algorithms) $A = (A_\lambda)_\lambda$ that take as input the receiver's register immediately after an honest execution of the commit stage of $\proto_{\lambda}$, the following holds for sufficiently large $\lambda$:
    \begin{equation*}\label{eqn:commitment_hiding_property}
        \Big|\Pr \mbracket{A_\lambda(\rho_{\lambda,0}) = 1} - \Pr\mbracket{A_\lambda(\rho_{\lambda,1}) = 1} \Big | \leq \eps(\lambda) \,.
    \end{equation*}
    Here, $\rho_{\lambda,b}$ denotes the reduced density matrix of the receiver's register after a honest execution of the commit stage of $\proto_{\lambda}$ when the sender is committing to $b$. 
    If $\eps$ is a negligible function of $\lambda$ then we simply say that the scheme satisfies \emph{strong} computational (resp.\ statistical) hiding. If $\epsilon(\lambda)\leq 1-\frac{1}{p(\lambda)}$ for some polynomial $p(\lambda)$ we say it satisfies \emph{weak} computational (resp.\ statistical) hiding.
\end{definition}

The honest binding property, at a high level, says that after the commit phase of a commitment scheme, an adversarial sender can only reveal to the bit that they committed to. Formally,

\begin{definition}[Honest binding property of commitment scheme]\label{def:honest_binding}
    Let $\eps(\lambda)$ denote a function. We say that a commitment scheme $(\proto_{\lambda})_{\lambda}$ satisfies \emph{$\eps$-computational (resp.\ $\eps$-statistical) honest binding} if for all 
    non-uniform polynomial-time algorithms  (resp.\ for all non-uniform algorithms) $A = (A_\lambda)_\lambda$ that take as input the sender's register immediately after an honest execution of the commit stage of $\proto_{\lambda}$, the following holds for sufficiently large $\lambda$:
    \begin{equation*}
        \fidelity\mparen{\Big( A_{\lambda} \otimes \id_{\reg{C}} \Big)(\psi_{\lambda,0}) , \psi_{\lambda,1}} \leq \eps(\lambda)\,,
    \end{equation*}
    where $\psi_{\lambda,b}$ denotes the state of the joint sender-receiver system after an honest execution of the commit stage of $\proto_{\lambda}$ when the committed bit is $b$.

    If $\eps$ is a negligible function of $\lambda$ then we simply say that the scheme satisfies \emph{strong} computational (resp.\ statistical) honest binding. Otherwise if $\epsilon(\lambda)\leq 1-\frac{1}{p(\lambda)}$ for some polynomial $p(\cdot)$ we say that it satisfies \emph{weak} computational (resp.\ statistical) honest binding.
\end{definition}

\paragraph{Parallel repetition of quantum bit commitments.} We show that parallel repetition can be used to amplify weak computationally binding (and strong statistically hiding) commitments into commitments where both hiding and binding are strong.

\begin{corollary}
    Let $\{C_{\lambda, b}\}_{\lambda, b}$ be a canonical commitment scheme satisfying the $\paren{1 - \frac{1}{p(\lambda)}}$-computational honest binding property for some polynomial $p(\cdot)$, and strong statistical hiding.  Then there exists a polynomial $q(\cdot)$ such that the commitment scheme $\{C_{\lambda, b}^{\otimes q(\lambda)}\}_{\lambda, b}$ satisfies the strong computational honest binding property and strong statistical hiding property.
\end{corollary}

\begin{proof}
The computational honest binding property of (canonical) quantum commitments can be equivalently viewed as the maximum winning probability of the following $2$-message game (over efficient adversaries):
\begin{enumerate}
    \item The challenger commits to $0$, and sends the \emph{reveal} register to the adversary.
    \item The adversary sends back a reveal register.
    \item The challenger accepts if applying $C_{\lambda, 1}$ and measuring in the computational basis yields $\ket{0^{\lambda}}$.  
\end{enumerate}
The maximum winning probability of this game is exactly the fidelity present in \Cref{def:honest_binding}.  By \Cref{theorem:non-uniform-rep}, no non-uniform adversary can win the parallel repeated game with probability non-negligibly greater than $(1 - \frac{1}{p(\lambda)})^{q(\lambda)} = \negl(\lambda)$.  Thus, the repeated commitment satisfies strong computational honest binding.  Since $q$ is polynomial, the game still satisfies strong statistical hiding.  
\end{proof}

We can also amplify commitments that satisfy weak computational hiding (instead of binding) using flavor switching \cite{gunn2023commitments,hhan2023hardness}, which we state below.

\begin{proposition}[{\cite[Theorem 7]{hhan2023hardness}}]
\label{prop:flavor-switching}
Let $\eps(n),\delta(n)$ be functions. If $\{C_{\lambda,b}\}_{\lambda,b}$ is an $\eps$-computationally (resp.~statistical) hiding and $\delta$-statistical (resp.~computational) binding commitment scheme, then there exists a $\sqrt{\delta}$-statistical (resp.~computational) hiding and $\eps$-computationally (resp.~statistical) binding commitment scheme. 
\end{proposition}
\begin{corollary}
    The existence of either strong statistical hiding, weak computational honest binding, \emph{or} strong statistical honest binding and weak computational hiding commitments imply standard commitments.
\end{corollary}
\begin{proof}
    The first part of the corollary is the previous corollary, the second part is a direct implication of flavor switching for commitments. 
\end{proof}

This resolves an open problem from Yan~\cite{yan2023general} about the strong amplification of the Uhlmann transformation problem (equivalently quantum bit commitments).
We can push this result slightly to allow for amplification of some commitment schemes that have both weak security against both parties.
\begin{corollary}
    Let $\{C_{\lambda, b}\}_{\lambda, b}$ be a canonical commitment scheme satisfying the $\paren{1 - \frac{1}{p(\lambda)}}$-computational honest binding property for some polynomial $p(\cdot)$, and $\frac{1}{q(\lambda)}$ weak computational hiding, where  $q(\lambda) \geq 2\lambda p(\lambda)$.  Then standard commitments exist.
\end{corollary}
\begin{proof}
    From \Cref{theorem:non-uniform-rep}, the $\lambda p(\lambda)$-fold parallel repetition of the commitment has binding error $(1 - \frac{1}{p(\lambda)})^{\lambda p(\lambda)} + \negl(\lambda) \leq e^{-\lambda} + \negl(\lambda)$, and thus satisfies strong honest binding property, and the $\frac{\lambda p(\lambda)}{q(\lambda)} \leq \frac{1}{2}$-weak hiding property by hybrid argument.
    We can then apply flavor switching to get a commitment that is satisfies the strong computational hiding property and the $\frac{1}{2}$-weak computational binding property.  The $\lambda$-fold parallel repetition of this commitment then satisfies the strong computational binding property too.
\end{proof}

For example, a scheme which is $\frac{1}{2}$-binding and $\lambda^{-1}$-hiding can be amplified to a fully-secure commitment in this way. We leave as an open question whether an $\alpha$-binding, $\beta$-hiding quantum commitment can be amplified for some constant $\alpha,\beta$; classically, this is known for $\alpha + \beta \leq 1 - 1/\poly(\lambda)$, which is tight for black-box constructions \cite{HolensteinS11}.

\subsection{Quantum XOR lemma}

\paragraph{XOR lemma for commitments.}
By examining the flavor switching technique for quantum commitments closely, we can get an \emph{XOR lemma} for quantum commitments, which states that one can amplify the computational hardness of breaking the hiding property of a quantum commitment via the \emph{XOR repetition} of a computationally hiding commitment scheme. In the context of classical commitments, the $k$-fold XOR repetition of a commitment $C$ is a new commitment $C^{\oplus k}$ where to commit to bit $b \in \bits$, the sender will send commitments (using the ``base commitment'' $C$) to $x_1,\ldots,x_k$ for a randomly chosen string $x \in \bits^k$ with parity $b$. It was shown by Yao~\cite{yao82theory,goldreich2011yao} that if $C$ was originally $\eps$-computationally hiding against classical adversaries, then $C^{\oplus k}$ is $O(\eps^k)$-computationally hiding. We now show that a quantum version of the XOR repetition applied to canonical quantum commitments does indeed amplify the hiding security of the commitment. 

Let $C = \{C_{\lambda,b}\}_{\lambda,b}$ be a canonical quantum commitment scheme  that is $\eps$-computational hiding and $\delta$-statistically binding. Define the \emph{XOR repetition} $C^{\oplus k} = \{ C^{\oplus k}_{\lambda,b} \}_{\lambda,b}$ to be the following quantum commitment scheme. Letting $\ket{\psi_{\lambda,b}} = C_{\lambda,b} \ket{0\cdots 0}$ for all $\lambda,b$, we define the circuits $C_{\lambda,b}^{\oplus k}$ so that the corresponding states $\ket{\psi^{\oplus k}_{\lambda,b}}$ are
\[
    \ket{\psi^{\oplus k}_{\lambda,b}}_{\reg{R C}} := \frac{1}{\sqrt{2^{k-1}}} \sum_{x \in \bits^k : |x| = b \, \mathrm{mod}\, 2} \ket{x} \otimes \ket{\psi_{\lambda,x_1}}_{\reg{R_1 C_1}} \otimes \cdots \otimes \ket{\psi_{\lambda,x_k}}_{\reg{R_k C_k}}.
\]
In other words, the XOR-repeated commitment to $b$ is the uniform superposition, over all $k$-bit strings $x$ with parity $b$, of commitments to $x_1, x_2, \ldots, x_k$.
This state can be efficiently prepared by first preparing uniform superposition on $x_1, ..., x_{k - 1}$ and then coherently compute $x_k = x_1 \oplus \cdots \oplus x_{k - 1} \oplus b$.
The commitment register (resp. reveal register) of the XOR-repetition is the concatenation of the commitment register (resp. reveal register plus an extra register storing $x$) of all the individual commitments.

\begin{lemma}[XOR lemma for quantum commitments]
    \label{lemma:xor-commitment}
    If the commitment $C$ is $\eps$-computational hiding and $\delta$-statistically binding, then the XOR repetition $C^{\oplus k}$ is $(\eps^{k/2} + \negl)$-computational hiding and $(k\sqrt\delta)$-statistically binding.
\end{lemma}
\begin{proof}
    By applying flavor switching for quantum commitments (\Cref{prop:flavor-switching}) to $C$ to obtain a commitment $\hat{C}$ that is $\sqrt{\delta}$-statistical hiding and $\eps$-computational binding. The corresponding commitment states have the following form:
    \[
        \ket{\hat{\psi}_{\lambda,b}} = \frac{1}{\sqrt{2}} \Big( \ket{0} \ot \ket{\psi_{\lambda,0}} + (-1)^b \ket{1} \ot \ket{\psi_{\lambda,1}} \Big)
    \]
    where the commitment register of $\ket{\hat{\psi}_{\lambda,b}}$ is the reveal register of $\ket{\psi_{\lambda,b}}$ and the reveal register is everything else (including the extra ancilla qubit). 

    Now we take the parallel repetition of $\hat{C}$ to get $\hat{C}^{\otimes k}$; the corresponding commitment states have the following form:
    \[
        \ket{\hat{\psi}_{\lambda,b}^{\ot k}} = \frac{1}{\sqrt{2^k}} \sum_{x \in \bits^k} (-1)^{|x|} \ket{x} \ot \ket{\psi_{\lambda,x_1}} \ot \cdots \ot \ket{\psi_{\lambda,x_k}}~.
    \]
    By \Cref{thm:main_theorem_uniform} we get that $\hat{C}^{\ot k}$ is $(k \sqrt{\delta})$-statistical hiding (by hybrid argument) and $(\eps^k + \negl)$-computational binding.

    Now we flavor switch back to get our final commitment, which with the construction of \Cref{prop:flavor-switching} happens to be the XOR repeated commitment $C^{\oplus k}$. We get that $C^{\oplus k}$ is thus a $(\eps^{k/2} + \negl)$-computational hiding and $(k\sqrt{\delta})$-statistical binding commitment. Furthermore, it is easy to check that the corresponding commitment states $\ket{\psi^{\oplus k}_{\lambda,b}}$ have the desired form. This completes the proof.
\end{proof}

We remark that due to the commitment duality \cite{hhan2023hardness}, parallel repetition theorem and XOR lemma for commitments are reducible to each other with very little cost, while classically the reduction to parallel repetition from XOR lemma requires majority \cite{SV08}.
The reason we are able to do this more efficiently quantumly can be understood as we are replacing the Goldreich--Levin part \cite{GL89} in the reduction (or the duality construction) with quantum Goldreich--Levin \cite{AC02}, which is much more efficient.

Another consequence is polarization for EFI pairs \cite{brakerski2022computational}: pairs of efficient mixed states that are $\beta$-weakly statistically distinguishable but $\alpha$-weakly computationally indistinguishable where $\beta^2 \gg \sqrt\alpha$, for example, when $\alpha = \frac14$ and $\beta = \frac34$.

\begin{corollary}[EFI polarization]
    \label{cor:efi-polarization}
    If there exist weak EFI pairs that are $\beta$-weakly statistically distinguishable but $\alpha$-weakly computationally indistinguishable such that $\beta^2 - \sqrt\alpha$ is at least constant, then standard EFI pairs exist.
\end{corollary}
\begin{proof}
    In fact a stronger statement is true: assuming we have computational XOR lemma that amplifies to $\eps^{ck} + \negl$ for constant $\frac12 \le c \le 1$ and $\beta^2 - \alpha^c$ is at least constant, then EFI pairs exist.
    Therefore, a tighter commitment duality (without the square root loss) would give better polarization parameters as well.
    To prove this,  Watrous' construction and proof of $\mathsf{QSZK}$ polarization \cite[Theorem 1]{Wat02-szk} immediately adapts to this setting, except that for we use $\alpha^c$ instead of $\alpha$ when picking the number of repetition.

    For completeness, we reproduce the proof here.
    \newcommand{\cd}{\mathrm{cd}}
    We write $\cd(\rho_0, \rho_1) \le \alpha$ to mean that states $\rho_0$ and $\rho_1$ are $\alpha$-computationally indistinguishable.
    Let $r = \lceil\log(8\lambda)/\log(\beta^2/\alpha^c)\rceil = O(\log \lambda)$ and $s = \lfloor\alpha^{-cr}/2\rfloor = \lambda^{O(1)}$.
    We first apply $r$-fold XOR to these states giving $\rho_0', \rho_1'$, then we have that $\cd(\rho_0', \rho_1') \le \alpha^{cr} + \negl$ and $\td(\rho_0', \rho_1') \ge \beta^r$ \cite[Lemma 2]{Wat02-szk}.
    Next we apply $s$-fold parallel repetition to them giving $\rho_0'', \rho_1''$, and obtain that $\cd(\rho_0'', \rho_1'') \le \frac12 + \negl$ by hybrid argument, and $\td(\rho_0'', \rho_1'') \ge 1 - e^{1 - 2\lambda}$ following the same computation as Watrous.
    Applying $\lambda$-fold XOR again completes the construction.
\end{proof}

\paragraph{XOR lemma for quantum predicates.}
Yao's classical XOR lemma states that taking the XOR of many copies of a Boolean predicate amplifies average-case hardness of prediction.
We can similarly use our XOR lemma to amplify unpredictability of \emph{quantum} predicates.

An (average-case) quantum predicate is defined to be a Hermitian matrix $\rho$ with trace $0$ and Schatten $1$-norm at most $2$.
By Jordan--Hahn decomposition, this gives a bijection to the YES instances $\rho_+$ and the NO instances $\rho_-$ such that $\rho = \rho_+ - \rho_-$, with trace $\Tr(\rho_{+}) = \Tr(\rho_{-}) \leq 1$ and orthogonal support $\rho_{+}\rho_{-} = 0$ (so that the unpredictability is not caused by the input being a superposition over YES and NO).
(A worst-case extension of this is given by the promise Boolean observable $\sgn\rho$.)
We say the predicate $\rho$ is $\epsilon$-unpredictable if for all efficient observables $0 \preccurlyeq P \preccurlyeq \id$ with advice $\sigma$, its advantage $\Tr(P(\rho \otimes \sigma)) = \Tr(P(\rho_+ \otimes \sigma)) - \Tr(P(\rho_- \otimes \sigma)) \le \epsilon$.

The advantage of this notation (other than capturing everything with a single matrix) is that the $k$-fold XOR of $\rho$ is simply $\rho^{\otimes k}$, since for any two matrices $\rho_0, \rho_1$, $(\rho_0 - \rho_1)^{\otimes k} = \sum_{x \in \{0, 1\}^k} (-1)^{|x|} \rho_{x_1} \otimes \cdots \otimes \rho_{x_k}$.
The following corollary gives an XOR lemma for quantum predicate indistinguishability.

\begin{corollary}[Quantum Yao's XOR lemma]
    Let $\rho$ be an $\epsilon$-unpredictable predicate.
    Then $\rho^{\otimes k}$ is $(\epsilon^{k/2} + \negl)$-unpredictable.
\end{corollary}
\begin{proof}
There is a bijection between every quantum state distinguishing problem $(\rho_0, \rho_1)$ and an average-case problem of implementing a quantum predicate $\rho = \rho_0 - \rho_1$, where in the backwards mapping, $\rho_0, \rho_1$ can be taken to be $\rho_+ + (1 - \Tr(\rho_+)) \cdot \id, \rho_- + (1 - \Tr(\rho_+)) \cdot \id$ respectively.
Furthermore, the distinguishing advantage of $P$ against $\rho_0$ vs $\rho_1$ is exactly $\Tr(P(\rho_0 \otimes \sigma)) - \Tr(P(\rho_1 \otimes \sigma)) = \Tr(P(\rho \otimes \sigma))$ for any $P, \sigma$.
We complete the proof by invoking \Cref{lemma:xor-commitment}.
\end{proof}

\subsection{Security amplification for public-key quantum money}
\label{subsec:money}

Another direct corollary of our main result (\Cref{thm:main_theorem_uniform}) is security amplification for \emph{public-key quantum money} schemes. We first define public-key quantum money.   

\paragraph{Public-key quantum money.} 
Public-key quantum money, informally, is a scheme in which a trusted bank can efficiently generate an unlimited number of quantum banknotes, everyone can verify a valid banknote, and no efficient adversary can produce counterfeit bank notes with non-negligible success probability.  Formally, 
\begin{definition}[Public-key quantum money]
A public-key quantum money scheme is a triple of efficient quantum algorithms $\cal{S} = (\mathsf{KeyGen},\mathsf{Mint},\mathsf{Ver})$ where 
\begin{itemize}
    \item $\mathsf{KeyGen}$ takes as input the security parameter $1^\lambda$ and outputs a private/public key pair $(k_{\mathrm{private}},k_{\mathrm{public}})$, 
    \item $\mathsf{Mint}(k_{\mathrm{private}})$ outputs a pair $(s,\rho_{\$})$ where $s$ is a string representing a serial number and $\rho_{\$}$ is a quantum state representing a bank note, and
    \item $\mathsf{Ver}$ takes as input the public key $k_{\mathrm{public}}$, a serial number $s$, and an alleged banknote $\sigma$, and either accepts or rejects.
\end{itemize}
A public-key quantum money scheme $\cal{S}$ satisfies correctness if for all $\lambda$, 
\[
    \Pr \left [ \mathsf{Ver}(k_{\mathrm{public}}, s, \rho_{\$}) \text{ accepts} : \begin{array}{c}
(k_{\mathrm{private}},k_{\mathrm{public}}) \leftarrow  \mathsf{KeyGen}(1^\lambda) \\
(s,\rho_{\$}) \leftarrow \mathsf{Mint}(k_{\mathrm{private}}) 
\end{array}
   \right ] \geq 1 - \negl(\lambda)\,.
\]
The scheme $\cal{S}$ is \emph{$\eps$-secure} if for all efficient adversaries $A$, the success probability of $A$ in the counterfeit security game (\Cref{prot:quantum_money_security}) is at most $\eps(\lambda)$.
\end{definition}

\begin{longfbox}[breakable=false, padding=1em, margin-top=1em, margin-bottom=1em]
\begin{algorithm}{\bf Counterfeit security game challenger}\label{prot:quantum_money_security}
\end{algorithm}
\begin{enumerate}
    \item Generate $(k_{\mathrm{private}},k_{\mathrm{public}}) \leftarrow \KeyGen(1^{\lambda})$, $(s, \rho_{\$})\leftarrow \Mint(k_{\mathrm{private}})$ and send $(k_{\mathrm{public}}, s, \rho_{\$})$.
    \item[] [Adversary returns two registers $\reg{AB}$ in some entangled state $\sigma_{\reg{AB}}$.]
    \item Run $\Ver(k_{\mathrm{public}}, s, \sigma_{\reg{A}})$ and $\Ver(k_{\mathrm{public}}, s, \sigma_{\reg{B}})$.  If either reject, reject, otherwise accept.
\end{enumerate}
\end{longfbox}

\paragraph{More money, fewer problems.} Aaronson and Christiano~\cite{aaronson2012quantum} raised the question of whether there exists a general security amplification procedure for public-key quantum money schemes. They exhibited a procedure based on amplitude amplification that (for a class of schemes where the verification is a rank-1 projection) reduces the soundness error of $\Ver$ from a constant to negligible.
However, we note that this is different than reducing the soundness error of the counterfeit security game.  Their construction is only amplifies the success probability of the counterfeiter \emph{per serial number}, in the sense that if the original $\Ver$ would have rejected an alleged banknote with constant probability, the new $\Ver$ will accept the same alleged banknote with negligible probability.
As a consequence, it would not amplify a weakly secure scheme where the proposed quantum money is unclonable for half of the serial numbers but trivially clonable for the other half.

We show that the parallel repetition of a quantum money scheme indeed achieves security amplification for \emph{all} weak quantum money schemes.
More precisely, let $\cal{S}^{n}$ denote the $n$-fold parallel repetition of the quantum money scheme $\cal{S}$, with algorithms $(\mathsf{KeyGen}^n,\mathsf{Mint}^n,\mathsf{Ver}^n)$. The algorithm $\mathsf{KeyGen}^n$ algorithm runs $n$ independent instances of $\mathsf{KeyGen}$ to get pairs $(k_{\mathrm{private}}^{(1)},k_{\mathrm{public}}^{(1)}),\ldots,(k_{\mathrm{private}}^{(n)},k_{\mathrm{public}}^{(n)})$, and treats the tuple of individual private keys (resp. public keys) as a large private key (resp. public key) for $\cal{S}^n$. The algorithm $\mathsf{Mint}^n$ takes as input the tuple of private key and runs $\mathsf{Mint}(k_{\mathrm{private}}^{(i)})$ for $i = 1,\ldots,n$ to obtain  $n$ pairs $(s^{(1)},\rho_{\$}^{(1)}),\ldots,(s^{(n)},\rho_{\$}^{(n)})$; the output serial number is the tuple $(s^{(1)},\ldots,s^{(n)})$ and the output bank note is the concatenation $(\rho_{\$}^{(1)},\ldots,\rho_{\$}^{(n)})$. Finally the algorithm $\mathsf{Ver}^n$ will simply run $\mathsf{Ver}(k_{\mathrm{public}}^{(i)},s^{(i)},\rho_{\$}^{(i)})$ for $i=1,\ldots,n$ and accept if they all accept. 

It is clear that if the base money scheme $\cal{S}$ satisfies correctness, so does the repeated scheme $\cal{S}^n$ (provided that the number of repetitions $n$ is polynomial in the security parameter $\lambda$). Next, the security game for the repeated scheme $\cal{S}^n$ is the $n$-fold parallel repetition of the security game for the base scheme $\cal{S}$. Since the security game is a $3$-message quantum interactive protocol, \Cref{thm:main_theorem_uniform} directly implies the following:

\begin{corollary}[Security amplification for public-key quantum money schemes]
\label{cor:money-amp}
    Let $\cal{S}$ be an $\eps$-secure public-key quantum money scheme. Then $\cal{S}^n$ is an $(\eps^n + \negl)$-secure public-key quantum money scheme.
\end{corollary}

We note that the same argument essentially generalizes to any quantum cryptographic primitive with a security game consisting of at most three messages.
So far we have already considered two examples (commitments and money) where the security game has two messages.
Another notable example is \emph{quantum lightning}~\cite{zhandry2021quantum}, a potentially stronger primitive than quantum money where the bank may not be able to produce two copies of the same banknote, also admits exponential security amplification via parallel repetition.
This is because the security game for quantum lightning is also a two-message quantum interactive protocol. 

\subsection{3-message quantum zero knowledge}
\label{sec:zk}

Weakly sound protocols also naturally occur in the context of zero knowledge protocols.
An important template for constructing zero knowledge protocols is the $\Sigma$-protocol.
At a high level, a $\Sigma$-protocol is a three-message protocol with inverse polynomial soundness (bounded away from 1), zero knowledge, and where the verifier's message is uniformly random.
$\Sigma$-protocols where the prover and verifier are quantum are sometimes called $\Xi$-protocols~\cite{broadbent2022qma}).
\begin{definition}[Quantum $\Sigma$-protocol]
    A quantum $\Sigma$-protocol for a language $L$ is a public coin with $O(\log|x|)$ random bits of challenge, quantum interactive proof system the following additional property called \emph{computational zero knowledge} defined as follows.
    There exists a polynomial time quantum simulator $\Sim$ that takes as input a string $x$ and randomness $r$, and outputs a pair of quantum states $(\rho, \sigma)$ such that
    \begin{equation*}
        \{(\rho, \sigma) : (\rho, \sigma) \leftarrow \Sim(x, r)\}_{r} \approx_{c} \{(\rho, \sigma) : (\rho, \sigma) \leftarrow V(x) \interact_{r} P(x, \ket{\psi_{x}})\}_{r}
    \end{equation*}
    where $(\rho, \sigma) \leftarrow V(x) \interact_{r} P(x, \ket{\psi_{x}})$ denotes the mixed state of the message registers that the honest prover $P$ sends, conditioned on the challenger's randomness being $r$.  
\end{definition}
Note that in the context of classical and post-quantum zero knowledge proofs, $\Sigma$-protocol in addition requires special soundness, which we intentionally omit here.
Furthermore, below we also slightly abuse this notation to also include discussions of statistical zero knowledge \emph{arguments} that follow this template.

While a $\Sigma$-protocol for all languages in $\class{QMA}$ is already known \cite{broadbent2022qma} and can be constructed from minimal assumptions \cite{brakerski2022computational}, the round collapse theorem also allows us to construct these differently by unconditionally compiling any honest-verifier zero knowledge protocol into a $\Sigma$-protocol.
As before, using our round collapse theorem can be advantageous (over building one from scratch by e.g.\ invoking \cite{broadbent2022qma}) for preserving certain properties of the original protocol like succinctness.

\begin{corollary}
    For any language $L$ that admits an $m$-message honest-verifier quantum statistical (resp.\ computational) zero knowledge protocol and computational (resp.\ statistical) soundness, $L$ also admits a malicious-verifier quantum statistical (resp.\ computational) zero knowledge protocol with 3 messages and computational (resp.\ statistical) soundness $1 - 1/\poly$.
    Furthermore, the verifier's message is a coin flip, and the verifier and communication complexity only blows up by $m^{O(1)}$.
\end{corollary}
\begin{proof}[Proof sketch]
    The statement for computational zero knowledge proofs is already proven by \cite[Theorem 34]{Kobayashi08}, and generalizing this to arguments is straightforward with \Cref{thm:compression,thm:public_coin}.
    For completeness, we sketch the proof strategy here.

    The first step of the construction is to apply the round collapse theorem (\Cref{thm:compression}) to the original protocol.
    We first show honest-verifier zero knowledge is preserved at each step in the iteration of \Cref{thm:compression}, and thus the final protocol is honest-verifier zero knowledge.
    This can be seen as the honest verifier's view at round $i$ (out of $r + 1$ rounds) is simply a random bit $b$ (except for the first round where $b$ is uninitialized) along with the precompiled verifier's view at round $r + 1 + (-1)^b \cdot (i - 1)$.

    In order to have malicious-verifier zero knowledge, we need to compile this protocol into a $\Sigma$-protocol where the verifier's second message is a single coin flip.  To do this, we apply \Cref{thm:public_coin}.
    Honest verifier zero knowledge is still preserved since the view after round 1 is identical to the original view at round 2, and the view after round 3 is a random bit $b$ along with the original view at round $2 + (-1)^b$.
    
    Finally, to show malicious-verifier zero knowledge, we can use the honest-verifier zero knowledge simulator to simulate the honest transcript and then apply Watrous rewinding \cite{watrous2006zero} to the malicious verifier.
\end{proof}

Thus we have round-collapsed the protocol while preserving relevant properties except soundness error.
To get back the same level of soundness, the natural approach is to apply $m^{O(1)}$-fold parallel repetition.
While the protocol after parallel repetition is unlikely to be zero knowledge \cite{holmgren2021fiat}, witness indistinguishability (as well as honest-verifier zero knowledge) is preserved via a standard hybrid argument.
We thus arrive at the following.
\begin{corollary}
	\label{corr:wi}
    For any language $L$ that admits an $m$-message honest-verifier quantum statistical (resp.\ computational) zero knowledge protocol and computational (resp.\ statistical) soundness, $L$ also admits a malicious-verifier statistical (resp.\ computational) witness indistinguishable protocol with 3 messages and negligible computational (resp.\ statistical) soundness.
    Furthermore, the verifier is public coin, and the verifier and communication complexity only increases by a multiplicative $m^{O(1)}$ factor.
\end{corollary}

We leave as future work to improve this to a stronger security, which is beyond the scope of this work.
For starters, we conjecture that this class of 3-message protocols is already witness hiding even for instances with unique witness \cite{DSYC18-unique}.
Another direction is to consider recovering zero knowledge by further augmenting the protocol like classically.
For example, we expect one could combine \Cref{corr:wi} with \cite{Yan23} using (instance-dependent) commitments to compile a computational zero knowledge proof down to 4 messages (with coherent expected-QPT simulation \cite{lombardi2022post}%
\footnote{
  It is known that constant-round post-quantum zero knowledge cannot be achieved with ``standard'' expected-QPT simulators \cite{chia2021on}.
  We leave as an open problem to investigate whether this impossibility generalizes to quantum protocols.
}).
It is plausible that further ideas could even yield a constant-round transform for statistical zero knowledge arguments.

\subsection{Simpler construction of commitments from undecodable black holes}
\label{sec:blackholes}

Finally, we give a more intuitive one-page proof of a theorem from \cite{brakerski2023black}.
The work by Brakerski~\cite{brakerski2023black} showed that the existence of EFI/commitments is equivalent to the existential hardness of black-hole radiation decoding, as formulated by Hayden and Harlow~\cite{HarlowH13-firewall}.

\paragraph{Black hole radiation decoding.} First, we recall the setting.
A radiation state is an efficiently preparable quantum state $\ket\psi_{\reg{HBR}}$.
We are promised that is that it is statistically possible to extract a qubit from $\reg R$ that is maximally entangled with $\reg B$, but computationally hard to do so with success probability non-negligibly greater than some threshold $\delta \ll 1$.
Note that $\delta \ge \frac14$ always since that is the trivial success probability that can always be achieved by outputting a constant 0 qubit.

\paragraph{Commitments from undecodable states, simplified.}
\cite{brakerski2023black} constructs an undecodable radiation state from a commitment in a straightforward way: consider the overall pure state where we simply use a statistically-binding quantum state commitment (like the folk-lore construction \cite{gunn2023commitments}) to commit to a half of an EPR pair; we assign the other half of EPR being the qubit $\reg B$, the commitment register being the prior radiation $\reg R$, and the reveal register being the rest of the black hole $\reg H$.
However, the construction of EFI from undecodable states is much more involved.
Their construction on a high level is constructing a canonical-form commitment that is weakly statistically binding but strongly computationally hiding.
In order to argue computational hiding, they have to develop additional tools like superdense decoding.

We give, from undecodable radiation states, an alternative construction of commitments that is instead weakly computationally binding but strongly statistically hiding with a tighter analysis (since our parallel repetition theorem is tight); furthermore, it can be viewed as the dual construction for the converse direction.  The construction of the bit commitment is as follows:
\begin{itemize}
    \item Prepare the radiation state along with the EPR pair $\ket\psi_{\reg{HBR}} \otimes \ket{\Psi^+}_{\reg{CD}}$.
    \item To commit to 0, send registers $\reg{HB}$.
    \item To commit to 1, send registers $\reg{HD}$.
\end{itemize}
This is statistically hiding, since by statistical decodability and monogamy of entanglement, after tracing out $\reg R$, the state on $\reg{HB}$ is close to a product state where $\reg B$ is maximally mixed.
On the other hand, the task of breaking honest computational binding (from 0 to 1) is exactly on input register $\reg R$ and an (unrelated) EPR pair $\reg{CD}$, output a state such that the overall state looks like $\reg{HD} \otimes \reg{RCB}$.
In other words, this ``teleports'' out the entanglement in $\reg H$ to $\reg B$ into $\reg C$, and now the register $\reg H$ is maximally entangled with $\reg D$ instead.
In particular, the output $\reg C$ register would be correctly entangled with the test qubit, and thus we conclude that this is computationally $\delta$-binding.
Combining this with the computational amplification (given that $\delta$ is inverse-polynomially bounded away from 1), we obtain a construction of EFI via parallel repetition and commitment flavor switching.

\fi

\ifacmart
\else
\printbibliography
\fi

\end{document}